\documentclass[twocolumn,prx,nofootinbib,superscriptaddress]{revtex4-2}
\usepackage{algorithm2e}

\usepackage{amssymb,amsxtra,amsmath,amsfonts, amsthm}
\usepackage{enumerate}
\usepackage[toc,page]{appendix}
\usepackage[caption=false]{subfig}

\def\aa{\mathfrak{a}}
\def\bb{\mathfrak{b}}
\def\cc{\mathfrak{c}}
\def\g{\mathfrak{g}}
\def\uu{\mathfrak{u}}
\def\s{\mathfrak{s}}
\def\su{\mathfrak{su}}
\def\so{\mathfrak{so}}
\def\sp{\mathfrak{sp}}
\def\sl{\mathfrak{sl}}
\def\gl{\mathfrak{gl}}
\def\mcA{\mathcal{A}}
\def\mcB{\mathcal{B}}
\def\mcL{\mathcal{L}}
\def\mcP{\mathcal{P}}
\def\Tr{\mathrm{Tr}}
\def\UU{\mathrm{U}}
\def\SU{\mathrm{SU}}
\def\SO{\mathrm{SO}}

\def\ad{\mathrm{ad}}
\def\Span{\mathrm{span}}
\def\Stab{\mathrm{St}}

\newcommand\Lie[1]{\langle #1 \rangle_{\mathrm{Lie}}}  
\newcommand{\btheta}{\boldsymbol{\theta}}
\newcommand{\bphi}{\boldsymbol{\phi}}

\newtheorem{definition}{Definition}[section]
\newtheorem{result}{Result}[section]
\newtheorem{theorem}{Theorem}[section]
\newtheorem{lemma}{Lemma}[section]
\newtheorem{proposition}{Proposition}[section]
\newtheorem{remark}{Remark}[section]
\newtheorem{corollary}{Corollary}[section]
\theoremstyle{definition}
\newtheorem{example}{Example}[section]

\usepackage[english]{babel}

\usepackage{amsmath}
\usepackage{graphicx}
\usepackage{lipsum}
\usepackage[colorlinks=true, allcolors=blue]{hyperref}
\usepackage[all]{hypcap}

\usepackage{enumitem}
\newlist{abbrv}{itemize}{1}
\setlist[abbrv,1]{label=,labelwidth=1.1in,align=parleft,itemsep=0.1\baselineskip,leftmargin=!}

\usepackage[dvipsnames]{xcolor}

\begin{document}

\title{Classification of dynamical Lie algebras for translation-invariant 2-local spin systems in one dimension}

\author{Roeland Wiersema}
\affiliation{Vector Institute, MaRS  Centre,  Toronto,  Ontario,  M5G  1M1,  Canada}
\affiliation{Department of Physics and Astronomy, University of Waterloo, Ontario, N2L 3G1, Canada}
\affiliation{Xanadu, Toronto, ON, M5G 2C8, Canada}

\author{Efekan K\"okc\"u}
\affiliation{Department of Physics, North Carolina State University, Raleigh, North Carolina 27695, USA}

\author{Alexander F. Kemper}
\affiliation{Department of Physics, North Carolina State University, Raleigh, North Carolina 27695, USA}

\author{Bojko N. Bakalov}
\affiliation{Department of Mathematics, North Carolina State University, Raleigh, North Carolina 27695, USA}

\begin{abstract}

Much is understood about $1$-dimensional spin chains in terms of entanglement properties, physical phases, and integrability. However, the Lie algebraic properties of the Hamiltonians describing these systems remain largely unexplored. In this work, we provide a classification of all Lie algebras generated by translation-invariant $2$-local spin chain Hamiltonians, or so-called dynamical Lie algebras. We consider chains with open and periodic boundary conditions and find $17$ unique dynamical Lie algebras. Our classification covers some well-known models such as the transverse-field Ising model and the Heisenberg chain, and we also find more exotic classes of Hamiltonians that cannot be identified easily. In addition to the closed and open spin chains, we consider systems with a fully connected topology, which may be relevant for quantum machine learning approaches. We discuss the practical implications of our work in the context of quantum control, variational quantum computing, and the spin chain literature. 

\end{abstract}

\date{September 11, 2023}

\maketitle

\section{Introduction}
Mathematical classifications of the fundamental symmetries of physical systems date back to the work of Wigner, who proposed three symmetry classes of non-interacting fermionic Hamiltonians depending on their time-reversal and spin-rotation properties \cite{wigner1959grouptheory}. Three decades later Dyson would mathematically solidify this theory and connect the spectral properties of these different types of Hamiltonians with random matrix theory \cite{dyson1962threefoldway} (see \cite{edelman2022cartan} for a modern treatment).
It would take another thirty years before Altland and Zinbauer extended these results to ten symmetry classes \cite{altland1997nonstandard}, each of which correspond to a symmetric space in Cartan's original classification of these spaces \cite{cartan1926classe, helgason1979differential}. Further extensions of these results were made in recent years with regards to topological phases of matter \cite{schnyder2008classification, ryu2010topological, barkeshli2013classification}.

The above mathematical classifications of quantum physics rest on the 
powerful theory of Lie groups,
which provides a framework for describing the continuous symmetries and transformations that characterize the behavior of quantum systems. The study of Lie groups, and by extension physical symmetries, can often be simplified by considering the corresponding Lie algebra of the group. The commutation relations of operators in the Lie algebra capture the essential features of the underlying symmetries and can be used to analyze the spectrum, eigenstates, and dynamics of quantum systems.

A Hamiltonian of a finite-dimensional system can be understood as ($i$ times) an element of some Lie algebra $\mathfrak{g}\subseteq\su(N)$. Here, $\su(N)$ is the vector space of all skew-Hermitian $N\times N$ matrices equipped with the standard commutator. 
Typically, a Hamiltonian is described by a linear combination of terms that correspond to a certain physical interaction. These individual terms can be used to generate a Lie algebra, which is called the Hamiltonian algebra or \emph{dynamical Lie algebra} (DLA) \cite{Albertini2001dynlie, albertini2021subspace,chen2017preparing,wang2016subspace,dalessandro2021dynamical}, which is intricately linked to the dynamics of quantum systems.

Since each DLA is a subalgebra of $\su(N)$, a classification of DLAs can be phrased as a classification of all subalgebras of $\su(N)$. Such a classification is intractable, except when specific constraints are placed on the subalgebras one considers. For example, in the original works of Killing and Cartan, all {\emph{simple}} Lie algebras were classified \cite{cartan1914lesgroupes};  similarly, Dynkin provided a classification of the {\emph{maximal}} subalgebras of simple Lie algebras \cite{dynkin1952maximal,tits1959}. We follow a different approach: instead of adding algebraic constraints such as simplicity or maximality, we make use of the fact that any Lie algebra can be described by a set of generators, and we consider the Lie algebras that arise by using the terms of specific Hamiltonians as
the generators. In contrast with the previously-mentioned classifications of \cite{wigner1959grouptheory, dyson1962threefoldway, altland1997nonstandard}, this approach contains interacting quantum many-body systems.

Specifically, we consider the class of Hamiltonians that correspond to $1$-dimensional $2$-local spin chains, and provide a classification of the
Lie algebras that are generated under commutation by the terms of the Hamiltonian. Much about these systems is well-understood, from their entanglement properties \cite{hastings2007arealaw}, their phases \cite{chen2011classification} and their integrability \cite{de2019classifying,jones2022integrable}. However, to the best of our knowledge, the Lie algebraic properties of these Hamiltonians have not yet been explored in full. It is thus reasonable to ask, given our comprehensive knowledge of the physics governing these systems, what more can be learned from the Lie algebra?
In short, our classification has bearing on areas of quantum control, variational quantum computing, and quantum dynamics and thermodynamics.

In \emph{quantum control}, the DLA of a dynamical quantum system can be related to the set of reachable states of that system. In particular, DLAs can be used to define a notion of controllability of a quantum system~\cite{schirmer2001complete, schirmer2002identification,wang2012symmetry}, which is highly relevant when it comes to designing unitary operations for quantum simulators and quantum computers. One is typically interested in Hamiltonians that can generate an arbitrary unitary operator, while the existence of symmetries can inhibit the control of a physical system \cite{zeier2011symmetry}.


For \emph{variational quantum computing}, one is not interested in representing the whole unitary group, but in using a parameterized subgroup in order to generate a state that maximizes a given objective function. Understanding what subgroup a particular quantum circuit parameterizes can give insight into its representational power. For example, one can connect the dimension of the DLA to a phenomenon called overparameterization \cite{kiani2020learning, wiersema2020exploring, you2022convergence, larocca2023theory}. Additionally, DLAs can be used to understand barren plateaus \cite{Larocca2022diagnosing} --- flat areas in the cost landscape of a variational quantum algorithm that hinder optimization \cite{mcclean2018barren, cerezo2021cost}. Finally, a recent work uses knowledge of the DLA to perform efficient classical simulations of several quantum algorithms \cite{goh2023lie}.

Finally, one can use the knowledge of the DLA to 
provide insights into the \emph{dynamics and thermodynamics}.
For example, one can
construct constant-depth quantum circuits for the dynamical simulation of a specific quantum system \cite{Kokcu2021cartan,kokcu2022algebraic,camps2022algebraic, kokcu2023algebraic}, 
or state preparation via adiabatic state preparation \cite{kokcu2022algebraic,kokcu2023algebraic},
or 
implement Hartree--Fock \cite{google2020hartree}. The dimension of the DLA is directly related to the quantum circuit depth needed to capture the full dynamics \cite{Kokcu2021cartan}.
Thermodynamic properties are in part encoded in the commutant of the
DLA, i.e., in the conserved charges \cite{majidy2023noncommuting}. Non-Abelian commutants lead
to non-trivial quantum effects in thermodynamics, which affects thermalization properties among others.

The paper is structured as follows. We end the introduction with a summary of our main mathematical results. Then, in Sect.\ \ref{sec:background}, we establish our notation and introduce the necessary mathematical preliminaries.
We discuss the method of our classification in Sect.\ \ref{sec:method} and present the main results in Sect.\ \ref{sec:results}. Finally, we discuss the implications of our results in Sect.\ \ref{sec:discussion}. In the Supplemental Materials (abbreviated SM), we first review preliminaries on Pauli strings and Lie algebras and then present the full details of the proofs of the main results.

\subsection{Summary of the main results}
Here, we give a brief summary of our main results, which include the classification of all DLAs generated by 2-local spin Hamiltonians of length $n$ in one dimension. Recall that a Lie algebra can be constructed by a set of generators so that it is closed under linear combinations and under the Lie bracket. In our case, the Lie bracket is the standard commutator $[A,B] = AB - BA$. We now choose the generators of our Lie algebra to be ($i$ times) the terms of any 2-local spin chain Hamiltonian. Since a Hamiltonian is always a Hermitian operator, we can understand it as ($i$ times) an element of the Lie algebra $\uu(2^n)$. Therefore, we can limit ourselves to the study of DLAs that are subalgebras of $\uu(2^n)$, for which we have the following useful fact.
 Although this result is known (see e.g.\ \cite{zeier2011symmetry, dalessandro2021}), for completeness, we provide its proof and review the necessary definitions in SM \ref{secpauli2}.

\begin{proposition}\label{rem:direct}
    A DLA must be either Abelian, $\su(N')$, $\so(N')$, $\sp(N'')$ $($with $N'\le 2^n$, $N''\le 2^{n-1})$, an exceptional compact simple Lie algebra, or a direct sum of such Lie algebras. 
    Indeed, any subalgebra of $\uu(N)$ is either Abelian or a direct sum of compact simple Lie algebras and a center. 
\end{proposition}

Note that all simple Lie algebras over the complex and real numbers have been classified by Killing and Cartan \cite{hall2015}. The above proposition forms the backbone of our classification, as we know that any DLA generated by our class of Hamiltonians must be of the described form. 

To obtain the classification, we first calculate all DLAs of the power set of Hamiltonians for a $2$-site system, and identify the unique sets of generators. Then we identify the orbits under the symmetries of the Pauli group and the swap of the two sites, thus reducing the number of unique Lie algebras to $27$. Next, we find several isomorphisms between some of the sets of generators, reducing the set of unique structures even further. Finally, we determine how these Lie algebras scale with system size as the number of spins grows beyond $2$ sites. In this final step, we take the topology of the spin chain into account, since the Lie algebra will behave differently for open or periodic boundary conditions of the chain. 
The following is our main result.

\begin{result}[Classification of spin chain DLAs]
We provide a classification of all dynamical Lie algebras of $2$-local spin Hamiltonians in one dimension. For both open and closed spin chains, there are $17$ unique Lie algebras that can be generated by a spin chain Hamiltonian.
\end{result}
The formal statement of this result is presented in the main text with Theorems \ref{the:classification} and \ref{the:classification-p} along with a sketch of the proof. The dimension of a DLA can be related to the trainability of variational quantum circuits, and may therefore be of high interest. Since we know the dimensions of all simple Lie algebras, a direct corollary of our result is the following.

\begin{result}[Dimension scaling of DLAs]\label{res:dim_dlas}
    The dimension of any dynamical Lie algebra of a $2$-local spin chain Hamiltonian of length $n$ will scale as either $O(4^n)$, $O(n^2)$ or $O(n)$.
\end{result}

To illustrate this, we plot the dimensions of the open DLAs in our classification in Figure \ref{fig:annotated}.
\begin{figure}[htpb]
    \centering
    \includegraphics[width=\columnwidth]{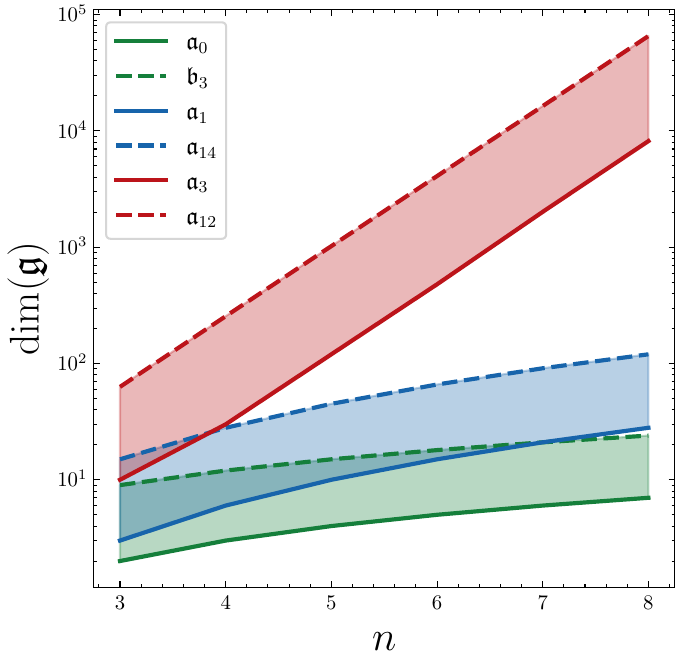}
    \caption{Scaling of the DLAs of spin chains with open boundary conditions. The exponentially scaling DLAs are denoted in red, the quadratically scaling ones in blue, and the linearly scaling algebras are denoted in green. The full and dashed line denote the smallest and largest scaling algebra in our classification, respectively.}
    \label{fig:annotated}
\end{figure}

Our proof technique also applies to the case of a permutation-invariant graph, where each site is interacting with every other site via at most 2-local interactions; in other words, all-to-all connected. We therefore also include this topology in our classification. 
\begin{result}[Classification of permutation invariant DLAs]
    We provide a classification of all dynamical Lie algebras of $2$-local permutation-invariant spin Hamiltonians. There are $8$ unique Lie algebras that can be generated by such a spin chain Hamiltonian.
\end{result}
We present the formal statement of this result in Theorem~\ref{the:classification-s}. Similarly to Result~\ref{res:dim_dlas}, we find DLAs with linear, quadratic and exponentially scaling dimensions.

In addition to the classification of unique Lie algebraic structures, we also provide an explicit list of Hamiltonians that generate them in Table \ref{tab:generator_list}. Some of the models in our classification correspond to well-known models such as the transverse-field Ising model or the Heisenberg model, whereas other Hamiltonians are perhaps not realizable in nature. However, some of these more exotic models may be of interest due to their properties. For instance, we find a class of Hamiltonians with globally non-commuting charges, which are known to be of interest in quantum thermodynamics \cite{majidy2023noncommuting, yunger2022build}.

\section{Background\label{sec:background}}

\subsection{Glossary}

We summarize the main choices of notation and nomenclature used throughout the paper in the following list.

\begin{abbrv}
\item[$\mathfrak{a}$, $\mathfrak{g}$] Lie algebras.
%
%
\item[$\mcA$, $\mathcal{G}$] Sets, usually generating sets for Lie algebras.
\item[$\Lie{\mathcal{G}}$] The Lie algebra generated by a set $\mathcal{G}$.
\item[$\mathfrak{a}\cong\mathfrak{b}$] $\mathfrak{a}$ and $\mathfrak{b}$ are isomorphic as Lie algebras.
\item[$\mcA\equiv\mcB$] $\mcA$ and $\mcB$ are equivalent under the symmetry group $S_3 \times \mathbb{Z}_2$ .
\item[$\mcA=\mcB$] $\mcA$ and $\mcB$ are equal as sets.
\item[$\mcP_n$] The set of $n$-qubit Pauli strings $\{I,X,Y,Z\}^{\otimes n}$.
\item[$X_j,Y_j,Z_j$] The action of the corresponding Pauli matrix on the $j$-th qubit in the spin chain.
\item[$\Span_{\mathbb R}(\mcA)$] The set of all linear combinations of elements of $\mcA$ with real coefficients.
\item[$\Span = i\,\Span_{\mathbb R}$] A shortcut for taking the real span and multiplying by the imaginary unit $i$. Takes Hermitian matrices to skew-Hermitian.
\end{abbrv}

\subsection{Preliminaries}
We assume knowledge of finite-dimensional Lie algebras (for a formal treatment, see e.g.\ Refs.~\cite{helgason1979differential, hall2015}), but will review some essential concepts here.
A Lie algebra $\mathfrak{g}$ is a vector space equipped with a Lie bracket $[\cdot,\cdot]\colon\mathfrak{g}\times\mathfrak{g}\to \mathfrak{g}$ satisfying certain axioms (which are reviewed in SM \ref{secpauli2}). The Lie bracket defines the adjoint endomorphism $\ad_a\colon\mathfrak{g}\to\mathfrak{g}$ where $\ad_a(b)= [a,b]$. For our purposes, the Lie bracket is the standard commutator of linear operators on a complex vector space: $[a,b] = ab - ba$. Due to the Lie-correspondence, we can associate a Lie group $G$ with a Lie algebra $\mathfrak{g}$ via the exponential map $G = e^{\mathfrak{g}}$.

Consider a set of generators $\mcA=\{ a_1, a_2,\ldots, a_{M}\}$ with $a_k\in \mathfrak{g}$. We first define the nested commutator, 
\begin{align}\label{eq:adj}
\ad_{a_{i_1}} \cdots \ad_{a_{i_r}} (a_j) = [a_{i_1},[a_{i_2},[\cdots[a_{i_r},a_j]\cdots]]],
\end{align}
which is just $a_j$ in the special case $r=0$.
The linear span of all nested commutators
\begin{align*}
    \Lie{\mcA} := \Span\bigl\{\ad_{a_{i_1}} \cdots \ad_{a_{i_r}} (a_j) \,\big|\, a_{i_1},\dots, a_{i_r},a_j\in \mcA\bigr\}
\end{align*}
is then called a \emph{dynamical Lie algebra} (DLA) \cite{dalessandro2021, Albertini2001dynlie}. This is the minimal (under inclusion) subalgebra of $\mathfrak{g}$ that contains the set $\mcA$.
The depth $r$ of the nested commutator is finite and will depend on the size of the DLA, which we typically do not know beforehand. In practice, the DLA of a given set of generators $\mcA$ can be obtained recursively with Algorithm \ref{alg:dyn_lie}.
\RestyleAlgo{ruled}
\begin{algorithm}[htb!]
\caption{Calculating the DLA}\label{alg:dyn_lie}
\KwIn{Set of generators $ \mcA$}
\For{$a_i \in \mcA$}{
    \For{$a_j \in \mcA$}{
        $a_k = [a_i,a_j]$\\
        \uIf{$a_k \notin \mcA$ \normalfont\textbf{and} $a_k\neq 0 $}{
            $\mcA \gets \mcA\cup \{a_k\}$\\
            }
    }
}
$\Lie{\mcA} \gets \Span\{\mcA\}$
\end{algorithm}

\subsection{Translation-invariant 2-local spin systems in one dimension}

Due to Proposition \ref{rem:direct}, we know what form the subalgebras of $\su(N)$ must take. Our goal is to find which of these direct sums of simple or Abelian Lie algebras can be generated by a physically inspired set of generators. 

In particular, we are interested in the subalgebras of $\su(2^n)$ that are generated by the terms of 
$1$-dimensional
$2$-local Hamiltonians 
with translation-invariant structure, i.e., the type of the interactions between qubits is the same, but interaction strength may vary.
We consider a spin system with a complex Hilbert space $(\mathbb{C}{^2})^{\otimes n}$ and a Hamiltonian $H$, which is a Hermitian operator on $(\mathbb{C}{^2})^{\otimes n}$ of the form
\begin{align}
    H = \sum_{k=1}^{n-1} \sum_{a\otimes a'\in \mcA} a_{k}\otimes a'_{k+1}, \label{eq:hamiltonian_1d}
\end{align}
where 
\begin{align}
    a_{k}\otimes a'_{k+1} &:= I^{\otimes (k-1)}\otimes a \otimes a' \otimes I^{\otimes (n-k-1)}, \label{eq:two_local}
\end{align}
with $a\otimes a'\in\mathcal{A}$ and $I$ is the $2\times 2$ identity matrix.
We consider $a, a'\in \mathcal{P}_1 :=\{I,X,Y,Z\}$ (Pauli matrices), and
one of $a,a'$ should be different from $I$. 

The generating set $\mcA\subseteq\mcP_2:=\mcP_1^{\otimes 2}$ defines a specific set of $2$-local operators that make up the Hamiltonian $H$; in the parlance of quantum computing and physics, this is a $2$-local Hamiltonian corresponding to a \emph{spin chain}. Note that physical models come with coefficients in front of each $a_{k}\otimes a'_{k+1}$ term. The values of these coefficients determine the resulting physics and the corresponding phases of matter. Here, we are only concerned with the algebraic properties of the Hamiltonian $H$ on the Lie algebra level, and we will not consider any spectral properties of~$H$.

Continuing, we note that Pauli matrices $i\mathcal{P}_1\setminus\{iI\}$ form a basis of $\su(2)$, and the tensor products $i\mathcal{P}_2\setminus\{iI^{\otimes2}\}$ form a basis of $\su(4)$. Hence,  $\Span(\mcA) \subseteq \su(4)$ (recall that $\Span=i\,\Span_{\mathbb{R}}$).
In the following, we will suppress the tensor product between Pauli operators and identities for clarity, and we denote $a\otimes a'= a a'$. We now give some examples to illustrate how several well-known spin chains can be written in this notation.

\begin{example}\textit{Transverse-field Ising model (TFIM).}
    The set of generators of the TFIM in one dimension with open boundary conditions is given by
        \begin{align*}
            \mcA_{\mathrm{TFIM}} := \{ZZ, XI\},
        \end{align*}
    which results in the Hamiltonian
    \begin{align*}
        H_{\mathrm{TFIM}} = \sum_{i=1}^{n-1} Z_i Z_{i+1} + X_i.
    \end{align*}
\end{example}
\begin{example}\textit{Heisenberg chain.}
    For the $1$-dimensional Heisenberg chain with open boundary conditions, the generators are given by
    \begin{align*}
        \mcA_{\mathrm{Heis}} := \{XX, YY, ZZ\},
    \end{align*}
    which results in the Hamiltonian
    \begin{align*}
        H_{\mathrm{Heis}} = \sum_{i=1}^{n-1} X_i X_{i+1} + Y_i Y_{i+1} + Z_i Z_{i+1}.
    \end{align*}
\end{example}

\begin{example}\textit{Spinless fermionic Gaussian state.}
    A free fermion Hamiltonian chain in one dimension with periodic boundary conditions can be built from the generators on two sites:
    \begin{align*}
        c_1^\dagger  c_2^\dagger,
        c_1^\dagger  c_1^{\phantom{\dagger}},
        c_2^\dagger  c_2^{\phantom{\dagger}},
        c_1^\dagger  c_2^{\phantom{\dagger}},
        c_2^\dagger  c_1^{\phantom{\dagger}},
        c_1^{\phantom{}}  c_2^{\phantom{\dagger}},
    \end{align*}
    where $c^\dagger$ and $c^{\phantom{\dagger}}$ are fermionic raising and lowering operators, respectively. The corresponding Hamiltonian is
    \begin{align*}
        H_{\mathrm{FF}} = \sum_{i=1}^{n-1} c^\dagger_i c^\dagger_{i+1} + c^\dagger_i c_{i+1} + c_i c^\dagger_{i+1}+c_i c_{i+1}.
    \end{align*}
    The fermionic raising and lowering operators may be translated to Pauli string form via a number of transformations. If we use the common Jordan--Wigner transformation, the resulting set of generators is
    \begin{align*}
        \mcA_{\mathrm{FF}} := \left\lbrace
            X  X, Z  I, I  Z, Y  Y, X Y, Y X\right\rbrace.
    \end{align*}
\end{example}

The terms in the Hamiltonian generate a Lie algebra $\Lie{\mcA} = \mathfrak{a}$ that is a subalgebra of $\su(4)$. We now investigate the structure of these algebras as we add terms that have been translated by one site. Starting from a subalgebra $\mathfrak{a} \subseteq\su(4)$, let $\mathfrak{a}(n)$ be the subalgebra of $\su(2^n)$ generated by the set
\begin{equation*}
    \bigcup_{1\le k\le n-1} I^{\otimes (k-1)} \otimes \mathfrak{a}\otimes I^{\otimes (n-k-1)}.
\end{equation*}

In particular, $\mathfrak{a}(2)=\mathfrak{a}$. By construction, we have a Lie algebra embedding $\mathfrak{a}(n) \hookrightarrow \aa(n+1)$, given by appending $I$ to the last  qubit (see Figure \ref{fig:embed}). 

\begin{example}
    Consider the generating set $\mcA = \{{XY}\}$ on $2$ qubits. The DLA is given by
    \begin{align*}
        \Lie{\mcA} = \Span\{XY\},
    \end{align*}
    which is an Abelian Lie algebra isomorphic to $\uu(1)$. Constructing the generators of $\aa(3)$ according to the procedure above gives
    \begin{align*}
        \aa(3) = \Lie{XYI, IXY} = \Span\{{XYI, IXY, XZY}\}.
    \end{align*}
    It is easy to confirm that $\aa(3) \cong \so(3)$. We see that in going from $n=2$ to $n=3$ we have $\uu(1) \to \so(3)$.
\end{example}
\begin{figure}[htb!]
    \centering
    \includegraphics[width=\columnwidth]{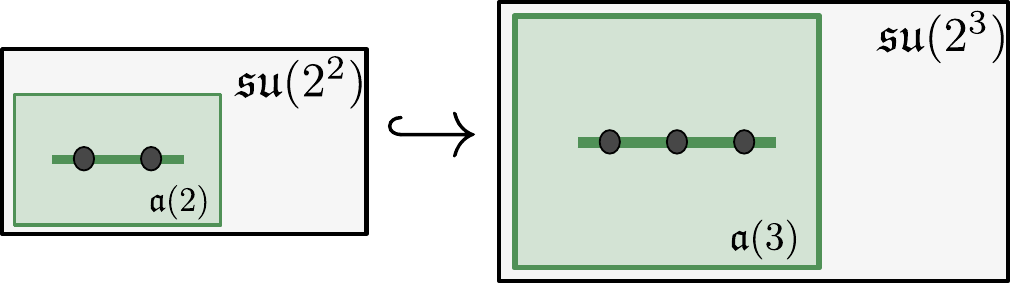}
    \caption{Growing a Lie algebra by adding a site to the chain. }
    \label{fig:embed}
\end{figure}
The above example illustrates that the algebraic structure of a DLA can change as we increase the system size. Additionally, when we extend the number of sites to $n>2$, we need to take into account what happens at the edge of the chain. For $1$-dimensional systems, this leads to two cases: open boundary conditions (operators on a line) and periodic boundary conditions (operators on a circle). We will denote the resulting DLAs of these two cases with $\mathfrak{g}$ and $\mathfrak{g}^\circ$, respectively. Additionally, we distinguish between Hamiltonians generated by Pauli strings that consist of two Pauli operators both different from $I$ (e.g.\ $\{XY, XZ\}$) and Hamiltonians generated by Pauli strings that include the identity (e.g.\ $\{IX, ZI, XX\}$). We denote the Lie algebras generated by the former as $\aa$ and the Lie algebras generated by the latter as $\bb$.

\section{Method \label{sec:method}}
We can now state the central question of our work. Given a Hamiltonian of the form \eqref{eq:hamiltonian_1d}, we seek a classification of all DLAs generated by the terms of the Hamiltonian, for $n\geq 3$ with both open and periodic boundary conditions.

\subsection{The power sets}
First, for the generators of $\aa$-type Lie algebras, we note that there are 9 Pauli strings that consist of two Pauli operators. Hence, the power set of the possible generators $\mcA$ contains $2^9-1=511$ elements. Similarly, for the $\bb$-type Lie algebras, there are 15 Pauli strings, which results in a power set of $2^{15}-1=32767$ possible sets of generators. Clearly, the sets of generators of the $\aa$-type are included in the $\bb$-type power set. We thus only report the $\bb$-type Lie algebras that are not also $\aa$-types. There is a third class of Lie algebras, the $\cc$-type Lie algebras, which are an edge case where the generators contain only Pauli strings of the form $a\otimes I$, but not the corresponding term $I\otimes a$. The structure of these Lie algebras is captured by the $\bb$-type Lie algebras, except for a small boundary effect at the last site in the chain. We therefore exclude the $\cc$-type Lie algebras from our classification.

We proceed by going through all sets of generators $\mcA$ (of either $\aa$ or $\bb$-type) and use Algorithm \ref{alg:dyn_lie} to perform the nested commutators in \eqref{eq:adj}. We then store only the unique subalgebras generated by this procedure, and we obtain only $127$ and $19$ generating sets for $\aa$-type and $\bb$-type, respectively. 
Since the largest power set we consider has only $32767$ elements, this procedure can be performed numerically with ease. We provide the code to reproduce this at \cite{our_code}.

\begin{example}  
    Consider the generating sets:
    \begin{align*}
        \mathcal{A}_1 &= \{XY, XZ\},\\
        \mathcal{A}_2 &= \{IX, XY\}.
    \end{align*}
    Note that $\mcA_1$ is of $\aa$-type and $\mcA_2$ is of $\cc$-type.
    After running Algorithm \ref{alg:dyn_lie}, we find: 
    \begin{align*}
        \Lie{\mathcal{A}_1} &= \Span\{XY,XZ, IX\},\\
        \Lie{\mathcal{A}_2} &= \Span\{IX,XY, XY\}.
    \end{align*}
    We see that $\mathcal{A}_1$ and $\mathcal{A}_2$ generate the same Lie algebra; hence, this Lie algebra is counted among the $\aa$-type Lie algebras.
\end{example}

\begin{example}
    Consider the generating sets of $\aa$-type:
    \begin{align*}
        \mathcal{A}_1 &= \{XX, YY\},\\
        \mathcal{A}_2 &= \{XX, YX\}.
    \end{align*}
    After running Algorithm \ref{alg:dyn_lie}, we find:
    \begin{align*}
        \Lie{\mathcal{A}_1} &= \Span\{XX,YY\},\\
        \Lie{\mathcal{A}_2} &= \Span\{XX,YX,ZX\}.
    \end{align*}
    Hence, $\mathcal{A}_1$ and $\mathcal{A}_2$ generate distinct Lie algebras, both of $\aa$-type.
\end{example}

\begin{example}
The set of $\bb$-type
\begin{align*}
     \mcA = \{XY, XI, IX\}
\end{align*}
generates the Lie algebra 
\begin{align*}
    \Lie{\mcA} = \Span\{XY, XZ, XI, IX\}. 
\end{align*}
If we try to generate it from the $\aa$-type subset $\{XY, XZ\}$, we find the strictly smaller subalgebra
\begin{align*}
    \Lie{XY, XZ} = \Span\{XY, XZ, IX\}. 
\end{align*}
\end{example}

\subsection{Symmetries of the power sets}
There are certain symmetries that can be exploited to reduce the number of subalgebras of $\su(4)$ in the above power sets. To start, we note that the Pauli matrices satisfy the following algebraic relations:
\begin{align*}
    [\sigma^\alpha, \sigma^\beta] =2i \sum_{\gamma=1}^3 \epsilon_{\alpha\beta\gamma} \sigma^\gamma, \label{eq:cyclic}
\end{align*}
where $\epsilon_{\alpha\beta\gamma}$ is the Levi-Civita tensor and $\alpha,\beta,\gamma \in\{1,2,3\}$, respectively 
(see SM \ref{secpauli1} for more details).
We will ignore the factor $2i$, since we only care about the linear span of nested commutators. Note that the above relation is independent of how we assign $X,Y,Z$ to $\sigma^\alpha$. In other words, we can relabel the Paulis and retain the algebraic structure of the subalgebras, which together with ignoring the prefactors formally corresponds to an $S_3$ permutation symmetry. 

In addition to relabelling, we consider the exchange of location of the two Pauli terms, since the order of such terms is an arbitrary choice that does not impact the structure of the resulting Lie algebras. This location exchange corresponds to a $\mathbb{Z}_2$ symmetry. Hence, the symmetry group of the Pauli algebra for $n=2$ is $S_3 \times \mathbb{Z}_2$. Subalgebras of $\su(4)$ that are in the same orbit of this symmetry group are considered equivalent, which allows us to reduce the number of subalgebras significantly. 


\begin{example}
    We have that $\{XX, YZ\} \equiv \{YY, ZX\}$ under relabeling $X\to Y\to Z \to X$. On the other hand, $\{XX, YZ\}$ and $\{XX, YY\}$ are not equivalent.
\end{example}


In order to determine the orbits of the symmetry group $S_3 \times \mathbb{Z}_2$ on the set of subalgebras of $\su(4)$, we introduce their invariants $s,p,e,d$, defined as follows.
These enumerate the number of single Paulis (such as $XI$) in the basis of the Lie algebra, the number of single Pauli pairs (such as $XI,IX$), the number of double equal Paulis (such as $XX$) and the number of double different Paulis (such as $XY$), respectively. Since all these quantities are invariant under the action of the symmetry group, two subalgebras are not equivalent if they have different invariant.

\begin{example}
    Consider the following bases of subalgebras and their invariants:
    \begin{align*}
        \mathcal{A}_1 &= \{ZZ, YX, XY\} \to (0, 0, 1, 2), \\
        \mathcal{A}_2 &= \{XX, YZ, ZY\} \to (0, 0, 1, 2), \\
        \mathcal{A}_3 &= \{YY, ZX, XZ\} \to (0, 0, 1, 2).
    \end{align*} 
    We see that $\mathcal{A}_1\equiv \mathcal{A}_2$ under $Z \rightleftharpoons X$. Similarly, $\mathcal{A}_2\equiv\mathcal{A}_3$ under $X\rightleftharpoons Y$ and $\mathcal{A}_3\equiv \mathcal{A}_1$ under $Y\rightleftharpoons Z$. 
\end{example}
\begin{example}
    Consider the following bases of subalgebras and their invariants:
    \begin{align*}
        \mathcal{A}_1 &=\{XX, XZ, IY\} \to (1,0,1,1), \\
        \mathcal{A}_2 &=\{XY, XZ, IX\} \to (1,0,0,2).
    \end{align*}
    We see that $\mathcal{A}_1\not\equiv \mathcal{A}_2$, since they have different invariants. 
\end{example}

\begin{example}\label{exsameinv}
    Even though the two bases
    \begin{align*}
        \mathcal{A}_1 &=\{XY, YX\} \to (0,0,0,2), \\
        \mathcal{A}_2 &=\{XY, YZ\} \to (0,0,0,2)
    \end{align*}
    have the same invariants, they are not equivalent under the symmetry group $S_3 \times \mathbb{Z}_2$. 
\end{example}

Carrying out this procedure exhaustively for the $127$ and $19$ subalgebras of $\aa$-type and $\bb$-type gives us $23$ and $5$ inequivalent Lie algebras, respectively. We denote these subalgebras by $\aa_k$ $(0\le k\le 22)$ and $\bb_l$ $(0\le l\le 4)$. For the full list of invariants, see Table \ref{tab:generator_list} in the Supplemental Materials. In particular, it turns out that the only case in which the invariants $(s,p,e,d)$ cannot distinguish inequivalent subalgebras is that presented in Example \ref{exsameinv}.

By Proposition \ref{rem:direct}, we can identify these subalgebras by inspection with direct sums of simple Lie algebras plus a center.

\begin{example}
The set
\begin{align*}
     \mcA = \{XX, YY, ZZ, ZY\}
\end{align*}
generates the Lie algebra
\begin{align*}
    &\aa_{20} := \Lie{\mcA} = \Span\{XX, YY, ZZ, YZ, ZY, XI, IX\} \\
    &= \Span\{YY+ZZ, YZ+ZY, XI+IX\} \\
    &\oplus \Span\{YY-ZZ, YZ-ZY, XI-IX\} \oplus \Span\{XX\} \\
    &\cong \su(2) \oplus \su(2)\oplus \uu(1).
\end{align*}
\end{example}

At this point, we have reduced the number of possible Hamiltonians in our class significantly by taking into account the symmetries of the Pauli group. As a final step, we now generate all $\aa_k(n)$, $\bb_k(n)$, $\aa^\circ_k(n)$ and $\bb_k^\circ(n)$ up to $n=8$. With knowledge of the structure of the Lie algebras for $3\le n\le 8$, we can construct formal proofs to determine them for all $n\geq 3$, which is discussed in the next section.

\section{Results\label{sec:results}}
\subsection{Main theorem}

We state the main theorem of our work below, and tabulate the generators of the Lie algebras of our classification in Table \ref{tab:algebras}.
\begin{theorem}[Classification of DLAs]\label{the:classification}
The complete list of dynamical Lie algebras of translation-invariant $2$-local Hamiltonians in one dimension of length $n\ge3$ is:
\allowdisplaybreaks
\begin{align*}
\aa_0(n) &\cong \uu(1)^{\oplus (n-1)}, \\
\aa_1(n) &\cong \so(n), \\
\aa_2(n) &\cong \aa_4(n) \cong \so(n) \oplus \so(n), \\
\aa_3(n) &\cong \begin{cases} 
        \so(2^{n-2})^{\oplus 4}, & n\equiv 0 \;\mathrm{mod}\; 8, \\
        \so(2^{n-1}), & n\equiv \pm1 \;\mathrm{mod}\; 8, \\
        \su(2^{n-2})^{\oplus2}, & n\equiv \pm2 \;\mathrm{mod}\; 8, \\
        \sp(2^{n-2}), & n\equiv \pm3 \;\mathrm{mod}\; 8, \\
        \sp(2^{n-3})^{\oplus4}, & n\equiv 4 \;\mathrm{mod}\; 8,      
        \end{cases} \\
\aa_5(n) & \cong \begin{cases} 
        \so(2^{n-2})^{\oplus 4}, & n\equiv 0 \;\mathrm{mod}\; 6, \\
        \so(2^{n-1}), & n\equiv \pm1 \;\mathrm{mod}\; 6, \\
        \su(2^{n-2})^{\oplus2}, & n\equiv \pm2 \;\mathrm{mod}\; 6, \\
        \sp(2^{n-2}), & n\equiv 3 \;\mathrm{mod}\; 6,
        \end{cases} \\
\aa_6(n) &\cong \aa_7(n) \cong \aa_{10}(n) \\
&\cong \begin{cases} \su(2^{n-1}), & n \;\;\mathrm{odd}, \\
        \su(2^{n-2})^{\oplus 4}, & n \ge 4 \;\;\mathrm{even},
        \end{cases} \\
\aa_8(n) &\cong \so(2n-1), \\
\aa_9(n) &\cong \sp(2^{n-2}), \\
\aa_{11}(n) &= \aa_{16}(n) = \so(2^n), \quad n\ge 4, \\
\aa_k(n) &= \su(2^n), \;\; k=12,17,18,19,21,22, \; n\ge 4, \\
\aa_{13}(n) &= \aa_{20}(n) \cong \aa_{15}(n) \cong \su(2^{n-1})^{\oplus2}, \\
\aa_{14}(n) &\cong \so(2n), \\
\bb_0(n) &\cong \uu(1)^{\oplus n}, \\
\bb_1(n) &\cong \uu(1)^{\oplus (2n-1)}, \\
\bb_2(n) &\cong \sp(2^{n-2}) \oplus \uu(1), \\
\bb_3(n) &\cong \su(2)^{\oplus n}, \\
\bb_4(n) &\cong \su(2^{n-1}) \oplus \su(2^{n-1}) \oplus \uu(1).
\end{align*}
\end{theorem}

\begin{table}[htb!]
\begin{center}
\begin{tabular}{|c|l|l|c|} 
\hline
\textbf{Label}   &   \textbf{Generating set}  &   \textbf{Model}
\\	\hline
$\aa_0$       &	$XX$	& Ising model 
\\	\hline
$\aa_1$       &	$XY$	& Kitaev chain 
\\	\hline
$\aa_2$       &	$XY, YX$	& Massless free fermion + \\
&& magnetic field
\\	\hline
$\aa_3$       &	$XX, YZ$	&  Kitaev chain + Coulomb
\\	\hline
$\aa_4$       &	$XX, YY$	& XY-model 
\\	 \hline
$\aa_5$       &	$XY, YZ$	& 
\\	\hline
$\aa_6$   	&	$XX, YZ, ZY$	& Massless free fermion + \\&& magnetic field +  Coulomb 
\\	\hline
$\aa_7$   	&	$XX, YY, ZZ$	& Heisenberg chain 
\\	\hline
$\aa_8$   	&	$XX, XZ$	& Ising model + transverse field 
\\	\hline
$\aa_9$   	&	$XY, XZ$	& Kitaev chain + longitudinal field 
\\	\hline
$\aa_{10}$    &	$XY, YZ, ZX$	& Heisenberg 
\\	\hline
$\aa_{11}$	&	$XY, YX, YZ$	& XY-model + longitudinal field
\\	\hline
$\aa_{12}$	&	$XX, XY, YZ$	& 
\\	\hline
$\aa_{13}$	&	$XX, YY, YZ$	& XY-model + longitudinal field
\\	\hline
$\aa_{14}$	&	$XX, YY, XY$	& XY-model + transverse field
\\	\hline
$\aa_{15}$	&	$XX, XY, XZ$	& Ising model + arbitrary field
\\	\hline
$\aa_{16}$	&	$XY, YX, YZ, ZY$	& Kitaev chain + longitudinal field 
\\	\hline
$\aa_{17}$	&	$XX, XY, ZX$ & Ising model + arbitrary field
\\	\hline
$\aa_{18}$	&	$XX, XZ, YY, ZY$	& XY-model + arbitrary field
\\	\hline
$\aa_{19}$	&	$XX, XY, ZX, YZ$ & 
\\	\hline
$\aa_{20}$	&	$XX, YY, ZZ, ZY$	& Heisenberg chain + magnetic field 
 \\	\hline
$\aa_{21}$	&	$XX, YY, XY, ZX$	& XY-model + arbitrary field
\\	\hline
$\aa_{22}$	&	$XX, XY, XZ, YX$& Ising model + arbitrary field
 \\	\hline
$\bb_0$	    &	$XI, IX$	& Uncoupled spins 
\\	\hline
$\bb_1$   	&	$XX, XI, IX$	& Ising model 
\\	\hline
$\bb_2$   	&	$XY, XI, IX$	&  Kitaev chain + longitudinal field
\\	\hline
$\bb_3$   	&	$XI, YI, IX, IY$	& Uncoupled spins
\\	\hline
$\bb_4$   	&	\parbox[t]{2cm}{$XX, XY, XZ, XI,$\\$IX, IY, IZ$}&	Ising model +  arbitrary field
\\	\hline
\end{tabular}
\end{center}
\caption{List of generators of the DLAs in Theorem \ref{the:classification}
and corresponding conventional spin models.
}
\label{tab:algebras}
\end{table}

The following corollary immediately follows from Theorem \ref{the:classification} and knowledge of the dimensions of $\su$, $\so$ and $\sp$ (see \eqref{eqdims}).

\begin{corollary}[Dimension scaling of DLAs\label{cor:scaling}]
    The dimension of all non-trivial DLAs of translation-invariant $2$-local Hamiltonians of length $n$ will scale as either $O(4^n)$, $O(n^2)$ or $O(n)$.
\end{corollary}
We thus see that the DLAs can be separated in three classes based on the scaling of their dimensions.

\subsection{Sketch of the proof}\label{secsketch}

The complete proof of Theorem \ref{the:classification} is presented in the Supplemental Materials. Here is a brief sketch of the proof; we refer to SM\ \ref{secout} for a more detailed outline. We divide the set of Lie algebras $\aa_k(n)$, $\bb_l(n)$ into three classes: linear, quadratic, and exponential, according to the anticipated growth of their dimension. The \emph{linear} class consists of $\aa_0(n)$ and $\bb_l(n)$ with $l=0,1,3$, and their treatment is obvious.
The \emph{quadratic} class contains $\aa_k(n)$ with $k=1,2,4,8,14$. These Lie algebras are determined by using the frustration graphs of their generators in SM\ \ref{secfrus}. For the \emph{exponential} class, we first observe that $\bb_2(n) = \aa_9(n) \oplus \Span\{X_1\}$ and $\bb_4(n) = \aa_{15}(n) \oplus \Span\{X_1\}$. Next, we identify the cases when $\aa_k(n) = \su(2^n)$; see SM\ \ref{secext} for details. We also find isomorphisms that are obtained by relabeling of the Pauli matrices among some of the algebras (SM\ \ref{seciso}).

The \textbf{strategy} in the remaining exponential cases is as follows.
\begin{enumerate}
\item For each of our Lie subalgebras $\s=\aa_k(n)\subseteq \su(2^n)$, we find its \emph{stabilizer} $\Stab(\s)$,
which is defined as the set of all Pauli strings $\in \mathcal{P}_n$ that commute with every element of $\s$. This can be done explicitly, because
the stabilizer is determined only from the generators of $\s$ (see Proposition \ref{pstab1}).
\item By definition, $\s$ commutes with all elements of its stabilizer $\Stab(\s)$; hence, it is contained in the 
\emph{centralizer} of $\Stab(\s)$ in $\su(2^n)$, which we denote $\su(2^n)^{\Stab(\s)}$. We can reduce the Lie subalgebra $\su(2^n)^{\Stab(\s)}$
further by factoring all elements of the center of $\Stab(\s)$, which will become central in it, because we have shown
that $\s$ has a trivial center (Lemma \ref{lemp3}). This results in a Lie algebra denoted $\g_k(n)$ when $\s=\aa_k(n)$.
\item By the above construction, we have $\aa_k(n) \subseteq \g_k(n)$. In the case of associative algebras, we would get equality due to (a finite-dimensional version of) von Neumann's Double Commutant Theorem (see e.g.\ \cite{procesi2007lie}, Theorem 6.2.5). However, in the Lie case, we might have a strict inclusion. We improve the upper bounds for $\aa_k(n)$
by finding \emph{involutions} $\theta_k$ of $\g_k(n)$ such that all elements of $\aa_k(n)$ are fixed under $\theta_k$.
The last condition can be checked only on the generators of $\aa_k(n)$ (see SM\ \ref{secup}). 
\item We prove by induction on $n$ that the upper bound is exact, that is $\aa_k(n) = \g_k(n)^{\theta_k}$ (see SM\ \ref{seclow}).
First we note that both $\aa_k(n)$ and $\g_k(n)^{\theta_k}$ are linearly spanned by the Pauli strings contained in them.
We start with an arbitrary Pauli string $a\in i\mcP_n\cap\g_k(n)^{\theta_k}$ and want to show that it is in $\aa_k(n)$.
The main idea is to use suitable commutators of $a$ with elements of $\aa_k(n)$ to produce a Pauli string $b\in i\mcP_n\cap\g_k(n)^{\theta_k}$ with $I$ in one of its positions. 
Erasing the $I$ in $b$ gives an element of $\g_k(n-1)^{\theta_k}$, which by induction is in $\aa_k(n-1)$.
\item 
Finally, we identify the Lie algebras $\g_k(n)^{\theta_k}$ with those from Theorem \ref{the:classification} (see SM\ \ref{secgkn}).
This is accomplished by applying in each case a suitable unitary transformation that brings the stabilizer $\Stab(\s)$ to a more convenient form
(cf.\ SM\ \ref{secpauli3}).
\end{enumerate}

\subsection{Example: \texorpdfstring{$\aa_9(n)$}{a\_9(n)}}\label{secexa9n}

Consider the example of $\aa_9=\Lie{XY, XZ}$, which produces the subalgebra $\aa_9(n) \subseteq\su(2^n)$ generated by:
\begin{equation}\label{a9ngen}
X_1Y_2, X_1Z_2, X_2Y_3, X_2Z_3, \dots, X_{n-1}Y_n, X_{n-1}Z_n. 
\end{equation}
Let us sketch the above steps in the strategy of the proof of Theorem \ref{the:classification} in the case $\s=\aa_9(n)$.
\begin{enumerate}
\item 
The stabilizer $\Stab(\s)$ is the set of all Pauli strings $P\in \mathcal{P}_n$ such that $[a,P]=0$ for every $a\in\s$.
It is enough to check this for all $a$ in the list of generators \eqref{a9ngen}, which means that the substring of $P$
in positions $j,j+1$ commutes with $XY$ and $XZ$ for all $1\le j\le n-1$.
By inspection, we find $\Stab(XY, XZ) = \{II,XI,YX,ZX\}$, so these are the only possible such substrings of $P$. 
This gives $\Stab(\s) = \{I^{\otimes n},X_1,Y_1X_2,Z_1X_2\}$.
\item 
The centralizer $\su(2^n)^{\Stab(\s)}$ is the set of all $a\in\su(2^n)$ such that $[a,P]=0$ for every $P\in\Stab(\s)$;
hence it contains $\s$. As the center of $\Stab(\s)$ is trivial, we have $\g_9(n) = \su(2^n)^{\Stab(\s)}$.
To illustrate this last step, we mention that $\Stab(\aa_{15}(n)) = \{I^{\otimes n},X_1\}$. 
In this case, $X_1 \in \su(2^n)^{X_1}$ is central and we have to quotient by it to obtain $\g_{15}(n) = \su(2^n)^{X_1} / \Span\{X_1\}$.
\item 
We saw above that $\s\subseteq \g_9(n)$. Now we find an involution $\theta_9$ of $\g_9(n)$ such that 
$\s\subseteq\g_9(n)^{\theta_9}$, the set of fixed points under $\theta_9$.
Since $\theta_9$ respects the Lie brackets, it is enough to check $\theta_9(a)=a$ only for the generators \eqref{a9ngen}.
Explicitly, we let $\theta_9(a) = -Q_9 a^T Q_9$ where $Q_9 = IYZZ \cdots Z$. 
\item 
We prove by induction on $n$ that $\aa_9(n) = \g_9(n)^{\theta_9}$. To show that any $a\in i\mcP_n\cap\g_9(n)^{\theta_9}$ with $n\ge4$ is in $\aa_9(n)$,
we first take suitable commutators of $a$ with the generators \eqref{a9ngen} to produce $b\in i\mcP_n\cap\g_k(n)^{\theta_k}$ that has $I$ in some position $j\ge3$.
Erasing the $I$ gives an element $c\in\g_9(n-1)^{\theta_9}$, which by induction is in $\aa_9(n-1)$. Inserting $I$ back in $j$-th place in $c$ gives that $b\in \aa_9(n)$.

\item 
As $\Stab(\s) \cong \{I^{\otimes n},X_1,Y_1,Z_1\}$, we can simplify $\g_9(n)^{\theta_9}$ by applying a unitary transformation $a\mapsto UaU^\dagger$ that takes $\Stab(\s)$ to
$\{I^{\otimes n},X_1,Y_1,Z_1\}$. Explicitly, we take $U = e^{i\frac{\pi}{4} X_1 X_2}$. 
Then $\g_9(n) \cong \su(2^n)^{\{X_1,Y_1,Z_1\}} = I \otimes \su(2^{n-1}) \cong \su(2^{n-1})$.
The involution $\theta_9(a) = -Q_9 a^T Q_9$ gets transformed to $-\tilde Q_9 a^T \tilde Q_9$, where $\tilde Q_9 = U Q_9 U^T$ in this case happens to be $=Q_9$.
Restricted to $\su(2^{n-1})$, this gives the involution $-Q a^T Q$ with $Q = YZZ \cdots Z$, whose fixed points are $\cong\sp(2^{n-2})$ because $Q^T=-Q$.
\end{enumerate}

We conclude that $\aa_9(n)\cong\sp(2^{n-2})$.

\subsection{Periodic boundary conditions}

For the periodic case, we have the same sets of generators as in Theorem \ref{the:classification}, but we need to adapt our proof strategy,
because, unlike the open case, the periodic Lie algebras $\aa_k^\circ(n)$ are not generated inductively from $\aa_k^\circ(n-1)$.
Instead, we use that $\aa_k^\circ(n)$ is generated from $\aa_k(n)$ and its cyclic shifts, and we utilize the explicit description $\aa_k(n) = \g_k(n)^{\theta_k}$ (see Part 4. in Sect.\ \ref{secsketch}, and for more details SM\ \ref{seclow}).

\begin{theorem}[Classification of Periodic DLAs]\label{the:classification-p}
The complete list of dynamical Lie algebras of translation-invariant periodic $2$-local Hamiltonians in one dimension of length $n\ge3$ is:
\allowdisplaybreaks
\begin{align*}
\aa_0^\circ(n) &\cong \uu(1)^{\oplus n}, \\
\aa_1^\circ(n) &\cong \so(n)^{\oplus 2}, \\
\aa_2^\circ(n) &\cong \so(n)^{\oplus 4}, \\
\aa_3^\circ(n) &= \begin{cases} 
\aa_{13}(n), & n \;\;\mathrm{odd}, \\
\aa_3(n), & n\equiv 0 \mod 4, \\
\aa_6(n), & n\equiv 2 \mod 4,
\end{cases} \\
&\cong \begin{cases} 
\su(2^{n-1})^{\oplus2}, & n \;\;\mathrm{odd}, \\
\so(2^{n-2})^{\oplus 4}, & n\equiv 0 \mod 8, \\
\sp(2^{n-3})^{\oplus4}, & n\equiv 4 \mod 8, \\
\su(2^{n-2})^{\oplus 4}, & n\equiv 2 \mod 4,
\end{cases} \\
\aa_4^\circ(n) &\cong \begin{cases}
\so(2n), \quad\;\; n \;\;\mathrm{odd}, \\
\so(n)^{\oplus 4}, \quad n \;\mathrm{even},
\end{cases} \\
\aa_{5}^\circ(n) &= \begin{cases}
\aa_{16}(n), & n\equiv \pm1 \mod 3, \\ 
\aa_{5}(n), & n\equiv 0 \mod 3, 
\end{cases} \\
&\cong \begin{cases}
\so(2^n), & n\equiv \pm1 \mod 3, \\ 
\so(2^{n-2})^{\oplus 4}, & n\equiv 0 \mod 6, \\
\sp(2^{n-2}), & n\equiv 3 \mod 6,
\end{cases} \\
\aa_6^\circ(n) &= \begin{cases}
\aa_{13}(n) \cong \su(2^{n-1})^{\oplus2}, & n \;\;\mathrm{odd}, \\
\aa_6(n) \cong \su(2^{n-2})^{\oplus 4}, & n \;\;\mathrm{even},
\end{cases} \\
\aa_k^\circ(n) &= \aa_k(n), \quad k=7,13,16,20, \\
\aa_8^\circ(n) &\cong \so(2n)^{\oplus 2}, \\
\aa_{9}^\circ(n) &\cong \so(2^n), \quad n \ge 4, \\
\aa_{10}^\circ(n) &= \begin{cases}
\su(2^n), & n\equiv \pm1 \mod 3, \\ 
\aa_{10}(n), & n\equiv 0 \mod 3,
\end{cases} \\
&\cong \begin{cases}
\su(2^n), & n\equiv \pm1 \mod 3, \\ 
\su(2^{n-2})^{\oplus 4}, & n\equiv 0 \mod 6, \\
\su(2^{n-1}), & n\equiv 3 \mod 6,
\end{cases} \\
\aa_{11}^\circ(n) &= \so(2^n), \quad n\ge 4, \\
\aa_k^\circ(n) &= \su(2^n), \;\; k=12,15,17,18,19,21,22, \\
\aa_{14}^\circ(n) &\cong \so(2n)^{\oplus 2}, \\
\bb_0^\circ(n) &= \bb_0(n) \cong \uu(1)^{\oplus n}, \\
\bb_1^\circ(n) &\cong \uu(1)^{\oplus 2n}, \\
\bb_2^\circ(n) &\cong \so(2^n), \quad n \ge 4, \\
\bb_3^\circ(n) &= \bb_3(n) \cong \su(2)^{\oplus n}, \\
\bb_4^\circ(n) &= \su(2^n).
\end{align*}
\end{theorem}
The proof of this theorem is given in SM\ \ref{secper}.

\subsection{Permutation-invariant subalgebras}\label{sec:sym}
The strategies employed for periodic boundary conditions can also be used to classify the DLAs in the case when the Hamiltonian is defined on a permutation-invariant graph. Now each spin is connected to each other spin via $2$-local interactions given by Pauli strings. We denote the resulting DLAs by $\mathfrak{g}^\pi$. As this is not the main focus of our work, we simply state the result here and provide the details in the Supplemental Materials.
\begin{theorem}[Classification of Permutation Invariant DLAs]\label{the:classification-s}
The complete list of dynamical Lie algebras of permutation-invariant $2$-site Hamiltonians in one dimension of length $n\ge3$ is:
\begin{align*}
\aa_k^\pi(n) &= \aa_k(n), \qquad k=7,16,20,22, \\
\aa_0^\pi(n) &\cong \uu(1)^{\oplus n(n-1)/2}, \\
\aa_2^\pi(n) &= \so(2^n)^{Z \cdots Z} \cong \so(2^{n-1})^{\oplus2}, \\
\aa_4^\pi(n) &= \aa_7(n) \cong \begin{cases} 
\su(2^{n-1}), & n \;\;\mathrm{odd}, \\
\su(2^{n-2})^{\oplus 4}, & n \ge 4 \;\;\mathrm{even},
\end{cases} \\
\aa_6^\pi(n) &= \aa_{20}(n) \cong \aa_{14}^\pi(n) \cong \su(2^{n-1})^{\oplus2}, \\
\bb_0^\pi(n) &= \bb_0(n) \cong \uu(1)^{\oplus n}, \\
\bb_1^\pi(n) &\cong \uu(1)^{\oplus n(n+1)/2}, \\
\bb_3^\pi(n) &= \bb_3(n) \cong \su(2)^{\oplus n}.
\end{align*}
\end{theorem}
The proof of this theorem is given in SM\ \ref{sec:sym}.

\section{Discussion \label{sec:discussion}}

In the previous section, we classified all DLAs of $2$-local spin chain Hamiltonians in one dimension. In this section, we discuss the importance of this classification for various fields in physics.

\subsection{Relevance for quantum control}
In quantum control, one is interested in performing a specific unitary evolution. This can be achieved by controlling some physical system described by a Hamiltonian $H$, parameterized by controls $\lambda(t)\in\mathbb{R}$ at each time step,
\begin{align*}
    H = \sum_{k=1}^{|\mathcal{A}|} \lambda_k(t) a_k,
\end{align*}
where $\mathcal{A}$ is now any set of operators $a_k$ that can be physically implemented. The set of unitaries that can be realized after a certain time $T\geq0$ is then determined by the Schr\"odinger equation in the Heisenberg picture (after setting $\hbar=1$):
\begin{align*}
    \frac{dU}{dt} = -iH(\lambda(t))\,U(t), \qquad U(0)=I.
\end{align*}
If the set of unitaries that can be reached after time $T$ is the full unitary group, the system is said to be (operator) \emph{controllable} \cite{dalessandro2021}. This is a desirable property when one is interested in building a quantum computer with a universal gate set. 

To test whether a system is controllable, one determines the DLA $\Lie{\mathcal{A}}$ and checks if it equals $\su(2^n)$. The conditions for complete controllability are known \cite{schirmer2001complete}, and it is in principle easy to create a controllable system since any real simple Lie algebra can be generated from $2$ elements \cite{ionescu1976generators}. If the DLA is a proper subalgebra of $\su(2^n)$, we say that the system is \emph{uncontrollable}. Note that this includes simple Lie algebras like $\so$ and $\sp$ \cite{schirmer2002identification}. Uncontrollable systems can arise when there are conserved quantities or symmetries in the physical system one is trying to control. Note that, due to Proposition \ref{rem:direct}, the DLA must split into a direct sum of simple Lie algebras and a center. If all simple summands in this decomposition are of the form $\su$, then we say that the system is \emph{subspace controllable} \cite{dalessandro2021}.

We can contextualize our classification in terms of these definitions. For example, we know that $\aa_{12}(n)$ will produce a controllable quantum system for $n\geq 4$ since this DLA is equal to $\su(2^n)$. Similarly, since $\aa_1(n)\cong \so(n)$, we know that it is uncontrollable. Finally, there are many examples of uncontrollable systems that consist of direct sums of $\su$ blocks; hence are subspace controllable. For instance, $\aa_3^\circ(n)$ for odd $n$ produces a DLA of the form $\su(2^{n-1})^{\oplus 2}$.

In addition to the notion of controllability of spin systems, we can ask what other types of systems we can simulate with our spin chains, e.g., fermionic or bosonic systems. This question was originally explored for Hamiltonians on cubic lattices with translation symmetry \cite{schuch2006harmonic,kraus2007quantum}. In particular, the dynamics of quadratic fermionic Hamiltonians is described by DLAs of the $\so(2n)$ or $\so(2n+1)$ type, which show up in our classification as $\aa_{14}(n)$ and $\aa_8(n)$. Similarly, the dynamics of a bosonic quadratic Hamiltonian with $n$ modes is related to a symplectic DLA \cite{zeier2011symmetry}, which we can identify with $\aa_5(n)$ for $n\equiv3$ mod $6$. Finally, one can consider composite systems and explore which subgroups of $\SU(N)$ can generated in this manner \cite{liu2021qudits}.

\subsection{Relevance for variational quantum computing}

A quantum circuit can be described as a product of unitaries $U = \prod_k U_k$. Typically, the quantum circuit $U$ acts on a multi-qubit state, whereas the gates $U_k$ only act on single or two qubit subsystems, i.e., we can write $U_k = e^{a_k}$ where $a_k$ is a $1$- or $2$-local operator. 
For a set of generators $\mcA=\{a_k\}$ with a corresponding DLA $\aa=\Lie{\mcA}$, we have that
\begin{align}
    e^{\aa} = \bigl\{ e^{a_{k_1} t_1} e^{a_{k_2} t_2} \cdots e^{a_{k_r} t_r} \,\big|\, t_i\in\mathbb{R}, \; a_{k_i} \in \mcA \bigr\}. \label{eq:subgroup}
\end{align}
In other words, any element in the Lie group $e^{\aa}$ generated by the DLA can be reached by a finite product of unitaries in that group (see \cite{dalessandro2021}, Corollary 3.2.6). In quantum computing, if $e^\aa=\SU(2^n)$, then the gate set $\{e^{a_k}\}$ is called \emph{universal} \cite{lloyd1996universal}. It is known that almost any combination of unitaries is universal \cite{lloyd1995almost,deutsch1995}. However, we can make specific choices for the generators $\{a_k\}$ that correspond to a non-universal gate set, which instead will generate a proper subgroup of $\SU(2^n)$. This is especially relevant for a class of quantum algorithms called variational quantum algorithms \cite{cerezo2021variational, tilly2022variational}.

If limited to $1$-dimensional topology, the generators in our classification will produce a circuit that is an element of the Lie group $e^\aa$. This notion can be used to construct specific quantum algorithms that always act within a subgroup of $\SU(2^n)$. Here, one considers a circuit that consists of parameterized gates,
\begin{align*}
    U(\btheta) =  U_1(\theta_1) U_2(\theta_2) \cdots U_K(\theta_K).
\end{align*}
The gate parameters $\btheta=(\theta_1,\dots,\theta_K)$ are real parameters that are optimized with a classical optimization routine to minimize a scalar cost function. A widely used example of such an optimization is the Variational Quantum Eigensolver algorithm (VQE) \cite{peruzzo2014variational}, which has a cost function given by
\begin{align}
    C(\btheta) = \Tr\bigl[ U(\btheta)\rho_0U^\dag(\btheta) H_c  \bigr],\label{eq:cost_fn}
\end{align}
where $H_c$ is a Hermitian operator and $\rho_0 = |\psi_0\rangle\langle\psi_0|$ is the initial state of the system. Crucial to the success of this algorithms is the choice of a circuit ansatz $U(\btheta)$ and the properties of the cost function \eqref{eq:cost_fn}.

\subsubsection{VQE ans\"atze}
A large class of variational circuits consist of $L$ repeating layers of unitary blocks \cite{kandala2017hardware,farhi2014quantum, wecker2015progress,ho2019efficient, choquette2021quantum, Kokcu2021cartan, matos2023characterization,dallaire2019low,anand2022quantum}, each with its own set of parameters:
\begin{align*}
    U(\btheta) = \prod_{l=1}^L \left( \prod_{k=1}^K U_k\bigl(\theta_{k}^{(l)}\bigr) \right) .
\end{align*}
In this section, we will give some examples of these circuits and how our classification relates to them. 

\begin{example}\textit{Hamiltonian Variational Ansatz.}
    The Hamiltonian Variational Ansatz circuit is obtained by Trotterizing the exponential of a Hamiltonian \cite{wecker2015progress,ho2019efficient}. Consider the Hamiltonian $H_{XY} = \sum_{i=1}^{n-1} X_i Y_{i+1}$, which has $\aa_1(n)$ as its DLA. Exponentiation of $H$ and the application of the Trotter--Suziki formula then gives:
    \begin{align*}
        U(\btheta) = \prod_{l=1}^L \left( \prod_{\mathrm{even}\:k} e^{i\theta_{k}^{(l)}X_k Y_{k+1}} \prod_{\mathrm{odd}\:k} e^{i\theta_{k}^{(l)}X_k Y_{k+1}} \right),
    \end{align*}
    where we grouped the odd and even terms together due to the structure imposed by the 1- and 2-qubit gates
    available on the quantum computer. Due to \eqref{eq:subgroup} and the knowledge that $\aa_1(n)\cong\so(n)$, we know that the above circuit must be a parameterization of a unitary operator $U(\btheta) \in\SO(n)$. 
    
    Similarly, we can take the DLA $\aa_9(n)$ with generators $\{XY, XZ\}$, which gives a circuit within $\mathrm{Sp}(2^{n-2})$:
    \begin{align*}
        U(\btheta ,\bphi) = &\prod_{l=1}^L 
        \Biggl(\prod_{\mathrm{even}\:k} e^{i\theta_{k}^{(l)}X_k Y_{k+1}} \prod_{\mathrm{odd}\:k} e^{i\theta_{k}^{(l)}X_k Y_{k+1}} \\
        &\times \prod_{\mathrm{even}\:k} e^{i\phi_{k}^{(l)}X_k Z_{k+1}} \prod_{\mathrm{odd}\:k} e^{i\phi_{k}^{(l)}X_k Z_{k+1}} \Biggr).
    \end{align*}
    We illustrate these circuits schematically in Figure \ref{fig:var_circuits}.
    \begin{figure}[htb!]
        \centering
        \subfloat[$\aa_1(n)$]{\includegraphics[width=0.30\columnwidth]{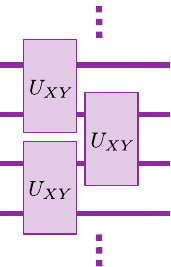}\label{fig:xy}}\hspace{5mm}
        \subfloat[$\aa_9(n)$]{\includegraphics[width=0.5\columnwidth]{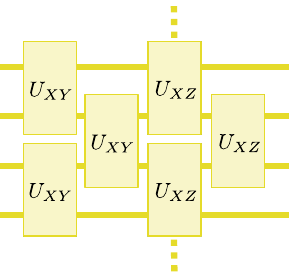}\label{fig:sp}}
        \caption{(a) Hamiltonian Variational Ansatz circuit for the Hamiltonian $H_{XY} = \sum_{i=1}^{n-1} X_i Y_{i+1}$, which parameterizes an element of the group $\SO(n)$. (b) Variational ansatz that parameterizes a unitary in $\mathrm{Sp}(2^{n-2})$ via products of unitaries generated by terms in $\aa_9(n)$.}
        \label{fig:var_circuits}
    \end{figure}
    We note that these types of brick-layer circuits also show up in the condensed matter physics literature on measurement-induced entanglement phase transitions \cite{yaodong2018mipt,li2019mipt,koh2023measurement}, hence our classification may be of use in that context as well.
\end{example}
\begin{example}\textit{Adapt-VQE.}
In ADAPT-VQE, one dynamically grows the circuit using a predetermined operator pool, so that each gate lowers the cost function by the largest amount \cite{grimsley2019adaptive}. This class of dynamical circuit ans\"atze can be understood as a Riemannian gradient flow over a specific subgroup \cite{wiersema2023riemannian}. This heuristic is popular in quantum chemistry for circuit design, where specific operator pools are considered that are tailored to fermionic Hamiltonians \cite{van2022scaling, yordanov2021qubit, tang2021qubit}. The operator pool can be seen as a set of generators, with a corresponding DLA. In the context of our classification, we can thus determine the resulting subgroup of the dynamically grown circuit ansatz based on the generators in the operator pool.

\begin{figure}[htb!]
    \centering{\includegraphics[width=0.8\columnwidth]{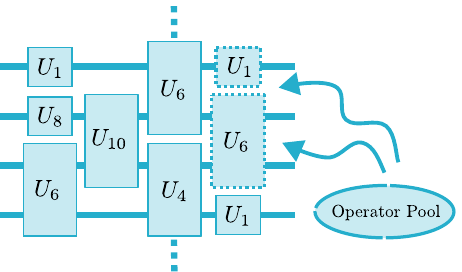}}\hspace{5mm}
    \caption{ADAPT-VQE circuit growing heuristic. We consider a generator pool in our classification and dynamically grow the circuit. }
    \label{fig:adapt_circuits}
\end{figure}
\end{example}

\begin{example}\textit{Permutation-invariant circuits.}
    Instead of a $1$-dimensional topology, one can consider a Hamiltonian with a fully connected topology (see Figure \ref{fig:rbm}):
    \begin{align*}
        H = \sum_{1\le i\neq j\le n} A_i B_j.
    \end{align*}
    This topology is common in ion trap quantum computers \cite{wright2019ionq} and also shows up in the context of quantum Boltzmann machines \cite{amin2018quantum, kappen2020learning}, which are the quantum equivalent of the Sherrington--Kirkpatrick model with tunable parameters \cite{sherrington1975solvable}. Closely related are the so-called permutation-equivariant circuits, which consist of parameterized blocks of unitaries that are permutation invariant \cite{schatzki2022theoretical}. These circuit ans\"atze were shown to be powerful quantum machine learning models for permutation-invariant data sets. Our classification of permutation-invariant $2$-site Hamiltonians in one dimension thus provides a classification of DLAs for these types of ans\"atze.
    \begin{figure}[htb!]
        \centering
        \includegraphics[width=\columnwidth]{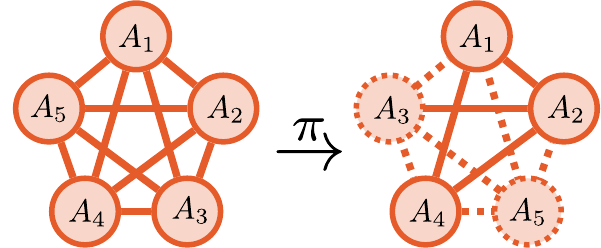}
        \caption{Permutation-invariant topology. Permuting sites leaves the Hamiltonian invariant.}
        \label{fig:rbm}
    \end{figure}
\end{example}

\subsubsection{Barren plateaus}

A hurdle in minimizing a cost function of the form \eqref{eq:cost_fn} are so-called \emph{barren plateaus} \cite{mcclean2018barren}, which are flat areas in the cost landscape of a variational quantum algorithm. When barren plateaus are present, the variance of gradients with respect to the gate parameters will decay, on average, exponentially as a function of system size. Hence, obtaining accurate estimates quickly becomes intractable due to the large number of shots required. There is a variety of different setups in which barren plateaus occur \cite{larocca2023theory, mcclean2018barren, cerezo2021cost, marrero2020entanglement, wang2021noiseinduced, holmes2022express,dankert2009exact}. 
To mitigate this problem, several recent works are aimed at finding ways to avoid the regions where optimization is hard \cite{taylor2020avoid, volkoff2021avoidbarren, grant2019initialization, zhou2020qaoa, skolik2020layerwise, pesah2020absence, wiersema2023measurement}. 

\begin{figure}[htb!]
    \centering
    \includegraphics[width=\columnwidth]{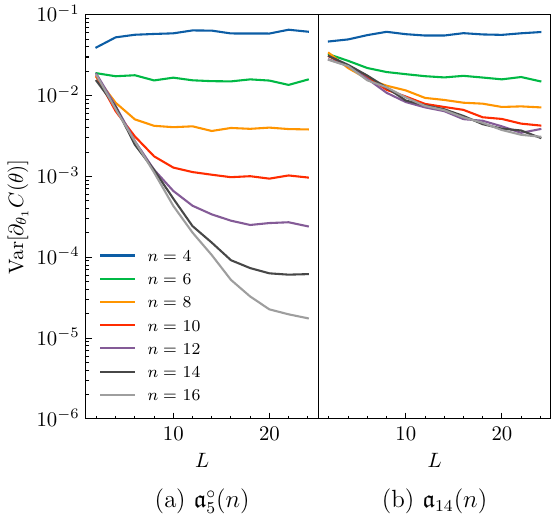}
    \caption{
    Barren plateaus in variational quantum circuits. We calculate the variance of $1000$ randomly initialized circuits with $\btheta$ sampled uniformly in $[0,\pi]$. 
     (a) The Lie algebra $\aa_5^\circ(n)$ is isomorphic to $\so(2^n)$, $\so(2^{n-2})^{\oplus 4}$ or $\sp(2^{n-2})$ depending on $n$; hence we expect exponentially decaying gradients for all $n$. This is confirmed in the figure above, since for a linear increase in $n$, we see an order of magnitude decrease in the gradient variances. (b) Since $\aa_{14}(n)\cong \so(2n)$, we find polynomially decaying gradients as a function of system size.
    }
    \label{fig:barren}
\end{figure}

The relevance of our classification for barren plateaus stems from the conjecture of \cite{Larocca2022diagnosing}, 
which states that the variance of the gradients of gate parameters is inversely proportional to the dimension of the DLA $\mathfrak{g}$ of the circuit:
\begin{align*}
    \mathrm{Var}[\partial_k C(\btheta)] \in O\left(\frac{1}{\mathrm{poly}(\dim\mathfrak{g})}\right).
\end{align*}
There are some subtleties involved in this conjecture, such as the locality of the cost function and the choice of initial state, which are discussed in \cite{Larocca2022diagnosing}. 
In the common case where $H_c\in i\g$, an exact 
formula for the variance was obtained independently in Refs.~\cite{fontana2023theadjoint,ragone2023},
which in particular refines and proves the above conjecture.
This formula was interpreted in Ref.~\cite{ragone2023} in terms of the \emph{$\g$-purity} \cite{somma2004nature,somma2005quantum} of the initial state $\rho_0$ and the observable $H_c$, underscoring again the crucial role of the DLA.

\begin{figure*}[htbp]
    \centering

    \subfloat{\includegraphics[width=0.67\textwidth]{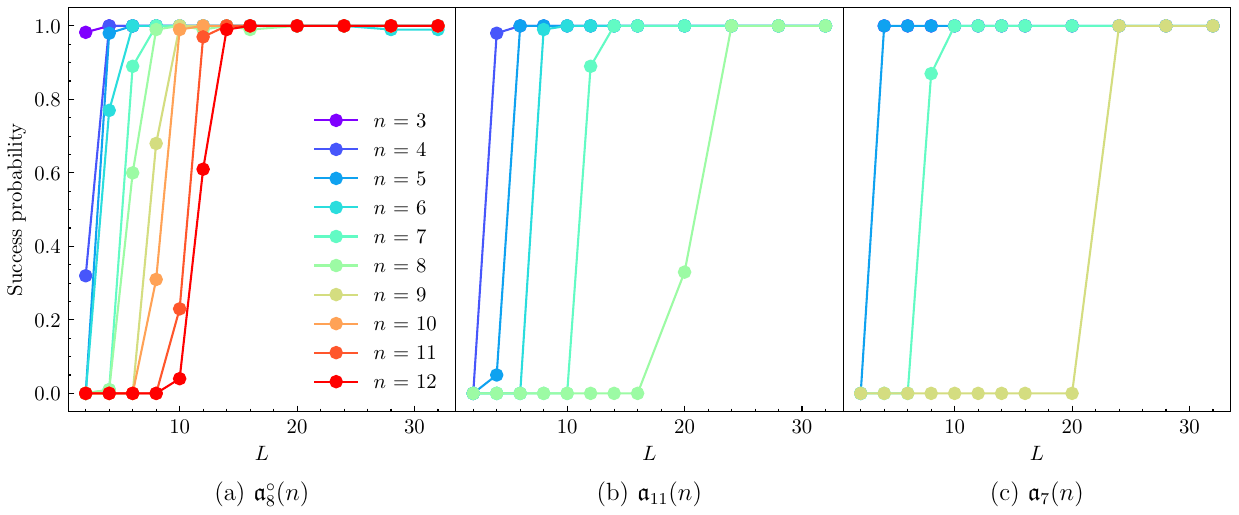}}
    \subfloat{\includegraphics[width=0.305\textwidth]{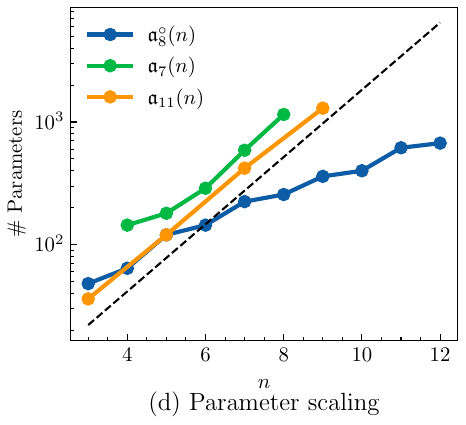}}
    \caption{
    Overparameterization of variational quantum circuits. We plot the success probability of reaching a state with $|E_0 - C(\btheta)| < 5\times 10^{-4}$ as a function of the circuit depth $L$, where $E_0$ is the lowest energy of the cost Hamiltonian $H_c$. These results were obtained by averaging $100$ random instances for each $L$ and $N$. The error bars indicate one standard deviation $\sigma$. A single random instance consists of $3000$ optimization steps of the Adam optimizer \cite{kingma2015adam} with learning rate $\eta = 10^{-2}$. A handful of instances converge to solutions that are further than $5\sigma$ from the mean and these outliers are therefore not included in the final plot. (a) The TFIM on a ring has DLA $\aa_8^\circ(n) \cong \so(2n)^{\oplus 2}$, whose dimension scales quadratically. We see that for a moderate circuit depth, the probability of success goes to 1. (b) Since the DLA $\aa_{11}(n)\cong \so(2^n)$, we expect that overparameterization occurs at depths that are exponential in the system size. Although this is not immediately clear in here, we see in (d) that the number of parameters indeed scales exponentially in $n$. (c) For the Heisenberg chain, which has an exponentially-scaling DLA $\aa_7(n)$, the choice of initial state $\rho_0=|0\rangle\langle0|^{\otimes n}$ prevents overparameterization from occurring for odd $n$. (d) If we set the threshold for overparameterization to be a success probability of $0.99$, we can plot the required number of parameters to reach this threshold. We see that $\aa_{11}(n)$ and $\aa_7(n)$ require an exponentially scaling number of parameters, whereas $\aa_8^\circ(n)$ only requires a polynomial number. The dashed line is a guiding line that indicates $O(4^n)$ scaling. 
    }
    \label{fig:overparam}
\end{figure*}

As an illustration, we compare the barren plateau
behavior for two of our Lie algebras, $\aa_5^\circ(n)$ and $\aa_{14}(n)$, whose dimensions scale exponentially and polynomially with $n$, respectively. We consider the cost function \eqref{eq:cost_fn} with $H_c=Z_1 Z_2$ and $\rho_0=|0\rangle\langle0|^{\otimes n}$. The circuit ansatz $U(\btheta)$ consists of unitaries generated by generators in our classification. To observe the barren plateau effect, we take the derivative of the cost function with respect to the first parameter in the first layer of the circuit, $\theta_1^{(1)}$. In Figure \ref{fig:barren}, we then observe the expected gradient decay as a function of the system size for an exponentially scaling DLA and a polynomially scaling DLA. In particular, in Figure \ref{fig:barren}(a) we consider the circuit generated by $H = \sum_{i=1}^{n-1} X_i Y_{i+1} + Y_i Z_{i+1}$ with periodic boundary conditions, whose DLA $\aa_5^\circ(n)$ is isomorphic to $\so(2^n)$, $\so(2^{n-2})^{\oplus 4}$ or $\sp(2^{n-2})$ depending on $n$ (see Theorem~\ref{the:classification-p}). Since $\dim\aa_5^\circ(n) = O(4^n)$, we expect the gradients to decay exponentially. Similarly, in Figure \ref{fig:barren}(b) we consider the circuit generated by $\aa_{14}(n)\cong \so(2n)$, which is described by the Hamiltonian $H = \sum_{i=1}^{n-1} X_i X_{i+1} + Y_i Y_{i+1}+ X_i Y_{i+1}$. Here, we have $\dim\aa_{14}(n) = O(n^2)$; hence, we expect the decay of gradients to be polynomial with respect to the system size.

According to Corollary~\ref{cor:scaling}, the only circuits free from barren plateaus generated by Hamiltonians in our classification, which are not composed of only $1$-qubit gates, have to be composed of $\so(n)$-type, since these are the only polynomially scaling DLAs in our classification. 

\subsubsection{Overparameterization}

Modern neural networks used in deep learning tend to have many more parameters than available data points, but are both easy to optimize and generalize well to unseen data in practice \cite{allen2019convergence}. This phenomenon is known as \emph{overparameterization}. In variational quantum computing, a similar effect has been observed \cite{kiani2020learning, wiersema2020exploring, kim2021overparam}, where deep variational quantum circuits tend to have favorable optimization properties. Recent works that have made progress in theoretically understanding this effect in quantum circuits can be connected to the DLA generated by the circuit ans\"atze used \cite{larocca2023theory,you2022convergence}. In particular, in \cite{larocca2023theory}, the dimension of the DLA can be used to analyze the Hessian around the global minimum of a typical variational quantum eigensolver cost function \cite{peruzzo2014variational}. Additionally, the authors find that the critical number of parameters needed to overparameterize a variational quantum circuit can be directly linked to the dimension of the associated DLA. In \cite{you2022convergence}, the authors study the optimization dynamics of overparameterized quantum circuit as perturbations of Riemannian gradient flows \cite{SchulteHerbruggen2010gradflow}. The size of the DLA (defined as the effective dimension in \cite{you2022convergence}) allows one to bound the number of parameters required to reach the overparameterization regime.

Corollary~\ref{cor:scaling} tells us that for quantum circuits constructed from the generators of $1$-dimensional spin chains, there are only DLAs whose dimension scales as $O(n)$, $O(n^2)$ and $O(4^n)$. Consequently, the linearly and quadratically scaling DLAs are expected to overparameterize with a non-exponential number of parameters. Additionally, the quadratically scaling DLAs in our classification correspond to free fermion models, whose dynamics can be simulated efficiently if $\rho_0$ is an eigenstate of $H_c$. However, choosing $\rho_0$ to be an arbitrary quantum state will still be intractable to simulate classically. 

As discussed in \cite{you2022convergence}, a requirement for overparameterzation is that the initial state has non-vanishing overlap with the ground state. Similarly, 
in \cite{mele2022avoiding}, it is shown that choosing the initial state in the right symmetry sector is crucial for the quality of the optimization. We highlight this importance in one of the numerical examples, where we choose an initial state that prevents overparameterization from occurring for an odd number of sites. 

In Figure~\ref{fig:overparam}, we illustrate the overparametrization phenomenon for three examples in our classification. In particular, in Figure~\ref{fig:overparam}(a), we consider the TFIM on a ring, which is given by the Hamiltonian $H_c = \sum_{i=1}^{n-1} Z_i Z_{i+1} + X_{i+1}$. The corresponding DLA is given by $\aa_8^\circ(n) \cong \so(2n)^{\oplus 2}$, whose dimension scales quadratically in $n$. We take the Hamiltonian Variational Ansatz of $H_c$ on even and odd qubits as a circuit ansatz, and take the initial state to be $\rho_0=|+\rangle\langle+|^{\otimes n}$. We observe that the cost landscape quickly becomes favorable, resulting in almost guaranteed convergence to the lowest energy state. 

In Figure~\ref{fig:overparam}(b), we take the DLA $\aa_{11}(n)\cong \so(2^n)$, and a Hamiltonian $H_c\in\so(2^n)$ given by a random orthogonal $2^n \times 2^n$ matrix. The circuit consists of unitaries generated by the generators of $\aa_{11}$ on even and odd qubits, and we take $\rho_0=|0\rangle\langle0|^{\otimes n}$. It now takes much deeper circuits to reach the same success probabilities as in Figure~\ref{fig:overparam}(a), which is due to the exponential scaling of the DLA. 

Finally, in Figure~\ref{fig:overparam}(c), we consider $\aa_7(n)$, which corresponds to the Heisenberg chain with $H_c = \sum_{i=1}^{n-1} X_i X_{i+1} + Y_i Y_{i+1} + Z_i Z_{i+1}$. The circuit is again the Hamiltonian Variational Ansatz of $H_c$, and $\rho_0=|0\rangle\langle0|^{\otimes n}$. This choice of an initial state only works for an odd number $n$ of qubits, while it fails to produce the overparameterization phenomenon for even $n$, leading to a success probability of $0$ (not plotted). Instead, for even $n$, the optimization of deep circuits gets stuck in a local minimum. We still observe the exponential scaling of the number of parameters, in accordance with the scaling of the dimension of $\aa_7(n)$, which is $O(4^n)$.

\subsection{Relevance for spin systems}

Our classification of Lie algebras arising in one dimension has significant 
bearing on a number of areas of physics and quantum simulation. The most
direct connection is that we have established a set of models, some
of which are traditional spin models \cite{rios2014,parkinson2010} studied in physics, while others are
new (cf.\ Table~\ref{tab:algebras_with_names}).
The integrability \cite{chen2011classification,franchini2017introduction,de2019classifying}, 
dynamical Lie algebra,
and symmetry of $1$-dimensional spin systems remains an active area of
research, and our result provides a database of models
where desired properties can be selected or different hypotheses tested.

For example, models such as the XXZ chain and the TFIM model arise quite naturally from nearest-neighbor weight-$2$ Pauli strings, and are represented in our Lie algebras as $\mathfrak{a}_{7}(n)$ and $\mathfrak{a}_{8}(n)$, respectively. A key
difference between these two models is that the TFIM is trivially integrable \cite{franchini2017introduction},
while the XXZ model is more complex.  This is reflected in the corresponding
Lie algebras as well, as $\mathfrak{a}_{8}(n) \cong \mathfrak{so}(2n-1)$ (which scales
polynomially in the number of qubits), and $\mathfrak{a}_{7}(n) \cong \su(2^{n-1})$ (which scales exponentially). 
Interestingly, the number of polynomially-scaling algebras is relatively small
($\mathfrak{a}_1,\mathfrak{a}_2,\mathfrak{a}_4,\mathfrak{a}_8,\mathfrak{a}_{14},
\aa^\circ_1, \aa^\circ_2, \aa^\circ_4, \aa^\circ_8,\aa^\circ_{14}$),
and they are all of the $\mathfrak{so}$ type. The limited size of these
algebras has been used to construct short- and/or fixed-depth circuits
for state preparation \cite{google2020hartree} and time evolution
\cite{csahinouglu2021hamiltonian,gu2021fast,Kokcu2021cartan,kokcu2022algebraic,camps2022algebraic} circuits.

The polynomially scaling algebras in principle come with a ``maximal set of independent commuting quantum operators'' \cite{caux2011remarks}, which
enables the integration in the first place. Unfortunately, our method does
not capture these because the conserved quantities are not single Pauli
strings. However, global symmetries are preserved for some of the models; these include $\mathbb{Z}_2$ (spin flip), $\SU(2)$ (global spin rotation) and $\UU(1)$ (global phase rotation). 

One particular property of note is the presence of \emph{non-commuting charges}---that 
is, elements of the stabilizer that do not commute. These are found in
$\aa_8(n), \aa_9(n)$ for all $n$, and in $\aa_2(n){-}\aa_7(n), \aa_{10}(n)$ for odd $n$ only.
Non-commuting charges give rise to a wide range of quantum effects in thermodynamics (see Ref.~\cite{majidy2023noncommuting} for a review). 
Notably,
the presence of non-commuting charges complicates questions regarding
thermalization.  Depending on the context, they either help thermalization
(e.g.\ by increasing entanglement entropy \cite{majidy2023non}) or hinder it (e.g.\ by 
invalidating the Eigenstate Thermalization Hypothesis \cite{murthy2023non}).
Although these effects have been primarily observed when the non-commuting charges are
extensive, and the ones discussed here are intensive, our framework could be extended to the
former as well.
Perhaps more interestingly within the context of quantum computing,
non-commuting charges couple the dynamics between different irreducible
representations of the charges, severely limiting the unitaries
that can be implemented \cite{marvian2023non}. 

A final point is the appearance of symplectic Lie algebras, which are not as common
as the orthogonal or unitary types. Here they appear from an AIII Cartan decomposition
of a larger Lie algebra; in applications, they come up
in the preparation of bosonic quantum states \cite{kraus2007quantum}, photonics \cite{yao2022recursive} and Clifford circuits/error correction \cite{rengaswamy2018synthesis}.

\section{Conclusion}

We have provided a classification of the dynamical Lie algebras (DLA) of $2$-local spin systems with open, periodic or permutation invariant topology in one dimension, and have discussed the relevance of this result in a variety of contexts. We have discovered several new examples beyond the standard Ising and Heisenberg models; thus increasing dramatically the number of explicit Hamiltonians available for theoretical investigations. 
It would be interesting to study in more detail the thermodynamic properties of these new Hamiltonians, and in particular to determine all of their symmetries, including the extensive non-commuting charges. 
We hope that our classification can be used to inspire new quantum algorithms and allow researchers to identify the circuits that they use in practice with the Lie algebras in our classification. Moreover, the methods that we have developed can be used to identify the DLA even in cases that fall outside of our classification.

One possible extension of our results would be to consider other topologies, such as $2$- and $3$-dimensional graphs. The challenge is that the number of possible Hamiltonians might become intractable. We would have to add more constraints to the Hamiltonians we consider, in order to reduce the size of the power set, or come up with alternative approaches to enumerate all unique DLAs.

Another future direction would be to consider other types of systems. For example, instead of spin systems, we could consider fermionic or bosonic Hamiltonians. Such a classification already exists for nearest-neighbor interactions on cubic lattices  \cite{kraus2007quantum, zimboras2014dynliefermion}, so this question would have to be explored in the context of non-cubic graphs.

\section*{Acknowledgements}
We acknowledge helpful discussions with Marco Cerezo, Ray Laflamme, Mart{\'\i}n Larocca, Carlos Ortiz Marrero, and Michael Ragone.
BNB was supported in part by a Simons Foundation grant No. 584741. 
RW acknowledges the resources provided by the Vector Institute  through its company sponsors \url{www.vectorinstitute.ai/#partners} and discussions with Juan Carrasquilla, Shayan Majidy, Roger Melko and Schuyler Moss.
EK and AFK acknowledge financial support from the National Science Foundation under award No. 1818914: PFCQC: STAQ: Software-Tailored Architecture for Quantum co-design.

\section*{Author Contributions}
The project was conceived by RW and BNB. Computer calculations and numerical simulations were performed by RW, with some assistance from AFK. Mathematical theorems were derived by BNB with the assistance of EK. The main part of the manuscript was written mostly by RW, while the Supplemental Materials were written by BNB. The figures were made by RW, and the tables by AFK and EK (except Table~\ref{tab:generator_list} by BNB). All authors contributed to reviewing and editing the manuscript. 

\section*{Data Availability}
Data generated and analyzed during the current study are available at \cite{our_code}.

\section*{Competing Interests}
The authors declare no competing interests.

\bibliographystyle{apsrev4-2}
\bibliography{literature}

\clearpage
\renewcommand{\appendixtocname}{Supplemental Material}
\onecolumngrid

\renewcommand\thefigure{S\arabic{figure}}  
\renewcommand\thetable{S.\Roman{table}}  
\setcounter{figure}{0}
\setcounter{table}{0}
\renewcommand\theHfigure{S\arabic{figure}}
\renewcommand\theHtable{S.\Roman{table}}

\renewcommand{\thesection}{\Alph{section}}
\renewcommand{\thesubsection}{\Roman{subsection}}

\renewcommand{\theequation}{\Alph{section}\arabic{equation}}
\counterwithin*{equation}{section}
\setcounter{section}{0}
\pagebreak
\begin{center}
\textbf{\large Supplemental Material}
\end{center}
\setcounter{equation}{0}
\setcounter{page}{1}
\makeatletter

\section{Preliminaries on Pauli strings and \texorpdfstring{$\su(2^n)$}{su(2 n)}}

Length-$n$ Pauli strings, when multiplied with the imaginary unit $i$, form a natural basis for the Lie algebra $\uu(2^n)$ of skew-Hermitian matrices. Because of this and their other remarkable properties, Pauli strings have been excessively used in this and many other works. 
In this section, we review the notation and basic properties of Pauli strings, the Lie algebra $\su(2^n)$, and its subalgebras $\so(2^n)$ and $\sp(2^{n-1})$.
We also discuss involutions of Lie algebras, and particularly of $\su(2^n)$.

\subsection{Pauli strings}\label{secpauli1}

Throughout the paper, we work with the \emph{Pauli matrices}
\begin{equation*}
\sigma_0=I=\begin{pmatrix} 1 & 0 \\ 0 & 1 \end{pmatrix}, \qquad
\sigma_1=X=\begin{pmatrix} 0 & 1 \\ 1 & 0 \end{pmatrix}, \qquad
\sigma_2=Y=\begin{pmatrix} 0 & -i \\ i & 0 \end{pmatrix}, \qquad
\sigma_3=Z=\begin{pmatrix} 1 & 0 \\ 0 & -1 \end{pmatrix},
\end{equation*}
including the identity matrix $I$,
which form a basis for the real vector space of $2\times 2$ Hermitian matrices.
We will denote by $A^T$ the transpose of a matrix, and by $A^\dagger$ its Hermitian conjugate
(which is obtained from $A^T$ by taking complex conjugates of all entries).
Thus, $A^\dagger = A$ for all $A\in\mcP_1 := \{I,X,Y,Z\}$. On the other hand, we have
\begin{equation*}
Y^T=-Y, \qquad A^T=A \quad\text{for}\quad A=I,X,Z.
\end{equation*}

Fix a positive number $n$. Length-$n$ \emph{Pauli strings} are tensor products of $n$ Pauli matrices of the form
\begin{equation}\label{pauli1}
a = A^1 \otimes A^2 \otimes \dots \otimes A^n, \qquad A^j \in\mcP_1
\end{equation}
(where the superscripts are indices not powers).
We denote the set of all such Pauli strings by $\mcP_n := \{I,X,Y,Z\}^{\otimes n}$.
Every $a\in\mcP_n$ is a linear operator on the Hilbert space $(\mathbb{C}^2)^{\otimes n}$ of $n$ qubits,
so $a$ can be represented as a matrix of size $2^n \times 2^n$ (by the Kronecker product).
In particular, $I^{\otimes n}$ is the $2^n \times 2^n$ identity matrix.
The Hermitian conjugate and transpose of a Pauli string are done componentwise:
\begin{align*}
a^\dagger &= (A^1)^\dagger \otimes (A^2)^\dagger \otimes \dots \otimes (A^n)^\dagger = a, \\
a^T &= (A^1)^T \otimes (A^2)^T \otimes \dots \otimes (A^n)^T = (-1)^{\# \{ A^j=Y \}} a.
\end{align*}
All Pauli strings are Hermitian, and $\mcP_n$ is a basis (over $\mathbb R$) of the vector space
of $2^n \times 2^n$ Hermitian matrices.

To shorten the notation, we will often omit the tensor product signs in Pauli strings, so \eqref{pauli1}
will be written as $a = A^1 A^2 \cdots A^n$. For example, we will write
\begin{equation*}
XX = X \otimes X, \qquad XY = X \otimes Y, \qquad Z \cdots Z = Z^{\otimes n}, \quad\text{etc.}
\end{equation*}
For $A\in\mcP_1$ and $1\le j\le n$, we will denote by
\begin{equation}\label{pauli2}
A_j := I^{\otimes (j-1)} \otimes A \otimes I^{\otimes (n-j)}
\end{equation}
the linear operator $A$ acting on the $j$-th qubit. For example, for $n=3$,
\begin{equation*}
X_1 = XII = X \otimes I \otimes I, \qquad Z_2 = IZI = I \otimes Z \otimes I, \qquad X_1Z_2Y_3 = XZY = X \otimes Z \otimes Y, \quad\text{etc.}
\end{equation*}
With this notation, we distinguish
\begin{equation*}
A_1 A_2 \cdots A_n = AA \cdots A = A \otimes A \otimes\cdots\otimes A = A^{\otimes n}
\end{equation*}
from \eqref{pauli1}, where in the latter the tensor factors $A^1,\dots,A^n$ are allowed to be different.

When there is a danger to confuse the tensor product and the matrix product, we will use $\cdot$ for the product of matrices. We have:
\begin{equation*}
X \cdot Y = i Z = -Y \cdot X, \qquad
Y \cdot Z = i X = -Z \cdot Y, \qquad
Z \cdot X = i Y = -X \cdot Z,
\end{equation*}
and each Pauli matrix squares to the identity:
\begin{equation*}
X \cdot X = Y \cdot Y = Z \cdot Z = I.
\end{equation*}
The matrix product of Pauli strings is done componentwise:
\begin{equation*}
(A^1 \otimes \dots \otimes A^n) \cdot (B^1 \otimes \dots \otimes B^n)
= (A^1 \cdot B^1) \otimes \dots \otimes (A^n \cdot B^n).
\end{equation*}
From here, it is easy to deduce the following important properties of Pauli strings.

\begin{lemma}\label{lemp0}
For any $a,b\in \mcP_n$, we have $a \cdot a = I^{\otimes n}$ and $a \cdot b = \pm b \cdot a$.
Hence, any two Pauli strings either commute or anticommute.
\end{lemma}

Notice that the product of two Pauli strings is again a Pauli string, up to a multiple of $\pm1,\pm i$. Thus, the set
$\{\pm a, \pm i a \,|\, a\in\mcP_n\}$ is a group under the matrix product, called the \emph{Pauli group}.
The following corollary of Lemma \ref{lemp0} will be useful in the future.

\begin{corollary}\label{lemp1}
For any $a,b\in \mcP_n$, if $[a,b] := a\cdot b-b\cdot a\ne 0$, then $[a,[a,b]]=4b$.
\end{corollary}
\begin{proof}
When $[a,b] \ne 0$, we have $[a,b]=2a\cdot b$ and $[a,[a,b]]=4a \cdot a \cdot b = 4b$.
\end{proof}

Another important consequence of Lemma \ref{lemp0} is \emph{Euler's formula}
\begin{align}\label{eqeuler}
e^{i \theta a} = (\cos \theta) I^{\otimes n} + i (\sin \theta) a, \qquad
a\in\mcP_n, \;\; \theta\in\mathbb R.
\end{align}
A useful special case is $\theta=\pi/2$; then
\begin{align}\label{eqeuler2}
e^{i\frac{\pi}{2} a} = i a, \qquad a\in\mcP_n.
\end{align}
Note that any $a\in\mcP_n$ is Hermitian (i.e., $a^\dagger=a$), $ia$ is skew-Hermitian (i.e., $(ia)^\dagger=-ia$),
and $U=e^{i \theta a}$ is unitary (i.e., $U U^\dagger=I^{\otimes n}$).
In the following, we will use the following corollary of Euler's formulas \eqref{eqeuler}, \eqref{eqeuler2}.

\begin{corollary}\label{corp1}
For any anticommuting $a,b\in \mcP_n$ and a real number $\theta$, we have
\begin{equation}\label{equab1}
e^{i\theta a} \cdot b = b \cdot e^{-i\theta a}.
\end{equation}
In particular,
\begin{equation}\label{equab}
e^{i\frac{\pi}{4} a} \, b \, e^{-i \frac{\pi}{4} a} = i a \cdot b. 
\end{equation}
\end{corollary}

\subsection{The Lie algebras \texorpdfstring{$\su(N)$}{su(N)}, \texorpdfstring{$\so(N)$}{so(N)}, and \texorpdfstring{$\sp(N)$}{sp(N)}}\label{secpauli2}

The purpose of this subsection is to review some standard terminology concerning Lie algebras, especially $\su(N)$, $\so(N)$, and $\sp(N)$.
We fix a positive integer $N$; later we will have $N=2^n$ where $n$ is the number of qubits.

The set $\mathbb{C}^{N \times N}$ of $N \times N$ matrices with complex entries is a vector space over $\mathbb C$,
as matrices can be added or multiplied by complex scalars. Then $\mathbb{C}^{N \times N}$ is an \emph{associative algebra} under the matrix product $ab$,
i.e., $ab$ is bilinear (depends linearly on both $a$ or $b$) and associative: $a(bc) = (ab)c$. 
Under the \emph{commutator bracket} $[a,b] = ab-ba$, we get a complex Lie algebra denoted as $\gl(N,\mathbb{C})$.

In general, a \emph{Lie algebra} is defined as a vector space $\g$ (over $\mathbb C$), equipped with a bilinear operation $[a,b]\in\g$ 
for $a,b\in\g$, which satisfies the following skew-symmetry and Jacobi identity:
\begin{equation*}
[a,b]=-[b,a], \qquad [a,[b,c]] = [[a,b],c]+[b,[a,c]].
\end{equation*}
It is convenient to use the notation $\ad_a(b) := [a,b]$ for $a,b\in\g$. Then $\ad_a$ is a linear operator on $\g$ for every $a\in\g$. 
A trivial example of a Lie algebra is any vector space $\g$ with the zero bracket $[a,b]=0$ for all $a,b\in\g$;
such Lie algebras are called \emph{Abelian}.

A \emph{subalgebra} $\s$ of a Lie algebra $\g$ is a subspace (i.e., closed under vector addition and scalar multiplication),
which is also closed under the bracket: $a,b\in\s$ $\Rightarrow$ $[a,b]\in\s$.
For example, the set $\sl(N,\mathbb{C})$ of all $N \times N$ complex matrices with trace $0$ is a subalgebra of $\gl(N,\mathbb{C})$;
hence, it is itself a Lie algebra. 
Recall that an \emph{ideal} in a Lie algebra $\g$ is a subspace $\s$ such that $a\in\g$, $b\in\s$ $\Rightarrow$ $[a,b]\in\s$.
For example, both $\sl(N,\mathbb{C})$ and $\mathbb{C} I_N$ (where $I_N$ is the $N \times N$ identity matrix) are ideals of $\gl(N,\mathbb{C})$, and
\begin{equation*}
\gl(N,\mathbb{C}) = \sl(N,\mathbb{C}) \oplus \mathbb{C} I_N
\end{equation*}
is a direct sum of not just subspaces but of commuting subalgebras and ideals.
When we write a direct sum of Lie algebras, we will always mean that the summands are subalgebras that commute with each other. 
The space $\mathbb{C} I_N$ is the \emph{center} $Z(\g)$ of $\g=\gl(N,\mathbb{C})$, i.e., the set of all $c\in\g$ such that $[c,a]=0$ for all $a\in\g$.

The Lie algebra $\sl(N,\mathbb{C})$ is \emph{simple}, which means that it is not Abelian and has no ideals other than the trivial $\{0\}$ and the whole algebra. 
Other examples of simple Lie algebras over $\mathbb C$ 
are provided by the \emph{orthogonal} Lie algebras $\so(N,\mathbb{C})$ 
and the \emph{symplectic} Lie algebras $\sp(2N,\mathbb{C})$. 
Let us recall that $\so(N,\mathbb{C})$ is defined as the set of all complex skew-symmetric matrices (i.e., such that $a^T=-a$), and it is a subalgebra
of $\sl(N,\mathbb{C})$. Consider the $2N \times 2N$ matrix
\begin{equation*}
J_{2N} := \begin{pmatrix} 0 & I_N \\ -I_N & 0 \end{pmatrix}.
\end{equation*}
Then $\sp(2N,\mathbb{C})$ is defined as the set of all $a\in\gl(2N,\mathbb{C})$ such that $a^T J_{2N} = -J_{2N} a$; this is a subalgebra
of $\sl(2N,\mathbb{C})$.

In this paper, we will work with Lie algebras over $\mathbb R$.
The set of all skew-Hermitian matrices (i.e., satisfying $a^\dagger=-a$) is a real vector space, and is closed under the commutator;
hence, it is a real Lie algebra, denoted $\uu(N)$.
Imposing that the trace of the matrix is $0$, we get the Lie algebra $\su(N)$.
If a matrix has real entries, then it is skew-Hermitian if and only if it is skew-symmetric. Thus, we have the subalgebra
\begin{equation}\label{soN}
\so(N) := \so(N,\mathbb{R}) = \bigl\{ a\in\su(N) \,\big|\, a^T = -a \bigr\} \subset\su(N).
\end{equation}
On the other hand, the Lie algebra of real symplectic matrices $\sp(2N,\mathbb{R})$ is \emph{not} a subalgebra of $\su(2N)$.
Instead of it, the relevant subalgebra is
\begin{equation}\label{spN}
\sp(N) := \sp(2N,\mathbb{C}) \cap \su(2N) = \bigl\{ a\in\su(2N) \,\big|\, a^T J_{2N} = -J_{2N} a \bigr\} \subset\su(2N).
\end{equation}
The Lie algebras $\su(N)$, $\so(N)$, $\sp(N)$ are simple 
and \emph{compact} (they are Lie algebras of compact Lie groups).
Their dimensions over $\mathbb{R}$ are given by:
\begin{equation}\label{eqdims}
\dim\su(N) = N^2-1, \qquad
\dim\so(N) = \frac{1}2 N(N-1), \qquad
\dim\sp(N) = N(2N+1).
\end{equation}
It is known (see e.g.\ \cite{zeier2011symmetry, dalessandro2021}) that any subalgebra of $\uu(N)$ is either Abelian or a direct sum of a center (which could be $\{0\}$) and Lie algebras isomorphic to one of $\su$, $\so$, $\sp$ or to one of five exceptional compact simple Lie algebras (cf.\ Proposition \ref{rem:direct}). 
For completeness, we provide the proof of this important fact here.

\begin{proposition}\label{rem:direct-SM}
Any subalgebra of $\uu(N)$ is either Abelian or a direct sum of compact simple Lie algebras and a center. 
\end{proposition}
\begin{proof}
First, recall that the trace form $(a,b) = \Tr[ab]$ is negative definite on $\uu(N)$.
Indeed, one can see that
\begin{equation*}
\Tr[H^2] = \sum_{j=1}^N \lambda_j^2 > 0
\end{equation*}
for any nonzero Hermitian matrix $H\in i\uu(N)$ with eigenvalues $\lambda_1,\dots,\lambda_N$, because all $\lambda_j$ are real.
Second, the trace form $(\cdot,\cdot)$ is bilinear, symmetric and invariant; the latter meaning that
\begin{equation*}
([a,b],c) = -(b,[a,c]), \qquad a,b,c \in\uu(N).
\end{equation*}
All of these follow easily from the properties of the trace.

Now let $\g$ be a subalgebra of $\uu(N)$. Then the same proof as in \cite{knapp2013lie}, Corollary 4.25, shows that $\g$ is \emph{reductive}, i.e., it is a direct sum of simple or Abelian ideals. Indeed, for any ideal $\mathfrak s$ of $\g$, we have an orthogonal direct sum
\begin{equation*}
\g = \mathfrak{s} \oplus \mathfrak{s}^\perp,
\end{equation*}
due to the definiteness of the trace form. Furthermore, the invariance of the trace form implies that $\mathfrak{s}^\perp$ is itself an ideal of $\g$. As both $\mathfrak s$ and $\mathfrak{s}^\perp$ are ideals, they must commute: 
$[\mathfrak{s},\mathfrak{s}^\perp] \subseteq \mathfrak{s} \cap \mathfrak{s}^\perp = \{0\}$. Thus, $\g$ is a direct sum of commuting ideals.
Suppose that $\dim\mathfrak s > 1$ and $\mathfrak s$ is not simple as a Lie algebra. Then $\mathfrak s$ has a nonzero proper ideal $\mathfrak t$. From $[\mathfrak{t},\mathfrak{s}] \subseteq \mathfrak{t}$ and $[\mathfrak{t},\mathfrak{s}^\perp] = \{0\}$, we get that $[\mathfrak{t},\mathfrak{g}] \subseteq \mathfrak{t}$, so $\mathfrak t$ is an ideal of $\g$. Hence, we can proceed by induction on $\dim\g$ and break $\g$ into a direct sum of commuting ideals, each of which is either simple or $1$-dimensional (Abelian).

Finally, we note that any simple subalgebra $\g$ of $\uu(N)$ is compact, i.e., the Lie group $e^\g$ is compact as a closed subgroup of the compact Lie group $\UU(N)$ (see e.g.\ \cite{knapp2013lie}, Proposition 4.27, whose proof still applies). 
The classification of all compact simple Lie algebras is due to Cartan, and can be found in standard textbooks; for example in \cite{knapp2013lie}, Chapter~VI.
\end{proof}

Let now $N=2^n$ where $n$ is the number of qubits. Then a basis for $\uu(2^n)$ is given by $i\mcP_n$ (see Sect.\ \ref{secpauli1}).
To get a basis for $\su(2^n)$, we just have to remove $i I^{\otimes n}$, since $I^{\otimes n} = I_N$ is the identity matrix.
A basis for $\so(2^n)$ consists of all $ia$ where $a\in\mcP_n$ is a Pauli string containing an odd number of $Y$'s.
Finally, to describe the subalgebra $\sp(2^{n-1}) \subset \su(2^n)$, we observe that
\begin{equation*}
Y_1 = Y \otimes I^{\otimes (n-1)} 
= -i \begin{pmatrix} 0 & I_{\frac{N}2} \\ -I_{\frac{N}2} & 0 \end{pmatrix} = -i J_{N}.
\end{equation*}
Therefore,
\begin{equation}\label{sp2n}
\sp(2^{n-1}) = \bigl\{ a\in\su(2^n) \,\big|\, a^T \cdot Y_1 = -Y_1 \cdot a \bigr\}.
\end{equation}

The following definition plays a crucial role throughout the paper.

\begin{definition}\label{defgen}
For a Lie algebra $\g$ and a subset $\mcA\subset\g$, we define $\Lie{\mcA}$ to be the minimal (under inclusion) subalgebra of $\g$ that contains $\mcA$. We say that $\Lie{\mcA}$ is the subalgebra \emph{generated} by $\mcA$, and that $\mcA$ is a set of \emph{generators} of $\Lie{\mcA}$. In the case when $\mcA\subset\mcP_n$ is a set of Pauli strings, we will slightly abuse the notation and write $\Lie{\mcA}$ for the subalgebra of $\su(2^n)$ generated by the subset $i\mcA\subset\su(2^n)$.
\end{definition}

More explicitly, it follows from the Jacobi identity that $\Lie{\mcA}$ is the set of all linear combinations of all nested commutators of the form
\begin{equation}\label{eqnested}
\ad_{a_1} \ad_{a_2} \cdots \ad_{a_r} (a_{r+1}) = [a_1,[a_2,[ \cdots [a_r,a_{r+1}] \cdots ]]], \qquad a_j \in\mcA, \;\; r\ge0,
\end{equation}
where $r=0$ corresponds to an empty commutator $=a_1$. The following simple observation will be useful.

\begin{lemma}\label{lemgen}
For any subset $\mcA\subset\mcP_n$, the Lie algebra $\Lie{\mcA} \subseteq\su(2^n)$ generated by $i\mcA$ has a basis consisting of Pauli strings times $i$. In other words,
\begin{equation*}
\Lie{\mcA} = \Span_{\mathbb{R}} \bigl( i\mcP_n \cap \Lie{\mcA} \bigr).
\end{equation*}
\end{lemma}
\begin{proof}
By definition, $\Lie{\mcA}$ is linearly spanned over $\mathbb{R}$ by all elements of the form \eqref{eqnested} with $a_j\in i\mcA \subset i\mcP_n$.
All such elements are scalar multiples of Pauli strings, i.e., lie in $i\mathbb{R}\mcP_n$. From any spanning set, one can choose a subset that forms a basis.
\end{proof}

\subsection{Involutions of\texorpdfstring{ $\su(2^n)$}{su(2 n)}}\label{secpauli3}

In this subsection, we explain how we can describe subalgebras of $\su(2^n)$ as fixed points of involutions.
We start more generally by recalling that an \emph{isomorphism} from a Lie algebra $\g$ to another Lie algebra $\g_1$ is an invertible linear transformation
$\varphi\colon\g\to\g_1$ that is compatible with the bracket, i.e., $\varphi([a,b])=[\varphi(a),\varphi(b)]$ for all $a,b\in\g$. We write $\g\cong\g_1$ to indicate that $\g$ is isomorphic to $\g_1$. An \emph{automorphism} of $\g$ is an isomorphism $\varphi\colon\g\to\g$. The set of fixed points of $\varphi$ is defined as:
\begin{equation}\label{gphi1}
\g^\varphi = \bigl\{ a\in\g \,\big|\, \varphi(a)=a \bigr\}.
\end{equation}
It is easy to check that $\g^\varphi$ is a subalgebra of $\g$, called the \emph{fixed-point subalgebra}.
An \emph{involution} on $\g$ is an automorphism $\theta\colon\g\to\g$ with the property that $\theta(\theta(a))=a$ for all $a\in\g$, 
i.e., $\theta^{-1}=\theta$. 

Later, we will need to understand how, for an involution $\theta$, the fixed-point subalgebra $\g^\theta$ transforms under another automorphism $\varphi$ of $\g$. The answer is given in the following lemma.

\begin{lemma}\label{leminv1}
Let $\varphi$ be an automorphism of a Lie algebra $\g$, and $\theta$ be an involution of $\g$. Then $\varphi\theta\varphi^{-1}$ is an involution of $\g$, and we have
$\varphi(\g^\theta) = \g^{\varphi\theta\varphi^{-1}}$.
\end{lemma}
\begin{proof}
Note that $\varphi\theta\varphi^{-1}$ is an automorphism of $\g$, since the composition of isomorphisms is again an isomorphism. It is an involution of $\g$, because
$(\varphi\theta\varphi^{-1})^{-1} = \varphi\theta^{-1}\varphi^{-1} = \varphi\theta\varphi^{-1}$. To check the claim about the fixed points, suppose that $a\in\g^\theta$, i.e., $\theta(a)=a$. Then
\begin{equation*}
(\varphi\theta\varphi^{-1}) \varphi(a) = \varphi(\theta(a)) = \varphi(a) \;\;\Rightarrow\;\; \varphi(a) \in \g^{\varphi\theta\varphi^{-1}}.
\end{equation*}
Conversely, if $b\in\g^{\varphi\theta\varphi^{-1}}$, then the above calculation shows that $a=\varphi^{-1}(b) \in\g^\theta$.
\end{proof}

Now we will discuss how to construct automorphisms and involutions of $\su(N)$.

\begin{lemma}\label{leminv2}
Suppose that $U$ and $Q$ are unitary $N\times N$ matrices, and $Q^T = \pm Q$.
Then the formulas
\begin{equation*}
\varphi(a) = U a U^\dagger, \qquad \theta(a) = -Q a^T Q^\dagger, \qquad a\in\su(N),
\end{equation*}
define an automorphism $\varphi$ and an involution $\theta$ of $\su(N)$. Moreover, we have
\begin{equation*}
(\varphi\theta\varphi^{-1})(a) = -(U Q U^T) a^T (U Q U^T)^\dagger, \qquad a\in\su(N).
\end{equation*}
\end{lemma}
\begin{proof}
First, note that $\varphi$ is invertible with $\varphi^{-1}(a) = U^\dagger a U$. It is clear that $\varphi([a,b])=[\varphi(a),\varphi(b)]$, because
\begin{equation*}
\varphi(a) \varphi(b) = U a U^\dagger U b U^\dagger = U a b U^\dagger = \varphi(a b) 
\end{equation*}
for any two matrices $a,b$. (This means that $\varphi$ is an automorphism of the associative algebra $\mathbb{C}^{N\times N}$.)

To check that $\theta$ is an involution, we calculate
\begin{equation*}
\theta(\theta(a)) = \theta(-Q a^T Q^\dagger) = Q (Q^\dagger)^T (a^T)^T Q^T Q^\dagger
= Q (\pm Q^\dagger) a (\pm Q) Q^\dagger = a,
\end{equation*}
where we used that $(Q^\dagger)^T = (Q^T)^\dagger$.
Next, we have
\begin{equation*}
\theta(a b) = -Q (a b)^T Q^\dagger 
= -Q b^T a^T Q^\dagger = -Q b^T Q^\dagger Q a^T Q^\dagger = -\theta(b) \theta(a).
\end{equation*}
This implies that $\theta([a,b]) = [\theta(a),\theta(b)]$ and proves that $\theta$ is an involution of $\su(N)$.

Finally, we find
\begin{align*}
(\varphi\theta\varphi^{-1})(a) &= \varphi\bigl(\theta(U^\dagger a U)\bigr) \\
&= -\varphi\bigl(Q U^T a^T (U^\dagger)^T Q^\dagger\bigr) \\
&= -U Q U^T a^T (U^\dagger)^T Q^\dagger U^\dagger \\
&= -(U Q U^T) a^T (U Q U^T)^\dagger,
\end{align*}
as claimed.
\end{proof}

\begin{example}
The subalgebra $\so(N)\subset\su(N)$ is the fixed-point subalgebra of the involution $a\mapsto -a^T$ (see \eqref{soN}).
Similarly, we see from \eqref{spN} that $\sp(N) = \su(2N)^\theta$, where $\theta(a) = -Q a^T Q^\dagger$
as in Lemma \ref{leminv2}, with $Q=-i J_{2N}$. 
\end{example}

It is well known, in general, that for an involution $\theta(a) = -Q a^T Q^\dagger$
as in Lemma \ref{leminv2}, the fixed-point subalgebra $\su(N)^\theta \cong \so(N)$ when $Q^T=Q$, and
$\su(N)^\theta \cong \sp(N/2)$ when $Q^T=-Q$ (in which case $N$ must be even).
For completeness, we will present the proof of these facts in the special case of interest to us: when $N=2^n$ and $Q\in\mcP_n$ is a length-$n$ Pauli string.
Note that all Pauli strings $Q$ satisfy $Q=Q^\dagger=Q^{-1}$ and $Q^T = \pm Q$.

\begin{lemma}\label{leminv3}
For any Pauli string $Q\in\mcP_n$, there exists a unitary $2^n\times 2^n$ matrix $U$ such that
\begin{equation*}
U Q U^T = \begin{cases} 
I^{\otimes n}, & \text{if} \;\; Q^T=Q, \\
Y_1, & \text{if} \;\; Q^T=-Q.
\end{cases}
\end{equation*}
\end{lemma}
\begin{proof}
In the case $Q^T = Q$, we let $U = e^{i\frac{\pi}{4} Q}$. Then
\begin{equation*}
U Q U^T = U Q U = U^2 Q = e^{i\frac{\pi}{2} Q} Q = i Q \cdot Q = i I^{\otimes n},
\end{equation*}
where we used Euler's formula \eqref{eqeuler2}. The superfluous phase $i$ can be eliminated by applying the unitary transformation
$V=e^{-i\frac{\pi}{4} I^{\otimes n}}$, which satisfies $V^T=V$ and $V^2 = -i I^{\otimes n}$.

Suppose now that $Q^T = -Q$, which means that $Q$ contains an odd number of $Y$'s.
If $Q$ has a $Y$ in $j$-th position, let $P = Y_j \cdot Q$ and $U = e^{i\frac{\pi}{4} P}$. Note that 
\begin{equation*}
P = Y_j \cdot Q = Q \cdot Y_j \;\;\Rightarrow\;\; P \cdot Q = Q \cdot P = Y_j
\end{equation*}
and
\begin{equation*}
P^T = Q^T \cdot Y_j^T = (-Q) \cdot (-Y_j) = Q \cdot Y_j = P.
\end{equation*}
Hence, as above, we find:
\begin{equation*}
U Q U^T = U Q U = U^2 Q = e^{i\frac{\pi}{2} P} Q = i P \cdot Q = i Y_j.
\end{equation*}
If $j=1$, we are done (after eliminating the phase $i$).
If $j\ne1$, we apply the unitary transformation $e^{-i\frac{\pi}{4} Y_1 Y_j}$ and obtain:
\begin{equation*}
e^{-i\frac{\pi}{4} Y_1 Y_j} (i Y_j) (e^{-i\frac{\pi}{4} Y_1 Y_j})^T 
= e^{-i\frac{\pi}{2} Y_1 Y_j} (i Y_j) = -i Y_1 Y_j \cdot (i Y_j) = Y_1,
\end{equation*}
completing the proof.
\end{proof}

\begin{corollary}\label{corinv4}
Any Pauli string $Q\in\mcP_n$ defines an involution $\theta$ of $\su(2^n)$ given by $\theta(a) = -Q a^T Q$. 
The fixed-point subalgebra of this involution is:
\begin{equation*}
\su(2^n)^\theta \cong \begin{cases} 
\so(2^n), & \text{if} \;\; Q^T=Q, \\
\sp(2^{n-1}), & \text{if} \;\; Q^T=-Q.
\end{cases}
\end{equation*}
\end{corollary}

\section{Statement of Results}\label{secres}

\subsection{Subalgebras of \texorpdfstring{$\su(4)$}{su(4)} up to symmetry}\label{secres1}

Recall that $\su(4)$ has a basis over $\mathbb R$ consisting of all possible tensor products
$iAB$, where $A,B\in\{I,X,Y,Z\}$, $AB\ne II$. As a first step, we found that there are $202$ subalgebras of $\su(4)$, which are generated
by subsets of this basis. The symmetry group $S_3\times\mathbb{Z}_2$ acts on $\su(4)$ as follows: the symmetric group $S_3$ permutes
simultaneously all $\{X,Y,Z\}$, while the non-identity element of $\mathbb{Z}_2$ acts as the flip $AB \rightleftharpoons BA$.
We examined the orbits of the action of $S_3\times\mathbb{Z}_2$ on the $202$ subalgebras, and found that there are $36$ orbits, which are
listed in Table \ref{tab:generator_list} below. The full list of all $202$ subalgebras is presented in the next subsection.

The numbers $s,p,e,d$ are equal to the numbers of: single Paulis (such as $XI$), single Pauli pairs (such as $XI,IX$), 
double equal Paulis (such as $XX$), and double different Paulis (such as $XY$), respectively, in the basis of the subalgebra.
These are invariant under the action of the symmetry group; hence, subalgebras of $\su(4)$ with different invariants are not equivalent to each other.
It turns out that the only two non-equivalent subalgebras with the same invariants are $\aa_2$ and $\aa_5$.

We distinguish between three types of subalgebras. The $\aa$-type Lie algebras are generated by products of two Paulis both different from the identity, whereas the $\bb$-type can be generated by Pauli strings that contain the identity. The $\cc$-type Lie algebras are an edge case where the generators contain some Pauli strings of the form $a\otimes I$ without the corresponding term $I\otimes a$. These Lie algebras will behave like the $\bb$-type Lie algebras but include a boundary effect at the last site in the chain. In particular, note that all single Pauli generators come in pairs such as $XI,IX$ or $YI,IY$ or $ZI,IZ$ due to the translation invariance; hence we exclude the Lie algebras $\cc_0,\dots,\cc_7$ from our classification. However, we include them in the following tables for completeness. 


\begin{table}[htb!]
\centering
\begin{tabular}{|c|l|l|l|l|l|} 
\hline
\textbf{Label}   &   \textbf{Basis}   &   \textbf{dim} &   \textbf{Stabilizer}  & \textbf{Orbit} &   $\boldsymbol{(s,p,e,d)}$ \\	\hline
$\aa_0$       &	$XX$	&	1	&	4	&	3	&	(0,0,1,0)	\\	\hline
$\aa_1$       &	$XY$	&	1	&	2	&	6	&	(0,0,0,1)	\\	\hline
$\aa_2$       &	$XY, YX$	&	2	&	4	&	3	&	(0,0,0,2)	\\	\hline
$\aa_3$       &	$XX, YZ$	&	2	&	2	&	6	&	(0,0,1,1)	\\	\hline
$\aa_4$       &	$XX, YY$	&	2	&	4	&	3	&	(0,0,2,0)	\\	\hline
$\aa_5$       &	$XY, YZ$	&	2	&	2	&	6	&	(0,0,0,2)	\\	\hline
$\aa_6$   	&	$XX, YZ, ZY$	&	3	&	4	&	3	&	(0,0,1,2)	\\	\hline
$\aa_7$   	&	$XX, YY, ZZ$	&	3	&	12	&	1	&	(0,0,3,0)	\\	\hline
$\aa_8$   	&	$XX, XZ, IY$	&	3	&	1	&	12	&	(1,0,1,1)	\\	\hline
$\aa_9$   	&	$XY, XZ, IX$	&	3	&	2	&	6	&	(1,0,0,2)	\\	\hline
$\aa_{10}$    &	$XY, YZ, ZX$	&	3	&	6	&	2	&	(0,0,0,3)	\\	\hline
$\aa_{11}$	&	$XY, YX, YZ, IY$	&	4	&	1	&	12	&	(1,0,0,3)	\\	\hline
$\aa_{12}$	&	$XX, XY, YZ, IZ$	&	4	&	1	&	12	&	(1,0,1,2)	\\	\hline
$\aa_{13}$	&	$XX, YY, YZ, IX$	&	4	&	1	&	12	&	(1,0,2,1)	\\	\hline
$\aa_{14}$	&	$XX, YY, XY, YX, ZI, IZ$	&	6	&	4	&	3	&	(2,1,2,2)	\\	\hline
$\aa_{15}$	&	$XX, XY, XZ, IX, IY, IZ$	&	6	&	2	&	6	&	(3,0,1,2)	\\	\hline
$\aa_{16}$	&	$XY, YX, YZ, ZY, YI, IY$	&	6	&	4	&	3	&	(2,1,0,4)	\\	\hline
$\aa_{17}$	&	$XX, XY, ZX, ZY, YI, IZ$	&	6	&	2	&	6	&	(2,0,1,3)	\\	\hline
$\aa_{18}$	&	$XX, YY, XZ, ZY, XI, IY$	&	6	&	2	&	6	&	(2,0,2,2)	\\	\hline
$\aa_{19}$	&	$XX, XY, ZX, ZY, YZ, YI, IZ$	&	7	&	2	&	6	&	(2,0,1,4)	\\	\hline
$\aa_{20}$	&	$XX, YY, ZZ, YZ, ZY, XI, IX$	&	7	&	4	&	3	&	(2,1,3,2)	\\	\hline
$\aa_{21}$	&	$XX, YY, XY, YX, ZX, ZY, XI, YI, ZI, IZ$	&	10	&	2	&	6	&	(4,1,2,4)	\\	\hline
$\aa_{22}$	&	all Paulis except $II$
                                        &	15	&	12	&	1	&	(6,3,3,6)	\\	\hline
$\bb_0$	    &	$XI, IX$	&	2	&	4	&	3	&	(2,1,0,0)	\\	\hline
$\bb_1$   	&	$XX, XI, IX$	&	3	&	4	&	3	&	(2,1,1,0)	\\	\hline
$\bb_2$   	&	$XY, XZ, XI, IX$	&	4	&	2	&	6	&	(2,1,0,2)	\\	\hline
$\bb_3$   	&	$XI, YI, ZI, IX, IY, IZ$	&	6	&	12	&	1	&	(6,3,0,0)	\\	\hline
$\bb_4$   	&	$XX, XY, XZ, XI, IX, IY, IZ$	&	7	&	2	&	6	&	(4,1,1,2)	\\	\hline
$\cc_0$	    &	$XI$	&	1	&	2	&	6	&	(1,0,0,0)	\\	\hline
$\cc_1$	    &	$XY, XI$	&	2	&	1	&	12	&	(1,0,0,1)	\\	\hline
$\cc_2$	    &	$XX, XI$	&	2	&	2	&	6	&	(1,0,1,0)	\\	\hline
$\cc_3$	    &	$XI, IY$	&	2	&	2	&	6	&	(2,0,0,0)	\\	\hline
$\cc_4$   	&	$XY, XI, IY$	&	3	&	2	&	6	&	(2,0,0,1)	\\	\hline
$\cc_5$   	&	$XI, YI, ZI$	&	3	&	6	&	2	&	(3,0,0,0)	\\	\hline
$\cc_6$   	&	$XX, XY, XI, IZ$	&	4	&	1	&	12	&	(2,0,1,1)	\\	\hline
$\cc_7$   	&	$XI, IX, YI, ZI$	&	4	&	2	&	6	&	(4,1,0,0)	\\	\hline
\end{tabular}
\caption{List of all subalgebras of $\su(4)$ up to symmetry $S_3 \times \mathbb{Z}_2$. 
For each subalgebra, we have listed: a basis (over $\mathbb R$, after multiplication by $i$), its dimension, the order of the stabilizer, the order of the orbit under the action of $S_3 \times \mathbb{Z}_2$, and the invariants $s,p,e,d$. Note that the orders of all orbits add up to $202$.}
\label{tab:generator_list}
\centering
\end{table}

\clearpage{}

For each subalgebra of $\su(4)$, we list below a minimal set of generators and the isomorphism class of the Lie algebra.
Commuting direct summands are denoted with $\oplus$. Note also that $\so(4) \cong \su(2) \oplus \su(2)$.

\allowdisplaybreaks
\begin{align*}
\aa_0 &= \Lie{XX} \cong \uu(1), \\
\aa_1 &= \Lie{XY} \cong \uu(1), \\
\aa_2 &= \Lie{XY, YX} \cong \uu(1) \oplus \uu(1), \\
\aa_3 &= \Lie{XX, YZ} \cong \uu(1) \oplus \uu(1), \\
\aa_4 &= \Lie{XX, YY} \cong \uu(1) \oplus \uu(1), \\
\aa_5 &= \Lie{XY, YZ} \cong \uu(1) \oplus \uu(1), \\
\aa_6 &= \Lie{XX, YZ, ZY} \cong \uu(1) \oplus \uu(1) \oplus \uu(1), \\
\aa_7 &= \Lie{XX, YY, ZZ} \cong \uu(1) \oplus \uu(1) \oplus \uu(1), \\
\aa_8 &= \Lie{XX, XZ} \cong \su(2), \\
\aa_9 &= \Lie{XY, XZ} \cong \su(2), \\
\aa_{10} &= \Lie{XY, YZ, ZX} \cong \uu(1) \oplus \uu(1) \oplus \uu(1), \\
\aa_{11} &= \Lie{XY, YX, YZ} \cong \su(2) \oplus \uu(1), \\
\aa_{12} &= \Lie{XX, XY, YZ} \cong \su(2) \oplus \uu(1), \\
\aa_{13} &= \Lie{YY, YZ, XX} \cong \su(2) \oplus \uu(1), \\
\aa_{14} &= \Lie{XX, YY, XY} = \so(4), \\
\aa_{15} &= \Lie{XX, XY, XZ} \cong \su(2) \oplus \su(2), \\
\aa_{16} &= \Lie{XY, YX, YZ, ZY} = \so(4), \\
\aa_{17} &= \Lie{XX, XY, ZX} \cong \su(2) \oplus \su(2), \\
\aa_{18} &= \Lie{XX, XZ, YY, ZY} \cong \su(2) \oplus \su(2), \\
\aa_{19} &= \Lie{XX, XY, ZX, YZ} \cong \su(2) \oplus \su(2) \oplus \uu(1), \\
\aa_{20} &= \Lie{XX, YY, ZZ, ZY} \cong \su(2) \oplus \su(2) \oplus \uu(1), \\
\aa_{21} &= \Lie{XX, YY, XY, ZX} \cong \sp(2), \\
\aa_{22} &= \Lie{XX, XY, XZ, YX} = \su(4), \\
\bb_0 &= \Lie{XI, IX} \cong \uu(1) \oplus \uu(1), \\
\bb_1 &= \Lie{XX, XI, IX} \cong \uu(1) \oplus \uu(1) \oplus \uu(1), \\
\bb_2 &= \Lie{XY, XI, IX} \cong \su(2) \oplus \uu(1), \\
\bb_3 &= \Lie{XI, YI, IX, IY} \cong \su(2) \oplus \su(2), \\
\bb_4 &= \Lie{XX, XY, XI, IX} \cong \su(2) \oplus \su(2) \oplus \uu(1), \\
\cc_0 &= \Lie{XI} \cong \uu(1), \\
\cc_1 &= \Lie{XY, XI} \cong \uu(1) \oplus \uu(1), \\
\cc_2 &= \Lie{XX, XI} \cong \uu(1) \oplus \uu(1), \\
\cc_3 &= \Lie{XI, IY} \cong \uu(1) \oplus \uu(1), \\
\cc_4 &= \Lie{XY, XI, IY} \cong \uu(1) \oplus \uu(1) \oplus \uu(1), \\
\cc_5 &= \Lie{XI, YI, ZI} \cong \uu(1) \oplus \uu(1) \oplus \uu(1), \\
\cc_6 &= \Lie{XX, XY, XI} \cong \su(2) \oplus \uu(1), \\
\cc_7 &= \Lie{XI, YI, IX} \cong \su(2) \oplus \uu(1).
\end{align*}
Finally, note that
\begin{align*}
\bb_2 &= \aa_9 \oplus \Span\{XI\}, \\
\bb_4 &= \aa_{15} \oplus \Span\{XI\},
\end{align*}
are central extensions of $\aa$-type subalgebras. 
\clearpage{}

\subsection{List of all \texorpdfstring{$202$}{} subalgebras of \texorpdfstring{$\su(4)$}{su(4)}}\label{secres2}

Here, we group the $202$ subalgebras of $\su(4)$ into orbits of the symmetry group $S_3\times\mathbb{Z}_2$.
For each orbit, we provide its label and its size.

\allowdisplaybreaks
\begin{align*}
\aa_0 &: 3 , \; \{XX\}, \{YY\}, \{ZZ\}; \\ 
\aa_1 &: 6 , \; \{XY\}, \{XZ\}, \{YX\}, \{YZ\}, \{ZX\}, \{ZY\}; \\ 
\aa_2 &: 3 , \; \{XY, YX\}, \{XZ, ZX\}, \{YZ, ZY\}; \\ 
\aa_3 &: 6 , \; \{XX, YZ\}, \{XX, ZY\}, \{YY, XZ\}, \{YY, ZX\}, \{ZZ, XY\}, \{ZZ, YX\}; \\ 	
\aa_4 &: 3 , \; \{XX, YY\}, \{XX, ZZ\}, \{YY, ZZ\}; \\
\aa_5 &: 6 , \; \{XY, YZ\}, \{XZ, ZY\}, \{YX, XZ\}, \{YZ, ZX\}, \{ZX, XY\}, \{ZY, YX\}; \\ 
\aa_6 &:	3	, \; \{XY, YX, ZZ\}, \{XZ, ZX, YY\}, \{YZ, ZY, XX\};	\\	
\aa_7 &:	1	, \; \{XX, YY, ZZ\};	\\	
\aa_8 &:	12	, \; \{XX, XY, IZ\}, \{XX, XZ, IY\}, \{XX, YX, ZI\}, \{XX, ZX, YI\}, \\
            &   \qquad\;\; \{YY, YX, IZ\}, \{YY, YZ, IX\}, \{YY, XY, ZI\}, \{YY, ZY, XI\}, \\
            &   \qquad\;\; \{ZZ, ZX, IY\}, \{ZZ, ZY, IX\}, \{ZZ, XZ, YI\}, \{ZZ, YZ, XI\}; \\ 	
\aa_9 &:	6	, \; \{XY, XZ, IX\}, \{YX, YZ, IY\}, \{ZX, ZY, IZ\}, \\
            &   \qquad\, \{YX, ZX, XI\}, \{XY, ZY, YI\}, \{XZ, YZ, ZI\}; \\	
\aa_{10} &: 2 , \;	\{XY, YZ, ZX\}, \{XZ, ZY, YX\}; \\ 	
\aa_{11} &:	12	, \; \{XY, YX, XZ, IX\}, \{XZ, ZX, XY, IX\}, \{XY, YX, ZX, XI\}, \{XZ, ZX, YX, XI\}, \\
            &   \qquad\;\; \{YX, XY, YZ, IY\}, \{YZ, ZY, YX, IY\}, \{YX, XY, ZY, YI\}, \{YZ, ZY, XY, YI\}, \\
            &   \qquad\;\; \{ZX, XZ, ZY, IZ\}, \{ZY, YZ, ZX, IZ\}, \{ZX, XZ, YZ, ZI\}, \{ZY, YZ, XZ, ZI\}; \\ 	
\aa_{12} &:	12	, \; \{XX, XY, YZ, IZ\}, \{XX, XZ, ZY, IY\}, \{XX, YX, ZY, ZI\}, \{XX, ZX, YZ, YI\}, \\
            &   \qquad\;\; \{YY, YX, XZ, IZ\}, \{YY, YZ, ZX, IX\}, \{YY, XY, ZX, ZI\}, \{YY, ZY, XZ, XI\}, \\
            &   \qquad\;\; \{ZZ, ZX, XY, IY\}, \{ZZ, ZY, YX, IX\}, \{ZZ, XZ, YX, YI\}, \{ZZ, YZ, XY, XI\}; \\ 	
\aa_{13} &:	12	, \; \{XX, YY, YZ, IX\}, \{XX, YY, XZ, IY\}, \{XX, YY, ZY, XI\}, \{XX, YY, ZX, YI\},	\\
            &   \qquad\;\; \{XX, ZZ, ZY, IX\}, \{XX, ZZ, XY, IZ\}, \{XX, ZZ, YZ, XI\}, \{XX, ZZ, YX, ZI\},	\\
            &   \qquad\;\; \{YY, ZZ, ZX, IY\}, \{YY, ZZ, YX, IZ\}, \{YY, ZZ, XZ, YI\}, \{YY, ZZ, XY, ZI\};	\\
\aa_{14} &:	3	, \; \{XX, YY, XY, YX, ZI, IZ\}, \{XX, ZZ, XZ, ZX, YI, IY\}, \{YY, ZZ, YZ, ZY, XI, IX\};	\\	
\aa_{15} &:	6	, \; \{XX, XY, XZ, IX, IY, IZ\}, \{YY, YX, YZ, IX, IY, IZ\}, \{ZZ, ZX, ZY, IX, IY, IZ\},	\\
            &   \qquad\, \{XX, YX, ZX, XI, YI, ZI\}, \{YY, XY, ZY, XI, YI, ZI\}, \{ZZ, XZ, YZ, XI, YI, ZI\};	\\	
\aa_{16} &:	3	, \; \{XY, YX, XZ, ZX, XI, IX\}, \{YX, XY, YZ, ZY, YI, IY\}, \{ZX, XZ, ZY, YZ, ZI, IZ\};	\\	
\aa_{17} &:	6	, \; \{XX, XY, ZX, ZY, YI, IZ\}, \{YY, YX, ZY, ZX, XI, IZ\}, \{ZZ, ZX, YZ, YX, XI, IY\},	\\
            &   \qquad\, \{XX, YX, XZ, YZ, IY, ZI\}, \{YY, XY, YZ, XZ, IX, ZI\}, \{ZZ, XZ, ZY, XY, IX, YI\};	\\	
\aa_{18} &:	6	, \; \{XX, YY, XZ, ZY, XI, IY\}, \{XX, ZZ, XY, YZ, XI, IZ\}, \{YY, ZZ, YX, XZ, YI, IZ\},	\\
            &   \qquad\, \{XX, YY, ZX, YZ, IX, YI\}, \{XX, ZZ, YX, ZY, IX, ZI\}, \{YY, ZZ, XY, ZX, IY, ZI\};	\\	
\aa_{19} &:	6	, \; \{XX, XY, ZX, ZY, YZ, YI, IZ\}, \{XX, YX, XZ, YZ, ZY, IY, ZI\},	\\
            &   \qquad\, \{YY, YX, ZY, ZX, XZ, XI, IZ\}, \{YY, XY, YZ, XZ, ZX, IX, ZI\},	\\
            &   \qquad\, \{ZZ, ZX, YZ, YX, XY, XI, IY\}, \{ZZ, XZ, ZY, XY, YX, IX, YI\};	\\	
\aa_{20} &:	3	, \; \{XX, YY, ZZ, XY, YX, ZI, IZ\}, \{XX, YY, ZZ, XZ, ZX, YI, IY\}, \\
            &  \qquad\, \{XX, YY, ZZ, YZ, ZY, XI, IX\}; \\ 	
\aa_{21} &:	6	, \; \{XX, YY, XY, YX, ZX, ZY, XI, YI, ZI, IZ\}, \{XX, YY, XY, YX, XZ, YZ, IX, IY, ZI, IZ\}, \\
            &   \qquad\, \{XX, ZZ, XZ, ZX, YX, YZ, XI, ZI, YI, IY\}, \{XX, ZZ, XZ, ZX, XY, ZY, IX, IZ, YI, IY\}, \\
            &   \qquad\, \{YY, ZZ, YZ, ZY, XY, XZ, YI, ZI, XI, IX\}, \{YY, ZZ, YZ, ZY, YX, ZX, IY, IZ, XI, IX\}; \\  
\aa_{22} &:	1 , \; \{XX, YY, ZZ, XY, YX, XZ, ZX, YZ, ZY, XI, IX, YI, IY, ZI, IZ\}; \\
\bb_0 &: 3 , \; \{XI, IX\}, \{YI, IY\}, \{ZI, IZ\}; \\ 
\bb_1 &:	3	, \; \{XX, XI, IX\}, \{YY, YI, IY\}, \{ZZ, ZI, IZ\};	\\	
\bb_2 &:	6	, \; \{XY, XZ, XI, IX\}, \{YX, YZ, YI, IY\}, \{ZX, ZY, ZI, IZ\},	\\
            &   \qquad\, \{YX, ZX, XI, IX\}, \{XY, ZY, YI, IY\}, \{XZ, YZ, ZI, IZ\}; \\ 	
\bb_3 &:	1	, \; \{XI, YI, ZI, IX, IY, IZ\}; \\ 	
\bb_4 &:	6	, \; \{XX, XY, XZ, XI, IX, IY, IZ\}, \{XX, YX, ZX, IX, XI, YI, ZI\},  \\
            &   \qquad\, \{YY, YX, YZ, YI, IY, IX, IZ\}, \{YY, XY, ZY, IY, YI, XI, ZI\},  \\
            &   \qquad\, \{ZZ, ZX, ZY, ZI, IZ, IX, IY\}, \{ZZ, XZ, YZ, IZ, ZI, XI, YI\}; \\  
\cc_0 &: 6 , \; \{XI\}, \{YI\}, \{ZI\}, \{IX\}, \{IY\}, \{IZ\}; \\  
\cc_1 &: 12 , \; \{XY, XI\}, \{XZ, XI\}, \{YX, YI\}, \{YZ, YI\}, \{ZX, ZI\}, \{ZY, ZI\}, \\
            &  \qquad\;\; \{YX, IX\}, \{ZX, IX\}, \{XY, IY\}, \{ZY, IY\}, \{XZ, IZ\}, \{YZ, IZ\}; \\ 	
\cc_2 &: 6 , \; \{XX, XI\}, \{XX, IX\}, \{YY, YI\}, \{YY, IY\}, \{ZZ, ZI\}, \{ZZ, IZ\}; \\ 	
\cc_3 &: 6 , \; \{XI, IY\}, \{XI, IZ\}, \{YI, IX\}, \{YI, IZ\}, \{ZI, IX\}, \{ZI, IY\}; \\ 	
\cc_4 &:	6	, \; \{XY, XI, IY\}, \{YZ, YI, IZ\}, \{ZX, ZI, IX\}, \\	
            &   \qquad\, \{YX, YI, IX\}, \{ZY, ZI, IY\}, \{XZ, XI, IZ\}; \\
\cc_5 &:	2	, \; \{XI, YI, ZI\}, \{IX, IY, IZ\};	\\	
\cc_6 &:	12	, \; \{XX, XY, XI, IZ\}, \{XX, XZ, XI, IY\}, \{XX, YX, IX, ZI\}, \{XX, ZX, IX, YI\}, \\
            &   \qquad\;\; \{YY, YX, YI, IZ\}, \{YY, YZ, YI, IX\}, \{YY, XY, IY, ZI\}, \{YY, ZY, IY, XI\}, \\
            &   \qquad\;\; \{ZZ, ZX, ZI, IY\}, \{ZZ, ZY, ZI, IX\}, \{ZZ, XZ, IZ, YI\}, \{ZZ, YZ, IZ, XI\}; \\ 	
\cc_7 &:	6	, \; \{XI, IX, YI, ZI\}, \{YI, IY, XI, ZI\}, \{ZI, IZ, XI, YI\},	\\
            &   \qquad\, \{XI, IX, IY, IZ\}, \{YI, IY, IX, IZ\}, \{ZI, IZ, IX, IY\}.
\end{align*}

Adding up the orders of the orbits, we obtain a total of $127$ Lie algebras of type $\aa$, $19$ of type $\bb$, and $56$ of type $\cc$.

\clearpage

\subsection{Identifying the subalgebras with known spin systems}\label{secres3}

In Table \ref{tab:algebras_with_names} below, we identify some of our Lie algebras with the DLAs of known spin models. The listed generating sets are from Sect.\ \ref{secres1}. However, these generating sets are not unique, and they can be replaced with alternative generators that are used to generate the terms of the Hamiltonian; we call those conventional generators.

\begin{table}[ht]
\begin{center}
\begin{tabular}{|c|l|l|l|c|} 
\hline
\textbf{Label}   &   \textbf{Generating set}  & \textbf{Conventional generators} & \textbf{Model } 
\\	\hline
$\aa_0$       &	$XX$	& $ZZ$ & Ising model \cite{baxter1982exactly} 
\\	\hline
$\aa_1$       &	$XY$	& $XY$ &  Kitaev chain 
\\	\hline
$\aa_2$       &	$XY, YX$	&& Massless free fermion in a magnetic field 
\\	\hline
$\aa_3$       &	$XX, YZ$	& $ZZ, XY$ & Kitaev chain with \\&&& nearest neighbor Coulomb interaction
\\	\hline
$\aa_4$       &	$XX, YY$	&&  XY-model \cite{franchini2017introduction} / Massless free fermion  
\\	 \hline
$\aa_5$       &	$XY, YZ$	&& 
\\	\hline
$\aa_6$   	&	$XX, YZ, ZY$	& $ZZ, XY, YX$& Massless free fermion in a magnetic field with \\
&&& nearest neighbor Coulomb interaction
\\	\hline
$\aa_7$   	&	$XX, YY, ZZ$	&& Heisenberg model / XXZ Chain \cite{franchini2017introduction} 
\\	\hline
$\aa_8$   	&	$XX, XZ$	& $ZZ, IX$	& Transverse-field Ising model \cite{franchini2017introduction} 
\\	\hline
$\aa_9$   	&	$XY, XZ$	& $XY, IX$&   Kitaev chain in an X field 
\\	\hline
$\aa_{10}$    &	$XY, YZ, ZX$	&& Heisenberg model
\\	\hline
$\aa_{11}$	&	$XY, YX, YZ$	&$XX, YY, IY$& XY model in a Y field
\\	\hline
$\aa_{12}$	&	$XX, XY, YZ$	&& 
\\	\hline
$\aa_{13}$	&	$XX, YY, YZ$	&$XX, YY, IX$&XY-model in a longitudinal field \cite{franchini2017introduction} 
\\	\hline
$\aa_{14}$	&	$XX, YY, XY$	&& Transverse-field XY / Ising model \cite{franchini2017introduction} 
\\	\hline
$\aa_{15}$	&	$XX, XY, XZ$	& $ZZ, IX, IY, (IZ)$& Ising model in an arbitrary magnetic field
\\	\hline
$\aa_{16}$	&	$XY, YX, YZ, ZY$	& $XY, YX, IY, YI $& Kitaev chain in a Y field
\\	\hline
$\aa_{17}$	&	$XX, XY, ZX$	&$ZZ, IX, IY, (IZ)$& Ising model in an arbitrary magnetic field
\\	\hline
$\aa_{18}$	&	$XX, XZ, YY, ZY$	&$XX, YY, IY, XI, (ZI)$& XY model in an arbitrary field
\\	\hline
$\aa_{19}$	&	$XX, XY, ZX, YZ$ && 
\\	\hline
$\aa_{20}$	&	$XX, YY, ZZ, ZY$	&$XX, YY, ZZ, IX, XI$& XXZ chain in an X field \cite{franchini2017introduction} 
 \\	\hline
$\aa_{21}$	&	$XX, YY, XY, ZX$	&$XX, YY, IZ, YI, (IX)$& XY model in an arbitrary field
\\	\hline
$\aa_{22}$	&	$XX, XY, XZ, YX$ &$ZZ, XI, IY, IZ$& Ising model in an arbitrary field
\\	\hline
$\bb_0$	    &	$XI, IX$	& $ZI, IZ$ & Uncoupled spins
\\	\hline
$\bb_1$   	&	$XX, XI, IX$	& $ZZ, ZI, IZ$ & Ising model \cite{baxter1982exactly} 
\\	\hline
$\bb_2$   	&	$XY, XI, IX$	&& Kitaev chain in an X field 
\\	\hline
$\bb_3$   	&	$XI, YI, IX, IY$	&& Uncoupled spins 
\\	\hline
$\bb_4$   	&	$XX, XY, XZ, XI, IX, IY, IZ$	& $ZZ, IX, IY, IZ, ZI$& Ising model in an arbitrary field 
\\ 	\hline
\end{tabular}
\end{center}
\caption{Conventional spin models corresponding to the dynamical Lie algebras discussed in the main text. Terms in parentheses do not appear explicitly in the set of generators, but are generated from them.}
\label{tab:algebras_with_names}
\end{table}

\subsection{Extending subalgebras of \texorpdfstring{$\su(4)$}{su(4)} to \texorpdfstring{$\su(2^n)$}{su(2 n)}}\label{secext}

Starting from a subalgebra $\aa\subseteq\su(4)$, let $\aa(n)$ be the subalgebra of $\su(2^n)$ generated by the set
\begin{equation}\label{ext0}
\bigcup_{0\le k\le n-2} I^{\otimes k} \otimes \aa \otimes I^{\otimes (n-2-k)}.
\end{equation}
In particular, $\aa(2)=\aa$. By construction, we have two Lie algebra embeddings $\aa(n) \to \aa(n+1)$, given by appending $I$ in the last or first qubit,
and $\aa(n+1)$ is generated as a Lie algebra by the union of the two images:
\begin{equation}\label{ext01}
\aa(n+1) = \Lie{ (\aa(n) \otimes I) \cup (I \otimes \aa(n)) }.
\end{equation}
This allows us to determine the sequence of Lie algebras $\aa(2),\aa(3),\dots$ inductively.
In particular, if $\aa(n) = \su(2^n)$ for some $n=n_0$, then this is true for all $n\ge n_0$.
For instance, since
\begin{equation}\label{ext1}
\aa_{18}(3) = \aa_{19}(3) = \aa_{21}(3) = \aa_{22}(3) = \su(8), \quad
\aa_{12}(4) = \aa_{17}(4) = \su(16),
\end{equation}
we obtain that
\begin{equation}\label{ext2}
\aa_k(n) = \su(2^n), \qquad k=12,17,18,19,21,22, \quad n\ge 4.
\end{equation}

We also consider \emph{periodic boundary conditions}. For each $n\ge 2$, define the subalgebra of $\su(2^n)$:
\begin{equation}\label{ext5}
\aa^\circ(n) = \Lie{ \{A_i B_{i+1}, B_1 A_n \,|\, AB \in \aa, \; 1\le i\le n-1\} }.
\end{equation}
Since
\begin{equation}\label{ext6}
\aa(n) = \Lie{ \{A_i B_{i+1} \,|\, AB \in \aa, \; 1\le i\le n-1\} },
\end{equation}
it is obvious that $\aa(n) \subseteq \aa^\circ(n)$.
Introduce the cyclic shift operator $\tau_n\colon\su(2^n)\to\su(2^n)$, which acts on Pauli strings as
\begin{equation}\label{taun}
\tau_n(P^1 \otimes P^2 \otimes\cdots\otimes P^n) = P^2 \otimes\cdots\otimes P^n \otimes P^1, \qquad P^j\in\{I,X,Y,Z\}
\end{equation}
(where the superscipts are indices not powers).
Then $\tau_n$ is a Lie algebra automorphism, and $\tau_n \aa^\circ(n) = \aa^\circ(n)$. In particular, this implies that
$\tau_n \aa(n) \subseteq \aa^\circ(n)$.
Note that, by definition, $\aa^\circ(n)$ is generated as a Lie algebra by the union of $\aa(n)$ and $\tau_n \aa(n)$:
\begin{equation}\label{taun1}
\aa^\circ(n) = \Lie{ \aa(n) \cup \tau_n\aa(n) }.
\end{equation}

Another case we will investigate is when our subalgebras of $\su(2^n)$ are \emph{permutation invariant}, i.e., invariant under the action of the symmetric group $S_n$ that permutes the $n$ qubits.
We define
\begin{equation}\label{ext7}
\aa^\pi(n) = \Lie{ \{A_i B_j \,|\, AB \in \aa, \; 1\le i\ne j \le n\} }.
\end{equation}
Note that, in particular, $\aa^\circ(n) \subseteq \aa^\pi(n)$. Moreover, without loss of generality, we can assume that the generating subalgebra $\aa\subseteq\su(4)$ is itself invariant under $S_2$, i.e.,
under the flip of the two qubits. In other words, we have:
\allowdisplaybreaks
\begin{align*}
\aa_1^\pi(n) &= \aa_2^\pi(n), \\
\aa_3^\pi(n) &= \aa_6^\pi(n), \\
\aa_5^\pi(n) &=\aa_{11}^\pi(n)=\aa_{16}^\pi(n), \\
\aa_8^\pi(n) &\cong \aa_{14}^\pi(n), \\
\aa_9^\pi(n) &=\bb_2^\pi(n) \cong \aa_{16}^\pi(n), \\
\aa_{13}^\pi(n) &= \aa_{20}^\pi(n), \\
\aa_k^\pi(n) &= \bb_4^\pi(n) =\su(2^n), \qquad k=10,12,15,17,18,19,21,22.
\end{align*}
Thus, we only need to determine $\aa_k^\pi(n)$ for $k=0,2,4,6,7,14,16,20$ and $\bb_l^\pi(n)$ for $l=0,1,3$.

\subsection{Subalgebras of \texorpdfstring{$\su(8)$}{su(8)}}\label{secsu8}
For completeness and for later use, we list a linear basis (over $\mathbb R$, after multiplication by $i$) for each of the following subalgebras of $\su(8)$.

\medskip
\textbf{Open case:}
\allowdisplaybreaks
\begin{align*}
\aa_0(3) &: \{ IXX,XXI \}, \\
\aa_1(3) &: \{ IXY,XYI,XZY \}, \\
\aa_2(3) &: \{ IXY,IYX,XYI,XZY,YXI,YZX \}, \\
\aa_3(3) &: \{ IXX,IYZ,XXI,XZZ,YIY,YXZ,YYX,YZI,ZIZ,ZXY \}, \\
\aa_4(3) &: \{ IXX,IYY,XXI,XZY,YYI,YZX \}, \\
\aa_5(3) &: \{ IXY,IYZ,XYI,XZY,YIX,YXZ,YYY,YZI,ZIY,ZYX \}, \\
\aa_6(3) &: \{ IXY,IYX,IZZ,XIX,XXZ,XYI,XZY,YIY,YXI,YYZ,YZX,ZIZ,ZXX,ZYY,ZZI \}, \\
\aa_7(3) &: \{ IXX,IYY,IZZ,XIX,XXI,XYZ,XZY,YIY,YXZ,YYI,YZX,ZIZ,ZXY,ZYX,ZZI \}, \\
\aa_8(3) &: \{ IIY,IXX,IXZ,IYI,IZX,IZZ,XXI,XYX,XYZ,XZI \}, \\
\aa_9(3) &: \{ IIX,IXI,IXY,IXZ,XYI,XYY,XYZ,XZI,XZY,XZZ \}, \\
\aa_{10}(3) &: \{ IXY,IYZ,IZX,XIZ,XXX,XYI,XZY,YIX,YXZ,YYY,YZI,ZIY,ZXI,ZYX,ZZZ \}, \\
\aa_{11}(3) &: \{ IIX,IXI,IXY,IXZ,IYX,IZX,XIY,XIZ,XXX,XYI,XYY,XYZ,XZI,XZY, \\
&\quad XZZ,YXI,YYX,YZX,ZIX,ZXY,ZXZ \}, \\
\aa_{12}(3) &: \{ IIX,IIY,IIZ,IXX,IXY,IXZ,IYX,IYY,IYZ,IZI,XII,XXI,XYI,XZX,XZY, \\
&\quad XZZ,YIX,YIY,YIZ,YXX,YXY,YXZ,YYX,YYY,YYZ,YZI,ZIX,ZIY,ZIZ,ZXX, \\
&\quad ZXY,ZXZ,ZYX,ZYY,ZYZ,ZZI \}, \\
\aa_{13}(3) &: \{ IIX,IXI,IXX,IYY,IYZ,IZY,IZZ,XII,XIX,XXI,XYY,XYZ,XZY,XZZ,YIY, \\
&\quad YIZ,YXY,YXZ,YYI,YYX,YZI,YZX,ZIY,ZIZ,ZXY,ZXZ,ZYI,ZYX,ZZI,ZZX \}, \\
\aa_{14}(3) &: \{ IIZ,IXX,IXY,IYX,IYY,IZI,XXI,XYI,XZX,XZY,YXI,YYI,YZX,YZY,ZII \}, \\
\aa_{15}(3) &: \{ IIX,IIY,IIZ,IXI,IXX,IXY,IXZ,IYI,IYX,IYY,IYZ,IZI,IZX,IZY,IZZ,XIX, \\
&\quad XIY,XIZ,XXI,XXX,XXY,XXZ,XYI,XYX,XYY,XYZ,XZI,XZX,XZY,XZZ \}, \\
\aa_{16}(3) &: \{ IIX,IXI,IXY,IXZ,IYX,IZX,XII,XIY,XIZ,XXX,XYI,XYY,XYZ,XZI, \\
&\quad XZY,XZZ,YIX,YXI,YXY,YXZ,YYX,YZX,ZIX,ZXI,ZXY,ZXZ,ZYX,ZZX \}, \\
\aa_{17}(3) &: \{ IIZ,IXI,IXX,IXY,IYI,IYX,IYY,IZI,IZX,IZY,XIZ,XXI,XXX,XXY, \\
&\quad XYI,XYX,XYY,XZI,XZX,XZY,YII,YIX,YIY,YXZ,YYZ,YZZ,ZIZ,ZXI, \\
&\quad ZXX,ZXY,ZYI,ZYX,ZYY,ZZI,ZZX,ZZY \}, \\
\aa_{20}(3) &: \{ IIZ,IXX,IXY,IYX,IYY,IZI,IZZ,XIX,XIY,XXI,XXZ,XYI,XYZ,XZX,XZY, \\
&\quad YIX,YIY,YXI,YXZ,YYI,YYZ,YZX,YZY,ZII,ZIZ,ZXX,ZXY,ZYX,ZYY,ZZI \}.
\end{align*}

\textbf{Periodic case:}
\begin{align*}
\aa_0^\circ(3) &: \{ IXX,XIX,XXI \}, \\
\aa_1^\circ(3) &: \{ IXY,XYI,XZY,YIX,YXZ,ZYX \}, \\
\aa_2^\circ(3) &: \{ IXY,IYX,XIY,XYI,XYZ,XZY,YIX,YXI,YXZ,YZX,ZXY,ZYX \}, \\
\aa_3^\circ(3) &: \{ IIX,IXI,IXX,IYY,IYZ,IZY,IZZ,XII,XIX,XXI,XYY,XYZ,XZY,XZZ,YIY, \\
&\quad YIZ,YXY,YXZ,YYI,YYX,YZI,YZX,ZIY,ZIZ,ZXY,ZXZ,ZYI,ZYX,ZZI,ZZX \}, \\
\aa_4^\circ(3) &: \{ IXX,IYY,IZZ,XIX,XXI,XYZ,XZY,YIY,YXZ,YYI,YZX,ZIZ,ZXY,ZYX,ZZI \}, \\
\aa_6^\circ(3) &: \{ IIZ,IXX,IXY,IYX,IYY,IZI,IZZ,XIX,XIY,XXI,XXZ,XYI,XYZ,XZX,XZY, \\
&\quad YIX,YIY,YXI,YXZ,YYI,YYZ,YZX,YZY,ZII,ZIZ,ZXX,ZXY,ZYX,ZYY,ZZI \}, \\
\aa_8^\circ(3) &: \{ IIY,IXX,IXZ,IYI,IYY,IZX,IZZ,XIX,XIZ,XXI,XXY,XYX,XYZ,XZI,XZY, \\
&\quad YII,YIY,YXX,YXZ,YYI,YZX,YZZ,ZIX,ZIZ,ZXI,ZXY,ZYX,ZYZ,ZZI,ZZY \}, \\
\aa_9^\circ(3) &: \{ IIX,IXI,IXY,IXZ,XII,XYI,XYY,XYZ,XZI,XZY,XZZ,YIX,YXY,YXZ,YYX, \\
&\quad YZX,ZIX,ZXY,ZXZ,ZYX,ZZX \}, \\
\aa_{11}^\circ(3) &: \{ IIX,IXI,IXY,IXZ,IYX,IZX,XII,XIY,XIZ,XXX,XYI,XYY,XYZ,XZI,XZY, \\
&\quad XZZ,YIX,YXI,YXY,YXZ,YYX,YZX,ZIX,ZXI,ZXY,ZXZ,ZYX,ZZX \}, \\
\aa_{14}^\circ(3) &: \{ IIZ,IXX,IXY,IYX,IYY,IZI,IZZ,XIX,XIY,XXI,XXZ,XYI,XYZ,XZX,XZY, \\
&\quad YIX,YIY,YXI,YXZ,YYI,YYZ,YZX,YZY,ZII,ZIZ,ZXX,ZXY,ZYX,ZYY,ZZI \}.
\end{align*}
Moreover, 
\begin{equation*}
\aa_{k}^\circ(3) = \aa_{k}(3), \qquad k=5,7,10,13,16,20,
\end{equation*}
and
\begin{equation*}
\aa_{12}^\circ(3) = \aa_{15}^\circ(3) = \aa_{17}^\circ(3) = \su(8).
\end{equation*}


\subsection{Subalgebras of \texorpdfstring{$\su(2^n)$}{su(2 n)}}\label{secsubsu}

For convenience, here we repeat the statement of Theorem \ref{the:classification}, with the additional information of the dimensions of the Lie algebras (cf.\ \eqref{eqdims}).
The proof of the theorem is given in Sect.\ \ref{secproof} below.

\begin{align*}
\aa_0(n) &= \Span\{X_j X_{j+1}\}_{1\le j\le n-1} \cong \uu(1)^{\oplus (n-1)}, \quad \dim=n-1, \\
\aa_1(n) &= \Span\{X_i Z_{i+1} \cdots Z_{j-1} Y_j\}_{1\le i < j\le n} \cong \so(n), \quad \dim=\frac{n(n-1)}2, \\
\aa_2(n) &= \Span\{X_i Z_{i+1} \cdots Z_{j-1} Y_j\}_{1\le i < j\le n} \oplus \Span\{Y_i Z_{i+1} \cdots Z_{j-1} X_j\}_{1\le i < j\le n} \\
        &\cong \so(n) \oplus \so(n), \quad \dim=n(n-1), \\
\aa_3(n) &\cong \begin{cases} 
        \so(2^{n-2})^{\oplus 4}, \quad \dim = 2^{n-1}(2^{n-2}-1), & n\equiv 0 \mod 8, \\
        \so(2^{n-1}), \qquad\hspace{1pt} \dim=2^{n-2}(2^{n-1}-1), & n\equiv \pm1 \mod 8, \\
        \su(2^{n-2})^{\oplus2}, \quad \dim=2^{2n - 3} - 2, & n\equiv \pm2 \mod 8, \\
        \sp(2^{n-2}), \qquad\hspace{1pt} \dim=2^{n - 2}(2^{n - 1} + 1), & n\equiv \pm3 \mod 8, \\
        \sp(2^{n-3})^{\oplus4}, \quad \dim=2^{n - 1}(2^{n - 2} + 1), & n\equiv 4 \mod 8,     
        \end{cases} \\
\aa_4(n) &\cong \aa_2(n), \\
\aa_5(n) & \cong \begin{cases} 
        \so(2^{n-2})^{\oplus 4}, \quad \dim = 2^{n-1}(2^{n-2}-1), & n\equiv 0 \mod 6, \\         
        \so(2^{n-1}), \qquad\hspace{1pt} \dim = 2^{n-2}(2^{n-1}-1), & n\equiv \pm1 \mod 6, \\
        \su(2^{n-2})^{\oplus2}, \quad \dim=2^{2n - 3} - 2, & n\equiv \pm2 \mod 6, \\
        \sp(2^{n-2}), \qquad\hspace{1pt} \dim=2^{n - 2}(2^{n - 1} + 1), & n\equiv 3 \mod 6,
        \end{cases} \\
\aa_6(n) &\cong \aa_7(n) \cong \aa_{10}(n) \cong \begin{cases} 
        \su(2^{n-1}), \qquad \dim = 2^{2n-2} - 1, & n \;\; \mathrm{odd}, \\
        \su(2^{n-2})^{\oplus 4}, \quad \dim = 2^{2n-2} - 4, & n \ge 4\;\; \mathrm{even},
        \end{cases} \\
\aa_8(n) &\cong \so(2n-1), \quad \dim=(n-1)(2n-1), \\
\aa_9(n) &\cong \sp(2^{n-2}), \quad \dim=2^{n-2}(2^{n-1}+1), \\
\aa_{11}(n) &= \aa_{16}(n) = \so(2^n), \quad \dim = 2^{n-1}(2^n-1), \quad n\ge 4, \\
\aa_k(n) &= \su(2^n), \quad \dim=2^{2n}-1, \quad k=12,17,18,19,21,22, \;\; n\ge 4, \\
\aa_{13}(n) &= \aa_{20}(n) \cong \aa_{15}(n) \cong \su(2^{n-1}) \oplus \su(2^{n-1}), \quad \dim=2^{2n-1}-2, \\
\aa_{14}(n) &\cong \so(2n), \quad \dim=n(2n-1), \\
\bb_0(n) &= \Span\{X_i\}_{1\le i\le n} \cong \uu(1)^{\oplus n}, \quad \dim=n, \\
\bb_1(n) &= \Span\{X_i, X_j X_{j+1}\}_{1\le i\le n, \; 1\le j\le n-1} \cong \uu(1)^{\oplus (2n-1)}, \quad \dim=2n-1, \\
\bb_2(n) &= \aa_9(n) \oplus \Span\{X_1\} \cong \sp(2^{n-2}) \oplus \uu(1), \quad \dim=2^{n-2}(2^{n-1}+1)+1, \\
\bb_3(n) &= \Span\{X_i, Y_i, Z_i\}_{1\le i\le n} \cong \su(2)^{\oplus n}, \quad \dim=3n, \\
\bb_4(n) &= \aa_{15}(n) \oplus \Span\{X_1\} \cong \su(2^{n-1}) \oplus \su(2^{n-1}) \oplus \uu(1), \quad \dim=2^{2n-1}-1.
\end{align*}

\clearpage
\section{Proofs}\label{secproof}

This section contains detailed proofs of Theorems \ref{the:classification}, \ref{the:classification-p}, and \ref{the:classification-s}.
The proof of Theorem \ref{the:classification} occupies Sect.\ \ref{seciso}--\ref{secgkn}; its plan is outlined in Sect.\ \ref{secout} below.
The proofs of Theorems \ref{the:classification-p} and \ref{the:classification-s} utilize the results of Theorem \ref{the:classification}, 
and are given in Sect.\ \ref{secper} and \ref{app:sym}, respectively. For an index of each proof for each algebra, see Table~\ref{tab:proof_index}.

\begin{table}[htb!]
\begin{center}
\begin{tabular}{|c|l|l|c|c|} 
\hline
\textbf{Label}   &   \textbf{Generators}  &  \textbf{Scaling} &  \textbf{Isomorphism}  & \textbf{Reason}
\\	\hline
$\aa_0$       &	$XX$	& $O(n)$ & $\mathfrak{u}(1)^{\oplus(n-1)}$ & Sect.\ \ref{secout}
\\	\hline
$\aa_1$       &	$XY$	& $O(n^2)$ & $\so(n)$ & Sect.\ \ref{secfrus}, frustration graph
\\	\hline
$\aa_2$       &	$XY, YX$	& $O(n^2)$ & $\so(n) \oplus \so(n)$ & Sect.\ \ref{secfrus}, frustration graph
\\	\hline
$\aa_3$       &	$XX, YZ$	& $O(n)$ & $n$-dependent & Lemmas \ref{lemgkn4o}, \ref{lemgkn4e}
\\	\hline
$\aa_4$       &	$XX, YY$	& $O(n^2)$ & $\aa_2$ & Sect.\ \ref{seciso}, inclusion 
\\	 \hline
$\aa_5$       &	$XY, YZ$	& $O(4^n)$ & $n$-dependent & Lemmas \ref{lemgkn5o}, \ref{lemgkn5e}
\\	\hline
$\aa_6$   	&	$XX, YZ, ZY$ & $O(4^n)$ & $\aa_7$ & Sect.\ \ref{seciso}, inclusion
\\	\hline
$\aa_7$   	&	$XX, YY, ZZ$ & $O(4^n)$ & $\aa_6$ & Lemma \ref{lemgkn3}
\\	\hline
$\aa_8$   	&	$XX, XZ$	& $O(n^2)$ & $\so(2n-1)$ & Sect.\ \ref{secfrus}, frustration graph 
\\	\hline
$\aa_9$   	&	$XY, XZ$	& $O(4^n)$ & $\sp(2^{n-2})$ & Lemma \ref{lemgkn2}, fixed point under involution
\\	\hline
$\aa_{10}$    &	$XY, YZ, ZX$	& $O(4^n)$ & $\aa_6$ & Sect.\ \ref{seciso}, inclusion
\\	\hline
$\aa_{11}$	&	$XY, YX, YZ$	& $O(4^n)$ & $\aa_{16}$ & Sect.\ \ref{seciso}, inclusion
\\	\hline
$\aa_{12}$	&	$XX, XY, YZ$	& $O(4^n)$ & $\su(2^n)$ & Sect.\ \ref{secext}, explicit for $n=4$
\\	\hline
$\aa_{13}$	&	$XX, YY, YZ$	& $O(4^n)$ & $\aa_{20}$ & Lemma \ref{lemgkn1}, fixed point under involution
\\	\hline
$\aa_{14}$	&	$XX, YY, XY$	& $O(n^2)$ & $\so(2n)$ & Sect.\ \ref{secfrus}, frustration graph
\\	\hline
$\aa_{15}$	&	$XX, XY, XZ$	& $O(4^n)$ & $\aa_{13}$ & Lemma \ref{lemgkn1}, fixed point under involution
\\	\hline
$\aa_{16}$	&	$XY, YX, YZ, ZY$	& $O(4^n)$ & $\so(2^n)$ & Lemma \ref{lem7}, fixed point under involution
\\	\hline
$\aa_{17}$	&	$XX, XY, ZX$ & $O(4^n)$ & $\su(2^n)$ & Sect.\ \ref{secext}, explicit for $n=4$
\\	\hline
$\aa_{18}$	&	$XX, XZ, YY, ZY$	& $O(4^n)$ & $\su(2^n)$ & Sect.\ \ref{secext}, explicit for $n=3$
\\	\hline
$\aa_{19}$	&	$XX, XY, ZX, YZ$ & $O(4^n)$ & $\su(2^n)$ & Sect.\ \ref{secext}, explicit for $n=3$
\\	\hline
$\aa_{20}$	&	$XX, YY, ZZ, ZY$	& $O(4^n)$ & $\aa_{13}$ & Sect.\ \ref{seciso}, inclusion
 \\	\hline
$\aa_{21}$	&	$XX, YY, XY, ZX$	& $O(4^n)$ & $\su(2^n)$ & Sect.\ \ref{secext}, explicit for $n=3$
\\	\hline
$\aa_{22}$	&	$XX, XY, XZ, YX$ & $O(4^n)$ & $\su(2^n)$ & Sect.\ \ref{secext}, explicit for $n=3$
 \\	\hline
$\bb_{0}$	&	$XI,IX$ & $O(n) $& $\uu(1)^{\oplus n}$ & Sect.\ \ref{secout}
 \\	\hline
$\bb_{1}$	&	$XX, XI, IX$ & $O(n) $ & $\uu(1)^{\oplus (2n-1)}$ & Sect.\ \ref{secout}
 \\	\hline
$\bb_{2}$	&	$XY, XI, IX$ & $O(4^n)$  & $\sp(2^{n-2})\oplus \uu(1)$  & Sect.\ \ref{secout}, central extension
 \\	\hline
$\bb_{3}$	&	$XI, YI, IX, IY$ & $O(n)$ & $\su(2)^{\oplus n}$ & Sect.\ \ref{secout}
 \\	\hline
$\bb_{4}$	&	$XX, XY, XZ, XI, IX, IY, IZ$ & $O(4^n)$ & $\su(2^{n-1})\oplus \su(2^{n-1}) \oplus \uu(1) $ & Sect.\ \ref{secout}, central extension
 \\	\hline
\end{tabular}
\end{center}
\caption{Proofs and where to find them.}
\label{tab:proof_index}
\end{table}

\subsection{Plan of the proof of Theorem \ref{the:classification}}\label{secout}

Our starting point is the list of subalgebras of $\su(4)$ from Sect.\ \ref{secres1}: 
\begin{equation}\label{eqout1}
\aa_k, \bb_l \subseteq\su(4), \qquad 0\le k\le 22, \;\; 0\le l\le 4,
\end{equation}
and the goal is to determine (up to isomorphism) their extensions as subalgebras of $\su(2^n)$:
\begin{equation}\label{eqout2}
\aa_k(n), \bb_l(n) \subseteq\su(2^n), \qquad 3\le n, \;\; 0\le k\le 22, \;\; 0\le l\le 4,
\end{equation}
which are defined in Sect.\ \ref{secext}. The answer is presented in Theorem \ref{the:classification}, and in more detail, in Sect.\ \ref{secsubsu} above.
Here we outline the proof, which consists of multiple parts.
We divide the set of Lie algebras \eqref{eqout2} into three classes: linear, quadratic, and exponential, according to the anticipated growth of their dimension. 

\medskip
\textbf{Linear case: $\aa_0(n)$, $\bb_l(n)$ $(l=0,1,3)$.}
The linear case is obvious.
Indeed, the Lie algebras $\aa_0(n)$, $\bb_0(n)$, and $\bb_1(n)$ are Abelian (i.e., have identically zero Lie brackets), because all of their generators commute with each other.
The Lie algebra $\bb_3(n) \cong \su(2)^{\oplus n}$ is a direct sum of $n$ commuting copies of $\su(2)$, since its generators split into $n$ groups commuting with each other (acting independently on each qubit).

\medskip
\textbf{Quadratic case: $\aa_k(n)$ $(k=1,2,4,8,14)$.}
These Lie algebras are determined by using the \emph{frustration graphs} of their generators; see Sect.\ \ref{secfrus} and especially Lemma \ref{lemfrus1}. 
Note that $\aa_2(n) \cong \aa_4(n)$ by Lemma \ref{leminc3}.

\medskip
\textbf{Exponential case: $\aa_k(n)$, $\bb_l(n)$ $(k=3,5,6,7,9{-}13,15{-}22, \; l=2,4)$.}
First, recall that in Sect.\ \ref{secext}, we have already found that 
\begin{equation}\label{eqout3}
\aa_k(n) = \su(2^n), \qquad k=12,17,18,19,21,22, \;\; n\ge 4.
\end{equation}
Second, we observe that $\bb_2(n) = \aa_9(n) \oplus \Span\{X_1\}$ and $\bb_4(n) = \aa_{15}(n) \oplus \Span\{X_1\}$, because their generators consist of a central element $X_1$ (commuting with all other generators) and the generators of $\aa_9(n)$ or $\aa_{15}(n)$, respectively.
Third, in Sect.\ \ref{seciso}, we find equalities and isomorphisms among some of the Lie algebras $\aa_k(n)$.
Namely, 
\begin{equation}\label{eqout4}
\aa_6(n) \cong \aa_7(n) \cong \aa_{10}(n), \qquad
\aa_{11}(n) = \aa_{16}(n), \qquad
\aa_{13}(n) = \aa_{20}(n), \qquad n\ge4.
\end{equation}
Thus, we are left to investigate the Lie algebras $\aa_k(n)$ for $k=3,5,7,9,13,15,16$.

\medskip
\textbf{Strategy for $\aa_k(n)$ $(k=3,5,7,9,13,15,16)$.}
The strategy in the remaining exponential cases is as follows.
\begin{enumerate}
\item 
For each of our Lie algebras $\s=\aa_k(n)$, we find its \emph{stabilizer} $\Stab(\s)$,
which is defined as the set of all Pauli strings $\in \mathcal{P}_n$ that commute with every element of $\s$. This can be done explicitly, because
the stabilizer is determined only from the generators of $\s$; see Proposition \ref{pstab1} in Sect.\ \ref{secstab}.
\item 
By definition, $\s$ commutes with all elements of its stabilizer $\Stab(\s)$; hence, it is contained in the 
\emph{centralizer} of $\Stab(\s)$ in $\su(2^n)$, which we denote $\su(2^n)^{\Stab(\s)}$. We can reduce the Lie subalgebra $\su(2^n)^{\Stab(\s)}$
further by factoring all elements of the center of $\Stab(\s)$, which will become central in it, because $\s$ has a trivial center by Lemma \ref{lemp3}. 
This results in a Lie algebra denoted $\g_k(n)$ when $\s=\aa_k(n)$;
these are listed explicitly in \eqref{listgkn1}--\eqref{listgkn7}.
\item 
By the above construction, we have $\s \subseteq \g_k(n)$. However, equality does not hold in all cases. We improve the upper bound for $\s$
by finding an \emph{involution} $\theta_k$ of $\g_k(n)$ such that all elements of $\s$ are fixed under $\theta_k$, i.e., $\theta_k(a)=a$ for all $a\in\s$ 
(see Sect.\ \ref{secpauli3} for a refresher on involutions). The last condition can be checked only on the generators of $\s$, and the details are carried out in 
Sect.\ \ref{secup} (Theorem \ref{thmtheta} and Lemmas \ref{lemtheta1}, \ref{lemtheta2}, \ref{lemtheta3}).
Thus, we have the upper bound $\s \subseteq \g_k(n)^{\theta_k}$, where the superscript $\theta_k$ indicates fixed points under $\theta_k$.
\item 
Then, in Sect.\ \ref{seclow}, we establish a lower bound for $\s$, i.e., we prove that the upper bound is exact: $\aa_k(n) = \g_k(n)^{\theta_k}$. 
The main idea is to start with an arbitrary Pauli string $\in i\mcP_n\cap\g_k(n)^{\theta_k}$ and use suitable commutators with elements of $\aa_k(n)$
to produce a Pauli string $\in i\mcP_n\cap\g_k(n)^{\theta_k}$ with $I$ in one of its positions. Erasing the $I$ gives an element of $\g_k(n-1)^{\theta_k}$,
which by induction is in $\aa_k(n-1)$. The specific details are
broken into a sequence of lemmas. The cases $k=3,5,7$, $k=9$, $k=13$, $k=15$, and $k=16$
are treated in Lemmas \ref{lem4}, \ref{lem5}, \ref{lem6}, \ref{lem15}, and \ref{lem7}, respectively.
As a consequence, since $\g_{16}(n) = \su(2^n)$ by \eqref{listgkn6}, and $\theta_{16}(g) = -g^T$ by \eqref{thetaq11}, we obtain that $\aa_{16}(n) = \so(2^n)$.
%
\item 
Finally, in Sect.\ \ref{secgkn}, we identify the Lie algebras $\g_k(n)^{\theta_k}$ with those from Theorem \ref{the:classification}.
The idea is to apply a suitable unitary transformation that brings the stabilizer $\Stab(\s)$ to a more convenient form (such transformations and their effect on the
fixed points of involutions are reviewed in Sect.\ \ref{secpauli3}). For example, $\Stab(\aa_{13}(n)) = \{I^{\otimes n}, X^{\otimes n}\}$ and we can bring $X^{\otimes n}$
to $X_1$ with a unitary transformation, after which it is easy to determine the centralizer $\su(2^n)^{X_1}$.
This is carried out in Lemmas \ref{lemgkn1}, \ref{lemgkn2}, and \ref{lemgkn3} for $k=13,15$, $k=9$, and $k=7$, respectively.
The more complicated cases $k=3$ and $k=5$ are further broken down to $n$ odd and $n$ even; see Lemmas \ref{lemgkn4o}, \ref{lemgkn4e}, \ref{lemgkn5o}, \ref{lemgkn5e}.
Taken all together, this completes the proof of Theorem \ref{the:classification}.
\end{enumerate}

\subsection{Inclusions and isomorphisms}\label{seciso}

There are several obvious inclusions among the Lie algebras $\aa_k$ ($k=0{-}11,13{-}16,20$),
which extend to the corresponding subalgebras of $\su(2^n)$ due to the following trivial observation.

\begin{lemma}\label{leminc1}
If $\aa\subset\bb\subseteq\su(4)$, then $\aa(n) \subseteq \bb(n)$, $\aa^\circ(n) \subseteq \bb^\circ(n)$ and $\aa^\pi(n) \subseteq \bb^\pi(n)$ for all $n\ge2$.
Furthermore, if $\aa(n_0) = \bb(n_0)$ for some $n_0$, then $\aa(n) = \bb(n)$ for all $n\ge n_0$.
\end{lemma}
\begin{proof}
The first claim follows from the definitions of the Lie algebras $\aa(n)$, $\aa^\circ(n)$, and $\aa^\pi(n)$; see \eqref{ext6}, \eqref{ext5}, \eqref{ext7}, respectively.
The second claim follows from the inductive construction of the Lie algebras $\aa(n)$; see \eqref{ext01}.
\end{proof}

For the rest of this subsection, we focus on the open case; we will treat the periodic (closed) case and the permutation-invariant case later.
By comparing the bases of the subalgebras $\aa_k\subset\su(4)$ listed in Sect.\ \ref{secres1}, we note the following inclusions:
\allowdisplaybreaks
\begin{align*}
\aa_0 &\subset \aa_3 \subset \aa_6 \subset \aa_{20}, & 
\aa_1 &\subset \aa_2 \subset \aa_{11} \subset \aa_{16}, \\
\aa_0 &\subset \aa_4 \subset \aa_7, &
\aa_1 &\subset \aa_2 \subset \aa_{14}, \\
\aa_0 &\subset \aa_4 \subset \aa_{13} \subset \aa_{20}, &
\aa_1 &\subset \aa_5 \subset \aa_{10}, \quad \text{and} \quad\aa_1 \subset \aa_5 \subset \aa_{16},\\
\aa_0 &\subset \aa_8 \subset \aa_{15}, &
\aa_1 &\subset \aa_9 \subset \aa_{15}.
\end{align*}
The above chains of inclusions are maximal, i.e., cannot be extended further. For any pair $\aa_k\subset\aa_l$, we get an inclusion $\aa_k(n)\subseteq\aa_l(n)$ for all $n\ge3$.
However, even though $\aa_5 = \aa_{10} \cap \aa_{16}$, we only have $\aa_5(n) \subseteq  \aa_{10}(n) \cap \aa_{16}(n)$ and not necessarily an equality.
In fact, one checks that $IZX \in \aa_{10}(3) \cap \aa_{16}(3)$ but $IZX \notin \aa_5(3)$.
In the next lemma, we present two consequences of the above inclusions.

\begin{lemma}\label{leminc2}
We have:
\begin{align}
\aa_{11}(n) = \aa_{16}(n) 
, \qquad n\ge 4, \label{ext3} \\
\aa_{13}(n) = \aa_{20}(n), \qquad n\ge 3. \label{ext4}
\end{align}
\end{lemma}
\begin{proof}
Since $\aa_{11} \subset \aa_{16}$, we have $\aa_{11}(n) \subseteq \aa_{16}(n)$ for all $n$. But because $\dim \aa_{11}(4) = \dim \aa_{16}(4) = 120$, we obtain
that $\aa_{11}(4) = \aa_{16}(4)$, which implies \eqref{ext3}. Similarly, from $\aa_{13} \subset \aa_{20}$, we get $\aa_{13}(n) \subseteq \aa_{20}(n)$ for all $n$. 
But because $\dim \aa_{13}(3) = \dim \aa_{20}(3) = 30$, we have $\aa_{13}(3) = \aa_{20}(3)$.
\end{proof}

There are other inclusions among the Lie algebras $\aa_k$ after we relabel some of the Paulis. Such relabelings act as automorphisms of the Lie algebra $\su(2^n)$, i.e., they are invertible linear operator that respects the Lie bracket (see Sect. \ref{secpauli3}). We will express them in the form $\varphi(a) = U a U^\dagger$ ($a\in\su(2^n)$) for some fixed unitary matrix $U$, as in Lemma \ref{leminv2}.

As a first example, consider the linear operator $\psi$ on $\mathbb{C}^{2\times2}$, defined by
$\psi(A) = V A V^\dagger$ where $V=e^{i\frac{\pi}{4} Z}$. Using \eqref{equab}, we find:
\begin{equation*}
\psi(I) = I, \qquad \psi(X) = iZ\cdot X = -Y, \qquad \psi(Y) = iZ \cdot Y = X, \qquad \psi(Z) = Z.
\end{equation*}
We extend it to an automorphism $\psi_n$ of $\su(2^n)$ by:
\begin{equation*}
\psi_n := \psi^0 \otimes \psi^1 \otimes\dots\otimes \psi^{n-1},
\end{equation*}
where $\psi^j$ denotes the $j$-th power of $\psi$. Note that, up to an overall sign, $\psi_n$ swaps $X \rightleftharpoons Y$ on all even qubits.
It can be represented as a unitary transformation:
\begin{equation}\label{psi3}
\psi_n(a) = U a U^\dagger, \qquad\text{with}\quad
U = V^1 \otimes V^2 \otimes\dots\otimes V^n = \exp\Bigl( i\frac{\pi}{4} \sum_{j=1}^n (j-1)Z_j \Bigr).
\end{equation}

\begin{lemma}\label{leminc3}
The map $\psi_n$, defined by \eqref{psi3}, restricts to an isomorphism $\aa_2(n) \cong \aa_4(n)$.
\end{lemma}
\begin{proof}
Since $\psi_n$ is an automorphisms of $\su(2^n)$, it is in particular injective and respects the Lie bracket. The same is true for the restriction of $\psi_n$ to $\aa_2(n)$. In order to prove that $\psi_n$ is an isomorphism from $\aa_2(n)$ to $\aa_4(n)$, it remains to show that it is surjective, i.e., $\psi_n\aa_2(n) = \aa_4(n)$.
Note that $\psi_n\aa_2(n)$ is a subalgebra of $\aa_4(n)$. We will show that $\psi_n\aa_2(n)$ contains all generators of $\aa_4(n)$, which would imply that it is equal to it.

Indeed, $\psi_n$ acts as follows on the generators of $\aa_2(n)$:
\begin{align*}
\psi_n(X_i Y_{i+1}) &= (\psi^{i-1}(X))_i (\psi^{i}(Y))_{i+1} = (\psi^{i-1}(X))_i (\psi^{i-1}(X))_{i+1}, \\
\psi_n(Y_i X_{i+1}) &= (\psi^{i-1}(Y))_i (\psi^{i}(X))_{i+1} = -(\psi^{i-1}(Y))_i (\psi^{i-1}(Y))_{i+1}.
\end{align*}
Hence, up to a sign, $\psi_n$ sends the generators of $\aa_2(n)$ to the generators $X_i X_{i+1}$, $Y_i Y_{i+1}$ of $\aa_4(n)$.
Therefore, $\psi_n\aa_2(n) = \aa_4(n)$, which completes the proof of the lemma.
\end{proof}

As another similar example, consider the linear operator $\varphi$ on $\mathbb{C}^{2\times2}$, defined by
\begin{equation}\label{varphi1}
\varphi(A) := e^{i\frac{\pi}{4} X} A e^{-i\frac{\pi}{4} X} \;\;\Rightarrow\;\;
\varphi(I) = I, \quad \varphi(X) = X, \quad \varphi(Y) = -Z, \quad \varphi(Z) = Y.
\end{equation}
We extend it to an automorphism of $\su(2^n)$ by
\begin{equation}\label{varphi2}
\varphi_n := \varphi^0 \otimes \varphi^1 \otimes\dots\otimes \varphi^{n-1},
\end{equation}
which, up to a sign, swaps $Y \rightleftharpoons Z$ on all even qubits. As in \eqref{psi3}, we have
\begin{equation}\label{varphi3}
\varphi_n(a) = U a U^\dagger, \qquad\text{with}\quad
U = \exp\Bigl( i\frac{\pi}{4} \sum_{j=1}^n (j-1)X_j \Bigr).
\end{equation}

\begin{lemma}\label{leminc4}
The map $\varphi_n$, defined by \eqref{varphi1}, \eqref{varphi2}, restricts to an isomorphism $\aa_6(n) \cong \aa_7(n)$.
\end{lemma}
\begin{proof}
As in the proof of Lemma \ref{leminc3}, we find that $\varphi_n$ acts on the generators of $\aa_6(n)$ as follows:
\begin{align*}
\varphi_n(X_i X_{i+1}) &= X_i X_{i+1}, \\
\varphi_n(Y_i Z_{i+1}) &= (\varphi^{i-1}(Y))_i (\varphi^{i-1}(Y))_{i+1}, \\
\varphi_n(Z_i Y_{i+1}) &= -(\varphi^{i-1}(Z))_i (\varphi^{i-1}(Z))_{i+1}.
\end{align*}
Up to a sign, the images are exactly the generators $X_i X_{i+1}$, $Y_i Y_{i+1}$, $Z_i Z_{i+1}$ of $\aa_7(n)$. 
Hence, $\varphi_n\aa_6(n) = \aa_7(n)$.
\end{proof}

Now consider the composition $\gamma := \varphi\psi$, which acts as a cyclic rotation $X \mapsto Z \mapsto Y \mapsto X$:
\begin{equation}\label{rho1}
\gamma(I) = I, \qquad \gamma(X) = Z, \qquad \gamma(Y) = X, \qquad \gamma(Z) = Y.
\end{equation}
We extend it to automorphism of $\su(2^n)$ as follows:
\begin{equation}\label{rho2}
\gamma_n := \gamma^{1} \otimes \gamma^{2} \otimes \gamma^{3} \otimes\dots\otimes \gamma^{n}.
\end{equation}
Since $S:=(X+Y+Z)/\sqrt{3}$ satisfies $S\cdot S=I$, we can apply Euler's formula \eqref{eqeuler} to show that
\begin{equation*}
\gamma(A) = e^{i\frac{\pi}{4} X} e^{i\frac{\pi}{4} Z} A e^{-i\frac{\pi}{4} Z} e^{-i\frac{\pi}{4} X}
= e^{i\frac{\pi}{3} S} A e^{-i\frac{\pi}{3} S}.
\end{equation*}
Hence, similarly to \eqref{psi3}, \eqref{varphi3}, we can express $\gamma_n$ as 
\begin{equation}\label{rho3}
\gamma_n(a) = U a U^\dagger, \qquad\text{with}\quad
U = \exp\Bigl( i\frac{\pi}{3 \sqrt{3}} \sum_{j=1}^n j ( X_j + Y_j + Z_j ) \Bigr).
\end{equation}


\begin{lemma}\label{leminc5}
The map $\gamma_n$, defined by \eqref{rho1}, \eqref{rho2}, restricts to an isomorphism $\aa_{10}(n) \cong \aa_7(n)$.
\end{lemma}
\begin{proof}
We find that $\gamma_n$ acts on the generators of $\aa_{10}(n)$ as follows:
\begin{align*}
\gamma_n(X_i Y_{i+1}) &= (\gamma^{i}(X))_i (\gamma^{i+1}(Y))_{i+1} = (\gamma^{i}(X))_i (\gamma^{i}(X))_{i+1}, \\
\gamma_n(Y_i Z_{i+1}) &= (\gamma^{i}(Y))_i (\gamma^{i+1}(Z))_{i+1} = (\gamma^{i}(Y))_i (\gamma^{i}(Y))_{i+1}, \\
\gamma_n(Z_i X_{i+1}) &= (\gamma^{i}(Z))_i (\gamma^{i+1}(X))_{i+1} = (\gamma^{i}(Z))_i (\gamma^{i}(Z))_{i+1}.
\end{align*}
The images are exactly the generators $X_i X_{i+1}$, $Y_i Y_{i+1}$, $Z_i Z_{i+1}$ of $\aa_7(n)$;
hence, $\gamma_n\aa_{10}(n) = \aa_7(n)$.
\end{proof}

\subsection{Frustration graphs}\label{secfrus}


\begin{figure}[htb!]
    \centering
    \subfloat[$\aa_1(4)$]{%
    \begin{minipage}[c][1\width]{0.2\textwidth}%
    \includegraphics[clip,width=1\textwidth]{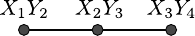}%
    \end{minipage}}\hspace{5mm}
    \subfloat[$\aa_2(4)$]{%
    \begin{minipage}[c][1\width]{0.2\textwidth}%
    \includegraphics[clip,width=1\textwidth]{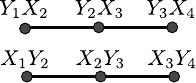}%
    \end{minipage}}\hspace{5mm}
    \subfloat[$\aa_2^\circ(4)$]{%
    \begin{minipage}[c][1\width]{0.2\textwidth}%
    \includegraphics[clip,width=1\textwidth]{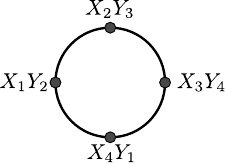}%
    \end{minipage}}
    \caption{Frustration graphs for several examples in our classification for $n=4$. (a) A frustration graph given by a line. (b) A frustration graph consisting of two disjoint lines. (c) A circular frustration graph.}
    \label{fig:frustration_graphs}
\end{figure}

In this subsection, we review the notion of frustration graph, which is a useful visualization tool; see e.g.\
\cite{chapman2020characterization,atia2017fast,gu2021fast}.
We determine the DLA in the cases when the frustration graph is a line or a circle, and apply these results to identify several of our Lie algebras,
namely, $\aa_k(n)$ and $\aa_k^\circ(n)$ for $k=1,2,4,8,14$.

\begin{definition}\label{deffrus}
Given a set of Pauli strings $\mcA\subset\mcP_n$, its \emph{frustration graph} is the graph with a set of vertices $\mcA$ and edges connecting all pairs of vertices $a,b\in\mcA$ such that $[a,b] \neq 0$.
\end{definition}

The frustration graph makes it especially easy to determine when two subsets of the generating set $\mcA$ commute with each other: it means that they are disconnected from each other in the frustration graph. Suppose that $\mcA = \mcA_1 \sqcup \mcA_2$ is a disjoint union of subsets with disconnected frustration graphs. Then
\begin{equation}\label{eqlieopl}
\Lie{\mcA_1 \sqcup \mcA_2} = \Lie{\mcA_1} \oplus \Lie{\mcA_2}
\end{equation}
is a direct sum of commuting subalgebras. This is illustrated in Figure \ref{fig:frustration_graphs}.

In the next proposition, we determine the DLA of a line frustration graph.

\begin{proposition}\label{lem:line_graph}
Suppose that the frustration graph of $\mcA = \{a_1, \dots, a_N\} \subset\mcP_n$ is a line, so that $[a_j, a_k] \neq 0$ for $1\le j < k \le N$
if and only if $k=j+1$. Then 
\begin{equation*}
\Lie{\mcA} \cong \so(N+1), 
\end{equation*}
and a basis for it is given by $\{ i^{k-j} L_{j,k} \}_{1\le j < k \le N+1}$, 
where $i$ is the imaginary unit and
\begin{equation}\label{eqLjk}
L_{j,k} := a_j \cdot a_{j+1} \cdots a_{k-1} \qquad (1\le j < k \le N+1)
\end{equation}
are products over line segments.
\end{proposition}
\begin{proof}
Recall that $\Lie{\mcA}$ is the subalgebra of $\su(2^n)$ generated by the subset $i\mcA \subset \su(2^n)$ (see Definition \ref{defgen}).
First, let us us prove that all $i^{k-j} L_{j,k}$ are in $\Lie{\mcA}$. For $k=j+1$, we have $i L_{j,j+1} := i a_j \in i\mcA\subseteq\Lie{\mcA}$.
Suppose, by induction on $k-j$, that $i^{k-j} L_{j,k} \in\Lie{\mcA}$ for some $1<j<k\le N+1$; then we will show that $i^{k-j+1} L_{j-1,k} \in\Lie{\mcA}$.
By definition,
\begin{equation*}
L_{j-1,k} = a_{j-1} \cdot a_j \cdot a_{j+1} \cdots a_{k-1} = a_{j-1} \cdot L_{j,k},
\end{equation*}
and by assumption, $a_{j-1}$ anticommutes with $a_j$ and commutes with $a_{j+1}, \dots, a_{k-1}$.
Hence, $a_{j-1}$ anticommutes with $L_{j,k}$, which implies that
\begin{equation*}
2i^{k-j+1} L_{j-1,k} = 2i^{k-j+1} a_{j-1} \cdot L_{j,k} = i^{k-j+1} [a_{j-1}, L_{j,k}] = [ia_{j-1}, i^{k-j} L_{j,k}] \in\Lie{\mcA}.
\end{equation*}
This proves the claim that $i^{k-j} L_{j,k} \in\Lie{\mcA}$ for all $1\le j < k \le N+1$.

Similarly to above, one can check that ($1\le j < k < l \le N+1$):
\begin{equation}\label{eq:P_commutation_relations}
[i^{k-j} L_{j,k}, i^{l-k} L_{k,l}] = 2 i^{l-j} L_{j,l}, \quad
[i^{l-k} L_{k,l}, i^{l-j} L_{j,l}] = 2 i^{k-j} L_{j,k}, \quad
[i^{l-j} L_{j,l}, i^{k-j} L_{j,k}] = 2 i^{l-k} L_{k,l}, 
\end{equation}
and all other commutators (not following from skewsymmetry) are zero.
In particular, the real linear span of all $i^{k-j} L_{j,k}$ is closed under the bracket, i.e., is a subalgebra of $\su(2^n)$.
Since $\Lie{\mcA}$ is the minimal (under inclusion) subalgebra of $\su(2^n)$ containing $i\mcA$, it follows that
\begin{equation*}
\Lie{\mcA} = \Span \{ i^{k-j} L_{j,k} \}_{1\le j < k\le N+1}.
\end{equation*}

Recall that $\so(N+1) = \so(N+1,\mathbb{R})$ is the Lie algebra of all skewsymmetric $(N+1) \times (N+1)$ real matrices; see \eqref{soN}.
Consider the standard basis $\{ E_{j,k} \}_{1\le j,k\le N+1}$ of $\gl(N+1,\mathbb{R})$, where $E_{j,k}$ is the matrix with $(j,k)$-entry $=1$ and all other entries $=0$.
Then a basis for $\so(N+1)$ is $\{ F_{j,k} := E_{j,k}-E_{k,j} \}_{1\le j<k\le N+1}$. Using that
\begin{equation*}
[E_{j,k}, E_{l,m}] = \delta_{k,l} E_{j,m} - \delta_{j,m} E_{l,k},
\end{equation*}
it is easy to see that
\begin{equation*}
[F_{j,k}, F_{k,l}] = F_{j,l}, \quad
[F_{k,l}, F_{j,l}] = F_{j,k}, \quad
[F_{j,l}, F_{j,k}] = F_{k,l}, \quad\text{for}\quad 1\le j < k < l \le N+1.
\end{equation*}
Hence, the matrices $2F_{j,k}$ satisfy the same commutation relations as $i^{k-j} L_{j,k}$ given in \eqref{eq:P_commutation_relations}.
This means that the map $\so(N+1) \to \Lie{\mcA}$ that sends $2F_{j,k}$ to $i^{k-j} L_{j,k}$ is a Lie algebra homomorphism. Its kernel is an ideal in $\so(N+1)$, but since $\so(N+1)$ is simple, it has no non-zero proper ideals. Therefore, this map is an isomorphism.
\end{proof}

\begin{remark}
One can see from the above proposition that, for linear frustration graphs, the dimension of the DLA scales quadratically with the number of generators. 
This was observed for free fermionic models in \cite{bassman2022constant, Kokcu2021cartan, kokcu2022algebraic, camps2022algebraic, kokcu2023algebraic},
where the number of generators is proportional to the system size and the circuit gate complexity is quadratic with respect to the system size. 
These models are fast forwardable along with the other Hamiltonians given in \cite{chapman2020characterization,atia2017fast,gu2021fast}, and the fundamental reason for this is the polynomial scaling of the DLA.

After a Jordan--Wigner transformation, it can be shown that the algebra of free fermions on $n$ sites can be generated by 
$\{Z_1, X_1X_2, Z_2, X_2X_3, Z_3, \dots, X_{n-1}X_n, Z_n\}$, which will be shown to be equivalent to $\aa_{14}(n)$ in Lemma \ref{lemfrus1}. These generators have a linear frustration graph with $2n-1$ vertices; hence, its DLA is $\so(2n)$.
\end{remark}

Next, we consider the case when the frustration graph is a circle.

\begin{proposition}\label{lem:circular_graph}
Suppose that the frustration graph of $\mcA = \{a_1, \dots, a_N\} \subset\mcP_n$ is a circle with $N\ge 3$, so that $[a_j, a_k] \neq 0$ for $1\le j < k \le N$
if and only if $k=j+1$ or $j=1$, $k=N$.
Then 
\begin{equation*}
\Lie{\mcA} \cong \so(N) \oplus \so(N), 
\end{equation*}
and it has a basis $\{ i^{k-j} L_{j,k},  i^{N+k-j} C \cdot L_{j,k}\}_{1\le j < k \le N}$, 
where $i$ is the imaginary unit, $L_{j,k}$ are defined in \eqref{eqLjk}, and
\begin{equation}\label{eqC}
C := 
a_1 \cdot a_2 \cdots a_{N-1} \cdot a_N.
\end{equation}
\end{proposition}
\begin{proof}
First, notice that $[C,a_j]=0$ for all $1\le j\le N$, because $a_j$ does not commute only with its two neighboring vertices in the circle frustration graph.
Moreover, using that $a_j \cdot a_j = I^{\otimes n}$ (cf.\ Lemma \ref{lemp0}), we get $C \cdot C = (-1)^N I^{\otimes n}$. From here, we deduce that
\begin{equation}\label{eqCC}
[i^N C, i^{k-j} L_{j,k}] = 0, \qquad (i^N C) \cdot (i^N C) = I^{\otimes n}.
\end{equation}

If we remove any vertex from the frustration graph of $\mcA$, we obtain a line frustration graph. 
By Proposition \ref{lem:line_graph}, we know that
\begin{equation*}
i^{k-j} L_{j,k} \in \Lie{a_1, \dots, a_{N-1}} \subseteq \Lie{\mcA}, \qquad 1\le j < k \le N,
\end{equation*}
and these elements form a basis for the subalgebra $\Lie{a_1, \dots, a_{N-1}} \cong \so(N)$.
In particular, we have
\begin{equation*}
\{ i a_1, \dots, i a_{N-1} \} \subset \Lie{a_1, \dots, a_{N-1}} = \Span_{\mathbb{R}} \{ i^{k-j} L_{j,k} \}_{1\le j < k \le N}.
\end{equation*}

Similarly, the set $\mcA\setminus\{a_{k-1}\} = \{a_k,a_{k+1},\dots,a_N,a_1,\dots,a_{k-2}\}$ has a line frustration graph and
its subset $\{a_k,a_{k+1},\dots,a_N,a_1,\dots,a_{j-1}\}$ is a line segment for $1\le j < k \le N$. Hence, again by Proposition \ref{lem:line_graph},
\begin{equation*}
i^{N+k-j} C \cdot L_{j,k} = \pm i^{N-k+j} a_k \cdot a_{k+1} \cdots a_N \cdot a_1 \cdots a_{j-1}
\in \Lie{\mcA\setminus\{a_{k-1}\}} \subseteq \Lie{\mcA}, \qquad 1\le j < k \le N.
\end{equation*}
In particular, the choice $j=1$, $k=N$ gives
\begin{equation*}
i^{N+N-1} C \cdot L_{1,N} = \pm i a_N.
\end{equation*}

The above discussion implies that
\begin{equation*}
i\mcA \subset \mcL := \Span_{\mathbb{R}} \{ i^{k-j} L_{j,k},  i^{N+k-j} C \cdot L_{j,k}\}_{1\le j < k \le N} \subseteq \Lie{\mcA}.
\end{equation*}
We claim that the vector space $\mcL$ is closed under the Lie bracket. Indeed, we already know that $\Span_{\mathbb{R}} \{ i^{k-j} L_{j,k} \}$ is closed.
For the other brackets, we use that from \eqref{eqCC}, we have:
\begin{align*}
[i^{k-j} L_{j,k}, i^{N+m-l} C \cdot L_{l,m}] &= i^N C \cdot [i^{k-j} L_{j,k}, i^{m-l} L_{l,m}], \\
[i^{N+k-j} C \cdot L_{j,k}, i^{N+m-l} C \cdot L_{l,m}] &= [i^{k-j} L_{j,k}, i^{m-l} L_{l,m}].
\end{align*}
As the Lie algebra $\mcL$ contains $i\mcA$, it must contain $\Lie{\mcA}$. Therefore, $\mcL=\Lie{\mcA}$.

Using \eqref{eqCC} again (or from the above brackets), we see that
\begin{equation*}
\Lie{\mcA} = \Span_{\mathbb{R}} \{ (I^{\otimes n} + i^N C) \cdot i^{k-j} L_{j,k} \}_{1\le j < k \le N} 
\oplus \Span_{\mathbb{R}} \{ (I^{\otimes n} - i^N C) \cdot i^{k-j} L_{j,k} \}_{1\le j < k \le N}
\end{equation*}
is isomorphic as a Lie algebra to a direct sum of two copies of $\Span_{\mathbb{R}} \{ i^{k-j} L_{j,k} \}_{1\le j < k \le N} \cong \so(N)$.
Therefore, $\Lie{\mcA} \cong \so(N) \oplus \so(N)$.
\end{proof}

\begin{remark}
A circular frustration graph corresponds to free fermionic evolution controlled with one ancilla, where the ancilla degree of freedom can be readily found as the operator $C$ defined in \eqref{eqC}. As expected, this is not the only example. Some periodic $1$-dimensional spin systems such as TFXY, XY and Kitaev models also have DLAs that are generated from Pauli strings with a circular frustration graph. For those spin models, $C = Z Z \cdots Z$. 
\end{remark}


Applying the results of Propositions \ref{lem:line_graph}, \ref{lem:circular_graph}, 
in the following lemmas we determine the Lie algebras $\aa_k(n)$ and $\aa_k^\circ(n)$ for $k=1,2,4,8,14$.
Examples are presented in Figures \ref{fig:frust_proofs_2} and \ref{fig:frust_proofs}.

\begin{figure}[htb!]
    \centering
    \subfloat[$\aa_2(n)$]{\includegraphics[width=0.25\textwidth]{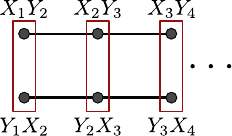}}\hspace{5mm}
    \subfloat[$\aa_4(n)$]{\includegraphics[width=0.25\textwidth]{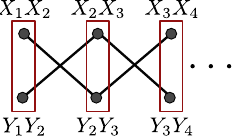}}
    \caption{Frustration graphs for $\aa_2(n)$ and $\aa_4(n)$. The red boxes denote a generator acting on the site ($i$, $i+1$). Both frustrations graphs are given by two disjoint lines for any $n$; hence we can conclude that $\aa_2(n)\cong\aa_4(n)$ (cf.\ Lemma \ref{leminc3}).
    }
    \label{fig:frust_proofs_2}
\end{figure}

\begin{lemma}\label{lemfrus1}
    We have:
    \begin{align*}
    \begin{split}
    \aa_1(n) & \cong \so(n), \\
    \aa_2(n) & \cong \aa_4(n) \cong \so(n) \oplus \so(n), \\
    \aa_8(n) & \cong \so(2n-1), \\
    \aa_{14}(n) & \cong \so(2n).
    \end{split}
    \end{align*}
\end{lemma}
\begin{proof}
The proof is based on the frustration graphs of the generating sets of these Lie algebras (see Sect.\ \ref{secres1}, \ref{secext}).

Generators of $\aa_1(n)$ are $XY$ on each adjacent pair of qubits:
\begin{equation*}
X_1Y_2, X_2Y_3, X_3Y_4, \dots, X_{n-1}Y_n. 
\end{equation*}
These form a linear frustration graph with $n-1$ vertices, leading to $\aa_1(n) \cong \so(n)$ (see Figure \ref{fig:frust_proofs}(a)).

The Lie algebra $\aa_2(n)$ is generated by $XY$ and $YX$ on adjacent pairs of qubits:
\begin{equation*}
X_1Y_2, X_2Y_3, X_3Y_4, \dots, X_{n-1}Y_n \qquad\text{and}\qquad 
Y_1X_2, Y_2X_3, Y_3X_4, \dots, Y_{n-1}X_n.
\end{equation*}
Both of these form linear frustration graphs with $n-1$ vertices, and commute with each other (see Figure \ref{fig:frust_proofs_2}). Thus, $\aa_2(n) \cong \mathfrak{so}(n) \oplus \mathfrak{so}(n)$. Note that $\aa_4(n) \cong \aa_2(n)$ due to Lemma \ref{leminc3} (see also Figure \ref{fig:frust_proofs_2}).

Since $\aa_8 = \Span\{XX, XZ, IY\} = \Lie{XX, IY}$, we can generate $\aa_8(n)$ by:
\begin{equation}\label{a8gen}
X_1X_2, Y_2, X_2X_3, Y_3, X_3X_4, Y_4, \dots, X_{n-1}X_n, Y_n. 
\end{equation}
These form a linear frustration graph with $2(n-1)$ vertices; hence $\aa_8(n) \cong \so(2n-1)$ (see Figure \ref{fig:frust_proofs}(b)).

Similarly, note that
\begin{equation*}
\aa_{14} = \Span\{XX, YY, XY, YX, ZI, IZ\} = \Lie{XX, ZI, IZ},
\end{equation*}
because $[XX,ZI] = 2iYX$, $[XX,IZ] = 2iXY$, and $[XY,ZI] = -2iYY$.
Thus, $\aa_{14}(n)$ is generated by:
\begin{equation}\label{a14gen}
Z_1, X_1X_2, Z_2, X_2X_3, Z_3, X_3X_4, Z_4, \dots, X_{n-1}X_n, Z_n,
\end{equation}
which gives a linear frustration graph with $2n-1$ vertices. Hence $\aa_{14}(n) \cong \so(2n)$ (see Figure \ref{fig:frust_proofs}(c)).
\end{proof}
\begin{figure}[htb!]
    \centering
    \subfloat[$\aa_1(n)$]{\includegraphics[width=0.25\textwidth]{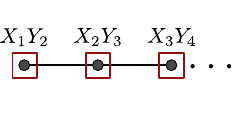}}\hspace{5mm}
    \subfloat[$\aa_8(n)$]{\includegraphics[width=0.25\textwidth]{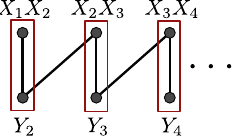}}\hspace{5mm}
    \subfloat[$\aa_{14}(n)$]{\includegraphics[width=0.285\textwidth]{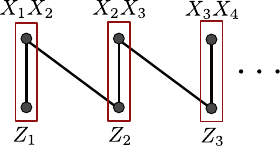}}
    \caption{ Visualization of the frustration graphs of generators of certain Lie algebras. The red boxes denote a generator acting on the site ($i$, $i+1$). For the three cases $\aa_1(n)$, $\aa_8(n)$ and $\aa_{14}(n)$, we see that the frustration graph is a line for any $n$. For (a), the 2-site DLA $\aa_1(2)\cong \uu(1)$, but as $n$ grows we find $\aa_1(n)\cong \so(n)$. 
    (b) The frustration graph is a line with $2n-2$ vertices; hence $\aa_8(n)\cong \so(2n-1)$.
    (c) The frustration graph is a line with $2n-1$ vertices, giving $\aa_{14}(n)\cong \so(2n)$.
    }
    \label{fig:frust_proofs}
\end{figure}
\begin{lemma}\label{lemfrus2}
We have\/
$\aa_k^\circ(n) \cong \aa_k(n)^{\oplus 2}$ for $k=1,2,14$ and $n\ge 3$.
\end{lemma}
\begin{proof}
The Lie algebra $\aa_1^\circ(n)$ is generated by $XY$ applied on adjacent qubits, including periodic boundary conditions:
\begin{equation*}
X_1Y_2, X_2Y_3, X_3Y_4, \dots, X_{n-1}Y_n, X_{n}Y_1.
\end{equation*}
The frustration graph is a circle with $n$ vertices. Therefore, $\aa_1^\circ(n) \cong \so(n)^{\oplus 2} \cong \aa_1(n)^{\oplus 2}$.

The Lie algebra $\aa_2^\circ(n)$ is generated by $XY$ and $YX$ applied on adjacent qubits with periodic boundary conditions:
\begin{equation*}
X_1Y_2, X_2Y_3, X_3Y_4, \dots, X_{n-1}Y_n, X_{n}Y_1 \qquad\text{and}\qquad 
Y_1X_2, Y_2X_3, Y_3X_4, \dots, Y_{n-1}X_n, Y_{n}X_1.
\end{equation*}
These form two circular frustration graphs with $n$ vertices that are disconnected from each other. Thus, $\aa_2^\circ(n) \cong \so(n)^{\oplus 2} \oplus \so(n)^{\oplus 2} \cong \aa_2(n)^{\oplus 2}$.

Using the generating set \eqref{a14gen} of $\aa_{14}(n)$, we see that $\aa_{14}^\circ(n)$ can be generated by:
\begin{equation*}
Z_1, X_1X_2, Z_2, X_2X_3, Z_3, X_3X_4, Z_4, \dots, X_{n-1}X_n, Z_n, X_nX_1.
\end{equation*}
This leads to a circular frustration graph with $2n$ vertices, so $\aa_{14}^\circ(n) \cong \so(2n)^{\oplus 2} \cong \aa_{14}(n)^{\oplus 2}$.
\end{proof}

\begin{lemma}\label{lemfrus3}
We have\/
$\aa_8^\circ(n) \cong \aa_{14}^\circ(n) \cong \so(2n)^{\oplus 2}$ for all $n\ge 3$.
\end{lemma}
\begin{proof}
The generating sets \eqref{a8gen} and \eqref{a14gen} of $\aa_8(n)$ and $\aa_{14}(n)$ are the same after swapping $Y \rightleftharpoons Z$, 
except that $\aa_8(n)$ does not have $Y_1$. When the periodic boundary condition is applied, this difference disappears and we obtain that
$\aa_8^\circ(n) \cong \aa_{14}^\circ(n)$. 
\end{proof}

\begin{lemma}\label{lemfrus4}
We have\/
$\aa_4^\circ(n) \cong \begin{cases}
\so(2n)^{\oplus 2}, & n \;\mathrm{odd}, \\
\so(n)^{\oplus 4}, & n \;\mathrm{even}.
\end{cases}$
\end{lemma}
\begin{proof}
As $\aa_4 = \Lie{XX,YY}$, the generators of $\aa_4^\circ(n)$ are:
\begin{equation*}
X_1X_2, Y_1Y_2, X_2X_3, Y_2Y_3, 
\dots, X_{n-1}X_n, Y_{n-1}Y_n, X_nX_1, Y_nY_1.
\end{equation*}
For odd $n$, these generators form a circular frustration graph with $2n$ vertices:
\begin{equation*}
X_1X_2, Y_2Y_3, X_3X_4, Y_4Y_5, \dots, Y_{n-1}Y_n, X_nX_1, Y_1Y_2, X_2X_3, Y_3Y_4, X_4X_5, \dots, X_{n-1}X_n, Y_nY_1.
\end{equation*}
Hence, in this case, $\aa_4^\circ(n) \cong \so(2n)^{\oplus 2}$.

When $n$ is even, the generators form two disjoint circles with $n$ vertices each:
\begin{equation*}
X_1X_2, Y_2Y_3, X_3X_4, \dots, X_{n-1}X_n, Y_nY_1 \qquad\text{and}\qquad 
Y_1Y_2, X_2X_3, Y_3Y_4, \dots, Y_{n-1}Y_n, X_nX_1.
\end{equation*}
In this case, we get $\aa_4^\circ(n) \cong \so(n)^{\oplus 2} \oplus \so(n)^{\oplus 2}$.
\end{proof}

\begin{remark}
Notice that, although $\aa_2(n) \cong \aa_4(n)$ for all $n\ge 3$, we have $\aa_2^\circ(n) \not\cong \aa_4^\circ(n)$ for odd $n$.
\end{remark}

\subsection{Stabilizers, commutants, and centralizers}\label{secstab}

For any Pauli string $A\in\mcP_k$, we will use the notation $P_A=AAA\cdots\in\mcP_n$ truncated to the $n$-th qubit. For example,
\begin{align*}
P_X &= XXX\cdots, & P_{YZ} &= YZYZ\cdots, & P_{ZY} &= ZYZY\cdots, \\
P_{XYZ} &= XYZXYZ\cdots, & P_{YZX} &= YZXYZX\cdots, & P_{ZXY} &= ZXYZXY\cdots, \\
\end{align*}
where these are viewed as elements of $\mcP_n$.
In particular, $P_I=I\cdots I = I^{\otimes n}$.
Recall that $\pm\mcP_n \cup \pm i\mcP_n$ is a group under the matrix product, the Pauli group (see Sect.\ \ref{secpauli1}).

For any set of matrices $\mcA \subseteq \mathbb{C}^{2^n \times 2^n}$, we define its \emph{stabilizer} $\Stab(\mcA) \subseteq \mcP_n$ 
as the set of all Pauli strings that commute with every element of $\mcA$. 
It is clear that $\Stab(\mcA)$ is closed under multiplication, so after allowing appropriate powers of $i$ it is a group.
There are two related, and essentially equivalent, notions called the commutant and centralizer.  
All of these consist of elements commuting with the given set $\mcA$ but differ in where such elements are and what structure they form:
the stabilizer is a subgroup of the Pauli group (up to factors of $\pm1,\pm i$); 
the commutant is a subalgebra of the associative algebra $\mathbb{C}^{2^n \times 2^n}$ of all complex matrices; 
while the centralizer is a subalgebra of the real Lie algebra $\su(2^n)$.
The precise relations are explained in the following two remarks.

\begin{remark}\label{remstab1}
The \emph{commutant} of a set $\mcA \subseteq \mathbb{C}^{2^n \times 2^n}$ is defined as the set $\mcA'$ of all $2^n \times 2^n$ complex matrices that commute with all elements
of $\mcA$. Then $\mcA'$ is closed under addition, multiplication, and multiplication by any complex scalar, i.e., it is an associative algebra over $\mathbb C$.
It is easy to see that $\mcA' = \Span_{\mathbb C} \Stab(\mcA)$ is the complex linear span of the stabilizer $\Stab(\mcA)$.

Note that 
$\mcA\subseteq\mcA'' := (\mcA')'$.
By (a finite-dimensional version of) von Neumann's Double Commutant Theorem (see e.g.\ \cite{procesi2007lie}, Theorem 6.2.5), 
$\mcA''$ is the associative algebra generated by $\mcA$. In particular, for a complex vector space $\mcA$, we have $\mcA'' = \mcA + \mathbb{C} I^{\otimes n}$ if and only if $\mcA$ is closed under multiplication.
\end{remark}

\begin{remark}\label{remstab2}
The \emph{centralizer} of a set $\mcA\subseteq\su(2^n)$ is the set 
$\su(2^n)^\mcA \subseteq \su(2^n)$ of all traceless skew-Hermitian $2^n \times 2^n$ matrices that commute with all elements
of $\mcA$. Then $\su(2^n)^\mcA$ is closed under addition, commutator, and multiplication by any real scalar, i.e., it is a Lie subalgebra of $\su(2^n)$.
It is easy to see that 
$\su(2^n)^\mcA = \Span (\Stab(\mcA)\setminus\{I^{\otimes n}\})$ is the real span of $i(\Stab(\mcA)\setminus\{I^{\otimes n}\})$.
\end{remark}


When $\s\subseteq \su(2^n)$ is a Lie subalgebra, a given Pauli string commutes with all elements of $\s$
if and only if it commutes with all generators of $\s$. In other words, we have
\begin{equation}\label{stabalg}
\Stab(\Lie{\mcA}) = \Stab(\mcA).
\end{equation}
Thus, to determine the stabilizers of our Lie subalgebras of $\su(2^n)$, it suffices to find the stabilizers of their generating sets.
In the case $n=2$, it is easy to find the answer by inspection, which is given 
as follows:
\allowdisplaybreaks
\begin{align*}
\Stab(\aa_0) &= \{II,IX,XI,XX,YY,YZ,ZY,ZZ\}, \\
\Stab(\aa_1) &= \{II,IY,XI,XY,YX,YZ,ZX,ZZ\}, \\
\Stab(\aa_2) &= \{II,XY,YX,ZZ\}, \\
\Stab(\aa_3) &= \Stab(\aa_6) = \{II,XX,YZ,ZY\}, \\
\Stab(\aa_4) &= \Stab(\aa_7) = \{II,XX,YY,ZZ\}, \\
\Stab(\aa_5) &= \Stab(\aa_{10}) = \{II,XY,YZ,ZX\}, \\
\Stab(\aa_8) &= \{II,XI,YY,ZY\}, \\
\Stab(\aa_9) &= \{II,XI,YX,ZX\}, \\
\Stab(\aa_{11}) &= \{II,XY\}, \\
\Stab(\aa_{12}) &= \Stab(\aa_{17}) = \Stab(\aa_{19}) = \{II,YZ\}, \\
\Stab(\aa_{13}) &= \Stab(\aa_{20}) = \{II,XX\}, \\
\Stab(\aa_{14}) &= \{II,ZZ\}, \\
\Stab(\aa_{15}) &= \Stab(\bb_2) = \Stab(\bb_4) = \{II,XI\}, \\
\Stab(\bb_0) &= \Stab(\bb_1) = \{II,XX,XI,IX\}, \\
\Stab(\aa_k) &= \Stab(\bb_3) = \{II\},  \qquad k=16,18,21,22.
\end{align*}
Using that, we can find the stabilizers of the subalgebras of $\su(2^n)$.

\begin{proposition}\label{pstab1}
For $n\ge3$, we have the following stabilizers:
\begin{align*}
\Stab(\aa_0(n)) &= \{I,X\}^{\otimes n} \cup \{Y,Z\}^{\otimes n}, \\
\Stab(\aa_2(n)) &= \{P_I,P_{XY},P_{YX},P_Z\}, \\
\Stab(\aa_3(n)) &= \Stab(\aa_6(n)) = \{P_I,P_X,P_{YZ},P_{ZY}\}, \\
\Stab(\aa_4(n)) &= \Stab(\aa_7(n)) = \{P_I,P_X,P_Y,P_Z\}, \\
\Stab(\aa_5(n)) &= \Stab(\aa_{10}(n)) = \{P_I,P_{XYZ},P_{YZX},P_{ZXY}\}, \\
\Stab(\aa_8(n)) &= \{P_I,P_Y,X_1,ZY\cdots Y\}, \\
\Stab(\aa_9(n)) &= \{P_I,X_1,Y_1X_2,Z_1X_2\}, \\
\Stab(\aa_{13}(n)) &= \Stab(\aa_{20}(n)) = \{P_I,P_X\}, \\
\Stab(\aa_{14}(n)) &= \{P_I,P_Z\}, \\
\Stab(\aa_{15}(n)) &= \Stab(\bb_2(n)) = \Stab(\bb_4(n)) = \{P_I,X_1\}, \\
\Stab(\bb_0(n)) &= \Stab(\bb_1(n)) = \{I,X\}^{\otimes n}, \\
\Stab(\aa_k(n)) &= \Stab(\bb_3(n)) = \{P_I\}, \qquad k=11,12,16{-}19,21,22.
\end{align*}
\end{proposition}
\begin{proof}
Recall that, for any subalgebra $\aa\subseteq\su(4)$, the subalgebra $\aa(n)\subseteq\su(2^n)$ is generated by all Pauli strings $A_i B_{i+1}$, where $AB \in \aa$, $1\le i\le n-1$ 
(see \eqref{ext6}). Thus, a Pauli string $P^1 \otimes\cdots\otimes P^n \in\mcP_n$ is in $\Stab(\aa(n))$ if and only if $P^i \otimes P^{i+1} \in \Stab(\aa)$ for all $1\le i\le n-1$.
Using this observation and the knowledge of all $\Stab(\aa)$, it is straightforward to determine $\Stab(\aa(n))$. 

Let us consider the case of $\aa_2$ as an illustration. Since $\Stab(\aa_2) = \{II,XY,YX,ZZ\}$, we want to find all Pauli strings, such that for any two consecutive qubits, we have either $II$, $XY$, $YX$, or $ZZ$.
The only possible such strings are $I \cdots I$, $XYXY\cdots$, $YXYX\cdots$, or $Z\cdots Z$.
\end{proof}

The answer for $\Stab(\aa_1(n))$ is given without proof in Remark \ref{remstaba1} below, as it is more complicated but not needed for the rest of the paper.
For future use, we will need the centers of some of the above stabilizers.
Let us recall that the \emph{center} $Z(G)$ of a group $G$ consists of all $z\in G$ that commute with every $g\in G$. 
In particular, a group $G$ is Abelian if and only if $Z(G)=G$.

\begin{lemma}\label{lemstab}
For $n\ge3$, we have the following centers:
\begin{align*}
Z\bigl(\Stab(\aa_2(n))\bigr) &= 
\begin{cases} 
\{P_I,P_{XY},P_{YX},P_Z\}, & n \;\;\mathrm{even}, \\
\{P_I\}, & n \;\;\mathrm{odd},
\end{cases} \\
Z\bigl(\Stab(\aa_3(n))\bigr) &= Z\bigl(\Stab(\aa_6(n))\bigr) = 
\begin{cases} 
\{P_I,P_X,P_{YZ},P_{ZY}\}, & n \;\;\mathrm{even}, \\
\{P_I\}, & n \;\;\mathrm{odd},
\end{cases} \\
Z\bigl(\Stab(\aa_4(n))\bigr) &= Z\bigl(\Stab(\aa_7(n))\bigr) = 
\begin{cases} 
\{P_I,P_X,P_Y,P_Z\}, & n \;\;\mathrm{even}, \\
\{P_I\}, & n \;\;\mathrm{odd},
\end{cases} \\
Z\bigl(\Stab(\aa_5(n))\bigr) &= Z\bigl(\Stab(\aa_{10}(n))\bigr) = 
\begin{cases} 
\{P_I,P_{XYZ},P_{YZX},P_{ZXY}\}, & n \;\;\mathrm{even}, \\
\{P_I\}, & n \;\;\mathrm{odd},
\end{cases} \\
Z\bigl(\Stab(\aa_{13}(n))\bigr) &= Z\bigl(\Stab(\aa_{20}(n))\bigr) = \{P_I,P_X\}, \\
Z\bigl(\Stab(\aa_{14}(n))\bigr) &= \{P_I, P_Z\}, \\
Z\bigl(\Stab(\aa_{15}(n))\bigr) &= \{P_I,X_1\}, \\
Z\bigl(\Stab(\aa_k(n))\bigr) &= \{P_I\}, \qquad k=0,8,9,11,12,16{-}19,21,22.
\end{align*}
\end{lemma}

Next, we determine the stabilizers in the periodic case.

\begin{proposition}\label{pstab2}
For $n\ge3$, we have the following stabilizers:
\begin{align*}
\Stab(\aa_0^\circ(n)) &= \{I,X\}^{\otimes n} \cup \{Y,Z\}^{\otimes n}, \\
\Stab(\aa_2^\circ(n)) &= 
\begin{cases} 
\{P_I,P_{XY},P_{YX},P_Z\}, & n \;\;\mathrm{even}, \\
\{P_I,P_Z\}, & n \;\;\mathrm{odd},
\end{cases} \\
\Stab(\aa_3^\circ(n)) &= \Stab(\aa_6^\circ(n)) = 
\begin{cases} 
\{P_I,P_X,P_{YZ},P_{ZY}\}, & n \;\;\mathrm{even}, \\
\{P_I,P_X\}, & n \;\;\mathrm{odd},
\end{cases} \\
\Stab(\aa_4^\circ(n)) &= \Stab(\aa_7^\circ(n)) = \{P_I,P_X,P_Y,P_Z\}, \\
\Stab(\aa_5^\circ(n)) &= \Stab(\aa_{10}^\circ(n)) = 
\begin{cases} 
\{P_I,P_{XYZ},P_{YZX},P_{ZXY}\}, & n\equiv 0 \mod 3, \\
\{P_I\}, & n\equiv \pm1 \mod 3,
\end{cases} \\
\Stab(\aa_8^\circ(n)) &= \{P_I,P_Y\}, \\
\Stab(\aa_{13}^\circ(n)) &= \Stab(\aa_{20}^\circ(n)) = \{P_I,P_X\}, \\
\Stab(\aa_{14}^\circ(n)) &= \{P_I,P_Z\}, \\
\Stab(\bb_0^\circ(n)) &= \Stab(\bb_1^\circ(n)) = \{I,X\}^{\otimes n}, \\
\Stab(\aa_k^\circ(n)) &= \Stab(\bb_l^\circ(n)) = \{P_I\}, \qquad k=9,11,12,15{-}19,21,22, \quad l=2,3,4.
\end{align*}
\end{proposition}
\begin{proof}
For any subalgebra $\aa\subseteq\su(4)$, comparing the definitions of $\aa(n)$ and $\aa^\circ(n)$ (see \eqref{ext5}, \eqref{ext6}),
we see that $\Stab(\aa^\circ(n))$ consists of all Pauli strings $P^1 \otimes\cdots\otimes P^n \in \Stab(\aa(n))$ such that $P^n \otimes P^1 \in \Stab(\aa)$.
Thus, we determine $\Stab(\aa^\circ(n))$ by inspecting all elements of $\Stab(\aa(n))$.
\end{proof}

\begin{remark}\label{remstaba1}
One can show that $\Stab(\aa_1(n)) = \Stab(\aa_1^\circ(n))$, and as a group it is generated by the elements $P_Z$, $Y_1X_2$, $Y_2X_3$, $\dots$, $Y_{n-1}X_n$. This means that $\Stab(\aa_1(n))$ consists of all possible matrix products of these generators. For the center, we have $Z\bigl(\Stab(\aa_{1}(n))\bigr) = \{P_I, P_Z\}$.
\end{remark}

We also find the stabilizers in the permutation-invariant case.

\begin{proposition}\label{pstab3}
For $n\ge3$, we have the following stabilizers:
\begin{align*}
\Stab(\aa_0^\pi(n)) &= \{I,X\}^{\otimes n} \cup \{Y,Z\}^{\otimes n}, \\
\Stab(\aa_2^\pi(n)) &= \Stab(\aa_{14}^\pi(n)) = \{P_I,P_Z\}, \\
\Stab(\aa_4^\pi(n)) &= \Stab(\aa_7^\pi(n)) = \{P_I,P_X,P_Y,P_Z\}, \\
\Stab(\aa_6^\pi(n)) &= \Stab(\aa_{20}^\pi(n))= \{P_I,P_X\}, \\
\Stab(\bb_0^\pi(n)) &= \Stab(\bb_1^\pi(n)) = \{I,X\}^{\otimes n}, \\
\Stab(\aa_{16}^\pi(n)) &= \Stab(\bb_3^\pi(n)) = \{P_I\}.
\end{align*}
\end{proposition}
\begin{proof}
For any subalgebra $\aa\subseteq\su(4)$, from the definition of $\aa^\pi(n)$ (see \eqref{ext7}),
we see that $\Stab(\aa^\pi(n))$ consists of all Pauli strings $P^1 \otimes\cdots\otimes P^n \in \mcP_n$ such that $P^i \otimes P^j \in \Stab(\aa)$
for all $i\ne j$. Moreover, as we explained in Sect.\ \ref{secext}, $\aa$ can be assumed itself invariant under the flip of the two qubits;
so we only need to consider $\aa_k^\pi(n)$ for $k=0,2,4,6,7,14,16,20$ and $\bb_l^\pi(n)$ for $l=0,1,3$.
\end{proof}

We finish this subsection with an important lemma.

\begin{lemma}\label{lemp3}
The Lie algebras $\aa_k(n)$ have trivial centers for $1\le k\le 22$ and $n\ge3$. 
\end{lemma}
\begin{proof}
Due to Lemma \ref{lemgen}, $\aa_k(n)$ has a basis $\mathcal{B}\subseteq i\mcP_n \cap \aa_k(n)$ consisting of Pauli strings.
Suppose that $\aa_k(n)$ has a central element $c\ne 0$, and write $c$ as a linear combination of basis vectors:
\begin{equation*}
c = \sum \alpha_j c_j, \qquad \alpha_j \in\mathbb R, \;\; \alpha_j\ne0, \;\; c_j \in\mathcal{B}.
\end{equation*}
We claim that all $c_j$ in this expression are central too. 

Indeed, suppose that $[b,c_j]\ne 0$ for some index $j$ and a basis vector $b\in\mathcal{B}$.
Since $b,c_j\in i\mcP_n$, we have $[b,[b,c_j]]=-4c_j$ whenever $[b,c_j]\ne 0$, by Corollary \ref{lemp1}.
Hence,
\begin{equation*}
0 = [b,[b,c]] = \sum \alpha_j [b,[b,c_j]] = -4 \sum\nolimits' \alpha_j c_j,
\end{equation*}
where $\sum'$ denotes the sum over all indices $j$ such that $[b,c_j]\ne0$.
This contradicts the fact that all $\alpha_j\ne0$ and the vectors $c_j$ are linearly independent.

Therefore, without loss of generality, we can take $c\in\mathcal{B}$ to be itself one of the basis vectors. 
One can verify by inspection that the generators of $\aa_k(n)$ are not central for $1\le k\le 22$ and $n\ge3$; 
for example, by checking that $Z\bigl(\Stab(\aa_k(n))\bigr)$ does not contain any of the generators of $\aa_k(n)$.
Thus, we can write $c$ in the form
\begin{equation*}
c = \ad_{a_1} \ad_{a_2} \cdots \ad_{a_r}(a_{r+1}),
\end{equation*}
for some $r\ge1$ and generators $a_1,\dots,a_{r+1}$.
Since all generators $a_j \in i\mcP_n$, we have $a :=\ad_{a_2} \cdots \ad_{a_r}(a_{r+1}) \in i\mcP_n$.
Then $c=[a_1,a] \ne 0$, and from Corollary \ref{lemp1}, we get $-4a=[a_1,[a_1,a]] = [a_1,c]=0$, 
which implies that $c=0$, a contradiction.
\end{proof}

\subsection{Upper bounds for \texorpdfstring{$\aa_k(n)$}{a\_k(n)}}\label{secup}

In this subsection, we establish upper bounds for the Lie algebras $\aa_k(n)$, i.e., we find certain subalgebras $\g_k(n)^{\theta_k} \subseteq \su(2^n)$ that contain $\aa_k(n)$.
Then, in the next subsection \ref{seclow}, we will prove that these bounds are exact, that is $\aa_k(n) = \g_k(n)^{\theta_k}$.
While $\aa_k(n)$ is defined in terms of its generators, $\g_k(n)^{\theta_k}$ is defined as the set of elements of $\su(2^n)$ that are fixed under certain automorphisms and involutions. This will allow us, in the following subsection \ref{secgkn}, to identify the Lie algebras $\aa_k(n)$ as direct sums of $\su$, $\so$, and $\sp$ Lie algebras.

We start by recalling that any Pauli string $P\in\mcP_n$ defines an automorphism of $\su(2^n)$ by conjugation $a\mapsto P a P$ (recall that $P=P^\dagger=P^{-1}$); see Lemma \ref{leminv2}.
We will denote by $\su(2^n)^P$ the set of \emph{fixed points} under this
automorphism, i.e., the set of all $a\in\su(2^n)$ such that $P a P=a$. The latter is equivalent to
$P a=a P$; hence $\su(2^n)^P$ is equal to the \emph{centralizer} of $P$, i.e., the set of all $a\in\su(2^n)$ that commute with $P$ (see Remark \ref{remstab2}).
More generally, for a set $\Phi$ of automorphisms of a Lie algebra $\g$, we will
denote by $\g^\Phi$ the set of fixed points $a\in\g$ such that $\phi(a)=a$ for all $\phi\in\Phi$.

Given a subalgebra $\s\subseteq\su(2^n)$, recall from Sect.\ \ref{secstab}, that its stabilizer
$\Stab(\s)$ consists of all Pauli strings $P\in\mcP_n$ such that $[a,P]=0$ for every $a\in\s$.
On the other hand, the centralizer $\su(2^n)^{\Stab(\s)}$ of $\Stab(\s)$ in $\su(2^n)$ consists of all $a\in\su(2^n)$ such that $[a,P]=0$ for every $P\in\Stab(\s)$.
Hence, by definition,
\begin{equation}\label{stabcent1}
\s \subseteq \su(2^n)^{\Stab(\s)}.
\end{equation}
This simple observation will be the key to finding upper bounds for our subalgebras $\aa_k(n)$,
because we have already determined their stabilizers in Proposition \ref{pstab1}.
Here is another observation, which in some cases will allow us to further reduce the upper bound.

\begin{lemma}\label{eq:stabcent2}
For any subalgebra $\s\subseteq\su(2^n)$, we have
\begin{equation*}
\Stab(\s) \cap \su(2^n)^{\Stab(\s)} = 
Z\bigl(\Stab(\s)\bigr) \setminus\{I^{\otimes n}\} \subseteq Z\bigl(\su(2^n)^{\Stab(\s)}\bigr),
\end{equation*}
where $Z(G)$ denotes the center of a group or an algebra $G$.
\end{lemma}
\begin{proof}
Elements $z \in \Stab(\s) \cap \su(2^n)^{\Stab(\s)}$ satisfy $[z,a]=0$ for every $a\in\su(2^n)^{\Stab(\s)}$ since $z\in\Stab(\s)$,
and $[z,P]=0$ for every $P\in\Stab(\s)$ since $z\in\su(2^n)^{\Stab(\s)}$.
Hence, such $z$ are central 
in both $\Stab(\s)$ and $\su(2^n)^{\Stab(\s)}$. However, $I^{\otimes n}$ is excluded, because $I^{\otimes n} \not\in \su(2^n)$.
\end{proof}

As the Lie algebras $\s=\aa_k(n)$ for $1\le k\le 22$, $n\ge3$ have trivial centers (Lemma \ref{lemp3}), for them we can reduce the upper bound $\su(2^n)^{\Stab(\s)}$ 
if we quotient by the central elements $Z\bigl(\Stab(\s)\bigr) \setminus\{I^{\otimes n}\}$.
Thus, we introduce the notation
\begin{equation}\label{stabcent3}
\g_k(n) := \su(2^n)^{\Stab(\aa_k(n))} \big/ \Span\bigl( Z\bigl(\Stab(\aa_k(n))\bigr) \setminus\{I^{\otimes n}\} \bigr),
\end{equation}
and from the above discussion, we have
\begin{equation}\label{stabcent4}
\aa_k(n) \subseteq \g_k(n), \qquad 1\le k\le 22, \quad n\ge 3.
\end{equation}
Using Proposition \ref{pstab1} and Lemma \ref{lemstab}, we can write explicitly:
\begin{align}
\label{listgkn1}
\g_3(n) &= \g_6(n) = 
\begin{cases} 
\su(2^n)^{\{P_X,P_{YZ},P_{ZY}\}} / \Span\{P_X,P_{YZ},P_{ZY}\}, & n \;\;\mathrm{even}, \\
\su(2^n)^{\{P_X,P_{YZ},P_{ZY}\}}, & n \;\;\mathrm{odd},
\end{cases} \\
\label{listgkn2} 
\g_5(n) &= \g_{10}(n) = 
\begin{cases} 
\su(2^n)^{\{P_{XYZ},P_{YZX},P_{ZXY}\}} / \Span\{P_{XYZ},P_{YZX},P_{ZXY}\}, & n \;\;\mathrm{even}, \\
\su(2^n)^{\{P_{XYZ},P_{YZX},P_{ZXY}\}}, & n \;\;\mathrm{odd},
\end{cases} \\
\label{listgkn3} 
\g_7(n) &= 
\begin{cases} 
\su(2^n)^{\{P_X,P_Y,P_Z\}} / \Span\{P_X,P_Y,P_Z\}, & n \;\;\mathrm{even}, \\
\su(2^n)^{\{P_X,P_Y,P_Z\}}, & n \;\;\mathrm{odd},
\end{cases} \\
\label{listgkn4} 
\g_9(n) &= \su(2^n)^{\{X_1,Y_1X_2,Z_1X_2\}}, \\
\label{listgkn5} 
\g_{11}(n) &= \g_{16}(n) = \su(2^n), \\ 
\label{listgkn6} 
\g_{13}(n) &= \g_{20}(n) = \su(2^n)^{P_X} / \Span\{P_X\}, \\
\label{listgkn7} 
\g_{15}(n) &= \su(2^n)^{X_1} / \Span\{X_1\}.
\end{align}

It turns out that in some cases the inclusions \eqref{stabcent4} are strict, and we need to reduce the Lie algebras $\g_k(n)$ further to smaller subalgebras. 
We do that by finding suitable involutions and then taking their fixed points (see Sect.\ \ref{secpauli3}).

\begin{theorem}\label{thmtheta}
We have
\begin{equation}\label{stabcent6}
\aa_k(n) = \g_k(n), \qquad k=6,7,10,13,15,20, \quad n\ge 3.
\end{equation}
In the remaining cases,
there exists an involution $\theta_k$ of\/ $\g_k(n)$, such that
\begin{equation}\label{stabcent5}
\aa_k(n) = \g_k(n)^{\theta_k}, \qquad k=3,5,9,11,16, \quad n\ge 3,
\end{equation}
is the set of fixed points under $\theta_k$. 
\end{theorem}

In the remainder of this subsection, we will construct the involution $\theta_k$ explicitly, 
and will check that $\aa_k(n) \subseteq \g_k(n)^{\theta_k}$.
The opposite inclusion will be proved in the next subsection. Then, in Sect.\ \ref{secgkn}, we will identify the Lie algebras $\g_k(n)^{\theta_k}$ with those
from Theorem \ref{the:classification}.
For $k=9,11,16$, the involution $\theta_k$ will have the form (cf.\ Lemma \ref{leminv2}):
\begin{equation}\label{thetaq}
\theta(g) = -Q g^T Q
\end{equation}
for some given Pauli string $Q\in\mcP_n$. 

\begin{lemma}\label{lemtheta1}
For any fixed Pauli string $Q\in\mcP_n$, \eqref{thetaq} defines an involution of $\su(2^n)$, which restricts to an involution of $\g_k(n)$ for all $k$.
\end{lemma}
\begin{proof}
We already know from Lemma \ref{leminv2} that $\theta$ is an involution of $\su(2^n)$, so we only need to check that $\theta(g) \in \g_k(n)$ for all $g\in \g_k(n)$. 
As before, write $\s=\aa_k(n)$ for short. Consider an element 
$g\in \su(2^n)^{\Stab(\s)}$, which means that $[g,P]=0$ for all $P\in\Stab(\s)$. Then
\begin{equation*}
\theta(P) = -Q P^T Q = \pm P,
\end{equation*}
because $P^T=\pm P$ and $PQ = \pm QP$ for any two Pauli strings $P,Q\in\mcP_n$
(the signs here are not coordinated).
Hence
\begin{equation*}
[\theta(g),P] = \pm [\theta(g),\theta(P)] = \pm \theta([g,P]) = 0,
\end{equation*}
which implies that $\theta(g)\in \su(2^n)^{\Stab(\s)}$.
Furthermore,
\begin{equation*}
\theta(P) = \pm P \in \Span\bigl( Z\bigl(\Stab(\s)\bigr) \setminus\{I^{\otimes n}\} \bigr) \quad\text{for all}\quad 
P \in Z\bigl(\Stab(\s)\bigr) \setminus\{I^{\otimes n}\}.
\end{equation*}
Therefore, $\theta(g) \in \g_k(n)$ for $g\in \g_k(n)$.
\end{proof}

Let us record the following consequence of the proof of Lemma \ref{lemtheta1}, which will be useful later.

\begin{corollary}\label{cor:gknth}
Every element of $\g_k(n)^{\theta_k}$ is a linear combination of Pauli strings that are themselves in $\g_k(n)^{\theta_k}$, i.e.,
\begin{equation*}
\g_k(n)^{\theta_k} = \Span_{\mathbb R} \bigl( i\mcP_n \cap \g_k(n)^{\theta_k} \bigr).
\end{equation*}
\end{corollary}
\begin{proof}
We saw in the proof of Lemma \ref{lemtheta1} that $\theta_k(P) = \pm P$ for any Pauli string $P\in\mcP_n$. Similarly, for any given $S\in \Stab(\aa_k(n))$,
we have $SPS = \pm P$. Elements $g$ of $\g_k(n)^{\theta_k}$ are determined by the conditions
\begin{equation*}
SgS = g = \theta_k(g) \qquad\text{for all}\quad S\in \Stab(\aa_k(n)).
\end{equation*}
Writing $g\in\su(2^n)$ as $i$ times a real linear combination of Pauli strings, we see that $g$ satisfies these conditions if and only if every summand does.
\end{proof}

Now we go back to the construction of the involutions $\theta_k$.

\begin{lemma}\label{lemtheta2}
For $k=9,11,16$, we define $\theta_k(g) = -Q_k g^T Q_k$, where the Pauli strings $Q_k$ are given as follows:
\begin{align}
Q_9 &= IYZ \cdots Z, \label{thetaq9} \\
Q_{11} &= Q_{16} = I \cdots I \quad\Rightarrow\quad \theta_{11}(g) = \theta_{16}(g) = -g^T. \label{thetaq11}
\end{align}
Then $\aa_k(n) \subseteq \g_k(n)^{\theta_k}$.
\end{lemma}
\begin{proof}
We already know that $\aa_k(n) \subseteq \g_k(n)$, so we only need to check that $\theta_k(g)=g$ for all $g \in \aa_k(n)$.
It is enough to check this only for the generators $g$ of $\aa_k(n)$, because $\theta_k([a,b]) = [\theta_k(a),\theta_k(b)]$.

For $k=9$, we take $g=X_i Y_{i+1}$ or $g=X_i Z_{i+1}$. In the first case, $g^T=-g$ and $g Q_9 = Q_9 g$; while in the second case,
$g^T=g$ and $g Q_9 = -Q_9 g$. Hence, in both cases we have $\theta_9(g)=g$.

For $k=11$, the generators are $g = X_i Y_{i+1}, Y_i X_{i+1}, Y_i Z_{i+1}$; while for $k=16$, the generators are
$g = X_i Y_{i+1}, Y_i X_{i+1}, Y_i Z_{i+1}, Z_i Y_{i+1}$. All of them satisfy $g^T=-g$.
\end{proof}

The remaining cases $k=3$ and $k=5$ are a little more complicated. The trick is to first embed $\aa_3(n)$ and $\aa_5(n)$ as subalgebras of $\aa_7(n)$.
Recall from Sect.\ \ref{seciso} that $\aa_3(n) \subseteq \aa_6(n)$ and $\aa_6(n) \cong \aa_7(n)$
under the automorphism $\varphi_n$ of $\su(2^n)$ that swaps (up to a sign) $Y \rightleftharpoons Z$ on all even qubits (see \eqref{varphi1}, \eqref{varphi2}).
Then $\tilde\aa_3(n) := \varphi_n\aa_3(n) \subset \aa_7(n)$. Likewise, we have $\aa_5(n) \subseteq \aa_{10}(n)$ and $\aa_{10}(n) \cong \aa_7(n)$
under the automorphism $\gamma_n$ of $\su(2^n)$ that applies on the $j$-th qubit $\gamma^{j}$, where $\gamma$ is the cycle $X \mapsto Z \mapsto Y \mapsto X$ 
(see \eqref{rho1}, \eqref{rho2}).
Then $\tilde\aa_5(n) := \gamma_n\aa_5(n) \subset \aa_7(n)$. 

We consider the involutions
\begin{equation}\label{tildethetaq35}
\tilde\theta_k(g) = -Q_k g^T Q_k, \qquad k=3,5,
\end{equation}
where
\begin{align}
Q_3 &= P_{ZIYX} = (Z_1 Y_3 X_4) (Z_5 Y_7 X_8) (Z_9 Y_{11} X_{12}) \cdots\,, \label{thetaq3} \\
Q_5 &= P_{IYZ} = (Y_2 Z_3) (Y_5 Z_6) (Y_8 Z_9) (Y_{11} Z_{12}) \cdots\,. \label{thetaq5}
\end{align}
Then we define
\begin{equation}\label{thetaq35}
\theta_3 := \varphi_n^{-1}\tilde\theta_3\varphi_n, \qquad 
\theta_5 := \gamma_n^{-1}\tilde\theta_5\gamma_n.
\end{equation}

\begin{lemma}\label{lemtheta3}
For $k=3,5,$ and $\theta_k$ defined as above, we have $\aa_k(n) \subseteq \g_k(n)^{\theta_k}$.
\end{lemma}
\begin{proof}
As we already know that $\aa_k(n) \subseteq \g_k(n)$, we only need to check that $\theta_k(g)=g$ for all generators $g$ of $\aa_k(n)$.
By conjugation, it is equivalent to check that $\tilde\theta_k(g)=g$ for the generators $g$ of $\tilde\aa_k(n)$.
Applying $\varphi_n$ to the generators of $\aa_3(n)$, we find that the generators of $\tilde\aa_3(n)$ are:
\begin{align*}
&X_1X_2, X_2X_3, X_3X_4, X_4X_5, X_5X_6, X_6X_7, X_7X_8, \dots, \\
&Y_1Y_2, Z_2Z_3, Y_3Y_4, Z_4Z_5, Y_5Y_6, Z_6Z_7, Y_7Y_8, \dots \;.
\end{align*}
Similarly, applying $\gamma_n$ to the generators of $\aa_5(n)$, we find the generators of $\tilde\aa_5(n)$:
\begin{align*}
&Z_1Z_2, Y_2Y_3, X_3X_4, Z_4Z_5, Y_5Y_6, X_6X_7, Z_7Z_8, Y_8Y_9, X_9X_{10}, \dots, \\
&X_1X_2, Z_2Z_3, Y_3Y_4, X_4X_5, Z_5Z_6, Y_6Y_7, X_7X_8, Z_8Z_9, Y_9Y_{10}, \dots \;.
\end{align*}
We observe that all generators $g$ above satisfy
$g^T=g$ and $g Q_k = -Q_k g$; hence, $\tilde\theta_k(g)=g$.
\end{proof}

\subsection{Lower bounds for \texorpdfstring{$\aa_k(n)$}{a\_k(n)}}\label{seclow}

In this subsection, we prove that the upper bounds $\aa_k(n) \subseteq \g_k(n)^{\theta_k}$
established in Sect.\ \ref{secup} are exact. The proof will rely on the next lemma.

\begin{lemma}\label{lemp2}
Let $\s$ be a Lie subalgebra of $\su(2^n)$. If $\ad_{a_1} \cdots \ad_{a_r}(b) \in\s\setminus\{0\}$ for some Pauli strings $a_1,\dots,a_r$ $\in$ $i\mcP_n\cap\s$ and $b\in i\mcP_n$, then $b\in \s$.
\end{lemma}
\begin{proof}
Using induction on $r$, it is enough to prove the statement for $r=1$. In this case, it follows from Corollary \ref{lemp1}:
$[a_1,b] \in\s\setminus\{0\}$ implies that $b=-4[a_1,[a_1,b]] \in\s$.
\end{proof}

In order to prove that $\aa_k(n) = \g_k(n)^{\theta_k}$, we want to show that every element $b\in\g_k(n)^{\theta_k}$ is in $\aa_k(n)$.
Since, by Corollary \ref{cor:gknth}, $b$ is a linear combination of Pauli strings that are themselves in $\g_k(n)^{\theta_k}$,
we can assume without loss of generality that $b\in i\mcP_n \cap \g_k(n)^{\theta_k}$.
Then the strategy of the proof is to take suitable commutators of $b$ with elements of $i\mcP_n\cap\aa_k(n)$ to produce a Pauli string $c\in i\mcP_n\cap\g_k(n)^{\theta_k}$ 
that has $I$ in one of its positions. Erasing the $I$ will give an element of $\g_k(n-1)^{\theta_k}$, which by induction will be in $\aa_k(n-1)$.
From here, we will obtain that $c\in\aa_k(n)$, and then we can conclude that $b\in\aa_k(n)$ due to Lemma \ref{lemp2}.

In order to realize the above strategy, we will have to do a detailed case-by-case analysis. We start with the cases $k=3,5,7$, for which we need the following lemmas.

\begin{lemma}\label{lem1}
We have vector space decompositions:
\begin{align*}
\aa_3(4) &= (I \otimes \aa_3(3)) + P_X \cdot (I \otimes \aa_3(3)) + P_{YZ} \cdot (I \otimes \aa_3(3)) + P_{ZY} \cdot (I \otimes \aa_3(3)), \\
\aa_5(6) &= (I \otimes \aa_5(5)) + P_{XYZ} \cdot (I \otimes \aa_5(5)) + P_{YZX} \cdot (I \otimes \aa_5(5)) + P_{ZXY} \cdot (I \otimes \aa_5(5)), \\
\aa_7(4) &= (I \otimes \aa_7(3)) + P_X \cdot (I \otimes \aa_7(3)) + P_{Y} \cdot (I \otimes \aa_7(3)) + P_{Z} \cdot (I \otimes \aa_7(3)),
\end{align*}
where $\cdot$ denotes the componentwise matrix product.
\end{lemma}
\begin{proof}
By inspection. Here are all Pauli strings in $\aa_3(4)$ (after multiplication by $i$):
\begin{align*}
&IIXX,&&	IIYZ,&&	IXXI,&&	IXZZ,&&	IYIY,&&	IYXZ,&&	IYYX,&&	IYZI,&&	IZIZ,&&	IZXY,	\\
&XIIX,&&	XIYY,&&	XXII,&&	XXZY,&&	XYIZ,&&	XYXY,&&	XZIY,&&	XZXZ,&&	XZYX,&&	XZZI,	\\
&YIYI,&&	YIZX,&&	YXIY,&&	YXXZ,&&	YXYX,&&	YXZI,&&	YYXI,&&	YYZZ,&&	YZII,&&	YZZY,	\\
&ZIIY,&&	ZIXZ,&&	ZIYX,&&	ZIZI,&&	ZXYI,&&	ZXZX,&&	ZYXX,&&	ZYYZ,&&	ZZIX,&&	ZZYY.
\end{align*}
The Pauli strings in $\aa_7(4)$ (after multiplication by $i$) are:
\begin{align*}
&IIXX,&&	IIYY,&&	IIZZ,&&	IXIX,&&	IXXI,&&	IXYZ,&& IXZY,&& IYIY, \\
& IYXZ, && IYYI,&& IYZX,&& IZIZ,&& IZXY,&& IZYX,&& IZZI, \\
&XIIX,&&	XIXI,&&	XIYZ,&& XIZY,&& XXII,&& XXYY,&& XXZZ,&& XYIZ, \\
& XYXY,&& XYYX,&& XYZI,&& XZIY,&& XZXZ,&& XZYI,&& XZZX, \\
&YIIY,&&	YIXZ,&& YIYI	,&& YIZX,&& YXIZ,&& YXXY,&& YXYX,&& YXZI, \\
& YYII, && YYXX,&& YYZZ,&& YZIX,&& YZXI,&& YZYZ,&& YZZY, \\
&ZIIZ,&&	ZIXY,&& ZIYX,&& ZIZI,&& ZXIY,&& ZXXZ,&& ZXYI,&&	ZXZX, \\
& ZYIX,&& ZYXI,&& ZYYZ,&&	ZYZY,&& ZZII,&& ZZXX,&& ZZYY.
\end{align*}
We have: $|\aa_5(5)| = 120$, $|\aa_5(6)| = 480$, and there are $120$ elements in $\aa_5(6)$ starting with each of the letters $I,X,Y$, or $Z$.
The remaining claims were verified using Excel.
\end{proof}

For each pair $(k,n)=(3,4),(5,6),(7,4)$, consider the subalgebra $\s=\aa_k(n) \subset\su(2^n)$. 
Recall from Proposition \ref{pstab1} that the stabilizer $\Stab(\s)$ is given by:
\begin{align*}
\Stab(\aa_3(4)) &= \{P_I,P_X,P_{YZ},P_{ZY}\}, \\
\Stab(\aa_5(6)) &= \{P_I,P_{XYZ},P_{YZX},P_{ZXY}\}, \\
\Stab(\aa_7(4)) &= \{P_I,P_X,P_Y,P_Z\},
\end{align*}
and $\Stab(\s)$ is an Abelian group under the matrix product $\cdot$.
We can state Lemma \ref{lem1} succinctly as
\begin{equation}\label{scsis}
\aa_k(n) = \Stab(\aa_k(n)) \cdot (I \otimes \aa_k(n-1)), \qquad (k,n)=(3,4),(5,6),(7,4).
\end{equation}

\begin{lemma}\label{lem2}
Let $\s=\aa_k(n)$ for $(k,n)=(3,4),(5,6),(7,4)$.
Consider any Pauli string $a\in i\mcP_n\cap\g_k(n)$ not starting with $I$ in the first qubit. Then there exists a basis vector $b\in \s$ such that 
$[a,b] \ne 0$ and $[a,b]\in I\otimes\su(2^{n-1})$.
\end{lemma}
\begin{proof}
Let us write all Pauli strings up to a suitable multiple of $i$ that makes them skew-Hermitian.
Consider the case when $a=XA$ starts with $X$; the cases when it starts with $Y$ or $Z$ are similar. If there is $b=XB \in \s$ such that
$[A,B]\ne 0$, then $[a,b]=I[A,B] \ne 0$ and we are done.
By Lemma \ref{lem1}, any $b=XB \in \s$ can be written in the form $b=C \cdot (ID)$, where $D\in \aa_k(n-1)$ and $C \in \Stab(\s)$;
explicitly, $C=P_X,P_{XYZ},P_X$ for $\s=\aa_3(4),\aa_5(6),\aa_7(4)$, respectively. 
Similarly, as $C \cdot C=P_I$, we can write $a=C \cdot (IE)$ for some $E \in\mcP_{n-1}$.
Suppose that $[a,b]=0$.
Then $[C,a]=[C,b]=0$ imply that $[E,D]=0$.
Since this is true for all $D\in \aa_k(n-1)$, it follows that $E\in \Stab(\aa_k(n-1))$, from where $a=C \cdot (IE) \in\Stab(\s)$. 
This is a contradiction, because such elements are factored out from $\g_k(n)$; see \eqref{stabcent3} and Lemma \ref{eq:stabcent2}.
\end{proof}

\begin{lemma}\label{lem3}
Consider $(k,m)=(3,4),(5,6),(7,4)$, and let $n\ge m$. Then for any Pauli string $a\in i\mcP_n\cap\g_k(n)$,
there exist basis vectors $b_1,\dots,b_r\in \aa_k(n)$, $r\ge0$, such that 
$\ad_{b_1} \cdots \ad_{b_r}(a) \in (I^{\otimes (n-m+1)} \otimes\su(2^{m-1})) \setminus\{0\}$
$($with $r=0$ corresponding to $a)$.
\end{lemma}
\begin{proof}
The proof is by induction on $n$, the base $n=m$ being Lemma \ref{lem2}.
For the step of the induction, suppose that $n>m$ and the statement holds for $\g_k(n-1)$.
Again, let us write all Pauli strings up to a suitable multiple of $i$.
Take any Pauli string $a\in \g_k(n)$, and write it as $a=AD$ where $A$ is the substring consisting of the first $m$ 
Paulis. Then $A \in \su(2^m)^{\Stab(\aa_k(m))}$.

If $A \not\in \Stab(\aa_k(m))$, we can use Lemma \ref{lem2} to find $B \in \aa_k(m)$ such that $[B,A]$ starts with $I$.
Then $[b,a]=IC$ starts with $I$ for $b=BI\cdots I \in \aa_k(n)$.
After that, we can apply the inductive assumption for $C \in \g_k(n-1)$.

If $A \in \Stab(\aa_k(m))$, we repeat the same argument for the substring $E$ of $a$ corresponding to positions $2,\dots,m+1$.
When $E \not\in \Stab(\aa_k(m))$, we can make it to start with $I$, which will force the substring $A \not\in \Stab(\aa_k(m))$.
If both $A,E \in \Stab(\aa_k(m))$, putting them together we get that the first $m+1$ positions of $a$ are in $\Stab(\aa_k(m+1))$.
Continuing this way will give us $a \in \Stab(\aa_k(n))$, which is a contradiction, because such elements are factored out from $\g_k(n)$ (cf.\ \eqref{stabcent3} and Lemma \ref{eq:stabcent2}).
\end{proof}

Recall that, in Sect.\ \ref{seciso}, we constructed certain automorphisms $\varphi_n$ and $\gamma_n$ of $\su(2^n)$ such that
the images $\tilde\aa_3(n) := \varphi_n\aa_3(n)$ and $\tilde\aa_5(n) := \gamma_n\aa_5(n)$ are subalgebras of $\aa_7(n)$.
Note that, after these transformations, their stabilizers become equal:
\begin{equation}\label{stab357}
\Stab(\tilde\aa_3(n)) = \Stab(\tilde\aa_5(n)) = \Stab(\aa_7(n)) = \{P_I,P_X,P_Y,P_Z\}.
\end{equation}
Hence, we have (recall \eqref{tildethetaq35}--\eqref{thetaq35}):
\begin{equation}\label{aknlkn}
\tilde\aa_k(n) \subseteq \g_7(n)^{\tilde\theta_k}, \qquad k=3,5.
\end{equation}

\begin{lemma}\label{lem4}
We have\/ $\aa_7(n)=\g_7(n)$ and equalities in \eqref{aknlkn}.
Consequently, $\aa_k(n) = \g_k(n)^{\theta_k}$ for $k=3,5$.
\end{lemma}
\begin{proof}
As before, let $(k,m)=(3,4),(5,6),(7,4)$. In order to include the case $k=7$ in \eqref{aknlkn}, we let $\tilde\aa_7(n) = \aa_7(n)$
and $\tilde\theta_7$ be the identity. The statement is true for all $2\le n\le m$ by inspection.
For $n\ge m$, we prove it by induction on $n$. Consider any $a\in \g_7(n)^{\tilde\theta_k}$.
By Lemma \ref{lem3}, we can find $b_1,\dots,b_r\in \tilde\aa_k(n)$ such that 
$\ad_{b_1} \cdots \ad_{b_r}(a) = I\cdots I D$ for some $D\in\su(2^{m-1}) \setminus\{0\}$.
Since $b_i \in \tilde\aa_k(n) \subseteq \g_7(n)^{\tilde\theta_k}$, we get that $D\in \g_7(m)^{\tilde\theta_k} = \tilde\aa_k(m)$.
Therefore, $a \in \tilde\aa_k(n)$ due to Lemma \ref{lemp2}.
\end{proof}

\begin{remark}
It follows from Lemma \ref{lem4} that \eqref{scsis} holds for all $(k,n)$ such that: $k=3$, $n\equiv 0 \mod 4;$ $k=5$, $n\equiv 0 \mod 6;$ $k=7$, $n\equiv 0 \mod 2$.
As a consequence, for such $(k,n)$, we have $\aa_k(n) \cong \aa_k(n-1)^{\oplus 4}$ as a Lie algebra.
\end{remark}

Now that we are done with the cases $k=3,5,7$, we derive the cases $k=6,10$ from $k=7$ and the isomorphisms $\aa_6(n) \cong \aa_{10}(n) \cong \aa_7(n)$
obtained in Lemmas \ref{leminc4} and \ref{leminc5}.

\begin{lemma}\label{lem4a}
We have\/ $\aa_k(n)=\g_k(n)$ for $k=6,10$.
\end{lemma}
\begin{proof}
Recall from Sect.\ \ref{seciso} that we have an isomorphism $\varphi_n\colon \aa_6(n) \cong \aa_7(n)$ that up to a sign swaps $Y \rightleftharpoons Z$ on every even qubit
(see \eqref{varphi1}, \eqref{varphi2}). Under $\varphi_n$ the stabilizers
\begin{align*}
\Stab(\aa_6(n)) &= \{P_I,P_X,P_{YZ},P_{ZY}\}, \\
\Stab(\aa_7(n)) &= \{P_I,P_X,P_Y,P_Z\}
\end{align*}
are sent to each other; hence $\g_6(n) \cong \g_7(n)$. 
Since $\aa_7(n)=\g_7(n)$ by Lemma \ref{lem4}, it follows that $\aa_6(n)=\g_6(n)$.

Similarly, we have an isomorphism $\gamma_n\colon \aa_{10}(n) \cong \aa_7(n)$ given by applying on the $j$-th qubit ($j=1,\dots,n$) the permutation $\gamma^{j}$, 
where $\gamma$ is the cycle $X \mapsto Z \mapsto Y \mapsto X$ (see \eqref{rho1}, \eqref{rho2}).
Then $\gamma_n$ sends
\begin{equation*}
\Stab(\aa_{10}(n)) = \{P_I,P_{XYZ},P_{YZX},P_{ZXY}\}
\end{equation*}
to $\Stab(\aa_7(n))$; hence $\g_{10}(n) \cong \g_7(n)$ and $\aa_{10}(n)=\g_{10}(n)$.
\end{proof}

Now we consider the subalgebra $\aa_9(n)$. In this case, we have the involution $\theta_9(g) = -Q_9 g^T Q_9$, where $Q_9$ is given by \eqref{thetaq9}.

\begin{lemma}\label{lem5}
We have\/ $\aa_9(n) = \g_9(n)^{\theta_9}$.
\end{lemma}
\begin{proof}
The claim is true for $n=2$ and $3$ by comparing the dimensions. Suppose by induction that the statement is true for $\aa_9(n-1)$, and consider a Pauli string 
$a\in \g_9(n)^{\theta_9}$ for $n\ge4$. Again, we will omit the multiples of $i$ that make Pauli strings skew-Hermitian.

If $a$ ends with $I$, we can write $a=AI$ for some $A \in \g_9(n-1)^{\theta_9}$ and apply the inductive assumption.
Similarly, if $a=AIB$ has an $I$ in the $j$-th position for some $j\ge3$, we can delete it and get an element $AB \in \g_9(n-1)^{\theta_9}$,
which by induction is in $\aa_9(n-1)$. Then $a\in \aa_9(n)$, because $\aa_9(4)$ contains $IXIY$ and $IXIZ$, 
which generate elements of $\aa_9(n)$ with $I$ in the middle.

Suppose that $a$ has no $I$ in positions $3,\dots,n$. If $a=AXXB$ contains $XX$ in positions $j,j+1$, then $[X_j Y_{j+1},a] = -2i AIZB \in \aa_9(n)$. 
Since $X_j Y_{j+1} \in \aa_9(n)$, by Lemma \ref{lemp2}, we get that $a\in \aa_9(n)$.
So, if $a\not\in \aa_9(n)$ contains an $X$, then on the left of it must have a $Y$ or $Z$.
Then we can use $[XZ,XY]=-2iIX$, $[XY,YY]=2iZI$, $[XY,ZY]=-2iYI$ when $a$ contains a $Y$, and
$[XY,XZ]=2iIX$, $[XZ,YZ]=2iZI$, $[XZ,ZZ]=-2iYI$ when $a$ contains a $Z$.
\end{proof}

Finally, let us briefly discuss the remaining easier cases, $k=11,13,15,16,20$. Recall from \eqref{ext3}, \eqref{ext4} that
$\aa_{11}(n) = \aa_{16}(n)$ for $n\ge4$ and $\aa_{13}(n) = \aa_{20}(n)$ for $n\ge 3$.
Moreover, $\g_{11}(n) = \g_{16}(n)$ and $\g_{13}(n) = \g_{20}(n)$, because they have equal stabilizers by Proposition \ref{pstab1}.
Thus, we are left to consider only $k=13,15,16$.

\begin{lemma}\label{lem6}
We have\/ $\aa_{13}(n) = \g_{13}(n) = \su(2^n)^{P_X} / \Span\{P_X\}$ for $n\ge3$.
\end{lemma}
\begin{proof}
We know that $\Stab(\aa_{13}(n)) = \{P_I,P_X\}$ and $\aa_{13}(n) \subseteq \g_{13}(n)$. The proof of the opposite inclusion is similar to the proof of Lemma \ref{lem5}.
Consider a Pauli string $a\in \g_{13}(n)$ for $n\ge4$. 
If $a=AIB$ has an $I$ in the $j$-th position for some $1\le j\le n$, we can delete it and get an element $AB \in \g_{13}(n-1)$,
which by induction is in $\aa_{13}(n-1)$. Then $a\in \aa_{13}(n)$, because $\aa_{13}(3)$ contains $XIX, YIY, YIZ$, 
and these generate elements of $\aa_{13}(n)$ with $I$ in the middle.
If $a$ has no $I$'s, we can use commutators with the generators of $\aa_{13}(n)$ to produce one.
Then again we can apply Lemma \ref{lemp2}.
\end{proof}

\begin{lemma}\label{lem7}
We have\/ 
$\aa_{16}(n) = \su(2^n)^{\theta_{16}} = \so(2^n)$, where $\theta_{16}(g)=-g^T$.
\end{lemma}
\begin{proof}
The same as the proof of Lemma \ref{lem6}, using that $AIB \in\aa_{16}(3)$ for every generator $AB$ of $\aa_{16}$.
Indeed, one checks that $\aa_{16} = \Lie{XY, YX, YZ, ZY}$ and $XIY, YIX, YIZ, ZIY \in \aa_{16}(3)$; see Sect.\ \ref{secsu8}.
\end{proof}

\begin{lemma}\label{lem15}
We have\/ $\aa_{15}(n) = \g_{15}(n) = \su(2^n)^{X_1} / \Span\{X_1\}$.
\end{lemma}
\begin{proof}
Note that $\su(2^n)^{X_1}$ is the span of all Pauli strings $\ne I^{\otimes n}$ that start with $I$ or $X$. As in the proof of Lemma \ref{lem6}, 
pick any Pauli string $a\in \g_{15}(n)$ for $n\ge3$. If $a=AIB$ has an $I$ in the $j$-th position for some $2\le j\le n$, we can delete it and get an element $AB \in \g_{15}(n-1)$,
which by induction is in $\aa_{15}(n-1)$. The rest of the proof is the same,
using that $\aa_{15} = \Lie{XX, XY, XZ}$ and $XIX,XIY,XIZ \in \aa_{15}(3)$; see Sect.\ \ref{secsu8}.
\end{proof}

Combining the results of Sect.\ \ref{secup} and \ref{seclow} completes the proof of Theorem \ref{thmtheta}.

\subsection{Identifying the Lie algebras \texorpdfstring{$\g_k(n)^{\theta_k}$}{g\_k(n) theta\_k}}\label{secgkn}

In this subsection, we finish the proof of Theorem \ref{the:classification}, by identifying the Lie algebras $\g_k(n)^{\theta_k}$ from Theorem \ref{thmtheta} with the Lie algebras appearing in the right-hand sides in Theorem \ref{the:classification}. As in Theorem \ref{thmtheta}, we only consider the cases $k=3,5,6,7,9,10,11,13,15,16,20$. Moreover, due to the isomorphisms $\aa_6(n) \cong \aa_7(n) \cong \aa_{10}(n)$ and the equalities $\aa_{11}(n) = \aa_{16}(n)$ and $\aa_{13}(n) = \aa_{20}(n)$ (see Lemmas \ref{leminc2}, \ref{leminc4}, \ref{leminc5}), we can omit the cases $k=6,10,11,20$.

The case $k=16$ is obvious, because $\g_{16}(n) = \su(2^n)$ and $\theta_{16}(g)=-g^T$, leading to $\aa_{16}(n) = \so(2^n)$.
Two other easy cases, $k=15$ and $k=13$, are treated in the next lemma.

\begin{lemma}\label{lemgkn1}
We have:
\begin{align*}
\aa_{15}(n) &= \g_{15}(n) = \su(2^n)^{X_1} / \Span\{X_1\} \cong \su(2^{n-1}) \oplus \su(2^{n-1}), \\
\aa_{13}(n) &= \g_{13}(n) = \su(2^n)^{P_X} / \Span\{P_X\} \cong \su(2^{n-1}) \oplus \su(2^{n-1}).
\end{align*}
\end{lemma}
\begin{proof}
Note that 
\begin{equation*}
\su(2^n)^{X_1} / \Span\{X_1\} \cong \Span_{\mathbb R}\{I,X\} \otimes \su(2^{n-1})
\end{equation*}
has a basis consisting of all Pauli strings $\ne I^{\otimes n}, X_1$ that start with $I$ or $X$. Consider the projections $P_\pm$ onto the eigenspaces of $X$, given by $P_\pm := (I\pm X)/2$. They satisfy the identities:
\begin{equation*}
P_\pm \cdot P_\pm = P_\pm, \qquad P_+ \cdot P_- = 0, \qquad P_+ + P_- = I.
\end{equation*}
Then the map
\begin{equation*}
(a,b) \mapsto P_+ \otimes a + P_- \otimes b
\end{equation*}
is a Lie algebra isomorphism from $\su(2^{n-1}) \oplus \su(2^{n-1})$ to $\Span_{\mathbb R}\{I,X\} \otimes \su(2^{n-1})$.
This proves the claim about $\aa_{15}(n)$.

For the case $\aa_{13}(n)$, we can replace $X_1$ with $P_X$ because there exists a unitary transformation $U$ such that $P_X = U X_1 U^\dagger$. For example, we can take
\begin{equation*}
U = e^{-i\frac{\pi}{4} Y_1} e^{i\frac{\pi}{4} Y \otimes X^{\otimes(n-1)}}; 
\end{equation*}
then using \eqref{equab} we check that indeed $U X_1 U^\dagger = X\otimes X^{\otimes(n-1)} = P_X$.
The automorphism $a\mapsto UaU^\dagger$ of $\su(2^n)$ sends $\su(2^n)^{X_1}$ onto $\su(2^n)^{P_X}$, and $\aa_{15}(n)$ onto $\aa_{13}(n)$.
Therefore, $\aa_{13}(n) \cong \aa_{15}(n)$.
\end{proof}

We are left with the cases $k=3,5,7,9$, and we consider $k=9$ next.

\begin{lemma}\label{lemgkn2}
We have\/ $\aa_9(n) = \g_9(n)^{\theta_9} \cong \sp(2^{n-2})$.
\end{lemma}
\begin{proof}
Recall that $\g_9(n) = \su(2^n)^{\{X_1,Y_1X_2,Z_1X_2\}}$. Since $\Span\{X_1,Y_1X_2,Z_1X_2\} \cong \su(2)$, we can find a unitary transformation that takes
this Lie algebra to $\Span\{X_1,Y_1,Z_1\}$. Explicitly, similarly to the proof of Lemma \ref{lemgkn1}, let
$U = e^{i\frac{\pi}{4} X_1 X_2}$.
Then using \eqref{equab}, one easily checks that
\begin{equation*}
U X_1 U^\dagger = X_1, \qquad U Y_1 X_2 U^\dagger = -Z_1, \qquad U Z_1 X_2 U^\dagger = Y_1.
\end{equation*}
Therefore, the map $a\mapsto U a U^\dagger$ restricts to a Lie algebra isomorphism from $\g_9(n)$ to 
\begin{equation*}
\su(2^n)^{\{X_1,Y_1,Z_1\}} = I \otimes \su(2^{n-1}) \cong \su(2^{n-1}).
\end{equation*}

According to Lemmas \ref{leminv1} and \ref{leminv2}, under the transformation $a\mapsto U a U^\dagger$,
the fixed-point subalgebra $\g_9(n)^{\theta_9}$ is sent to the fixed points of the following involution:
\begin{equation*}
a \mapsto - (U Q_9 U^T) a^T (U Q_9 U^T)^\dagger.
\end{equation*}
Recalling that $Q_9=Y_2Z_3 \cdots Z_n$ (see \eqref{thetaq9}), we find from $U^T=U$ and $e^{i\frac{\pi}{4} X} Y e^{i\frac{\pi}{4} X} = Y$ that
\begin{equation*}
U Q_9 U^T = Q_9.
\end{equation*}
Hence, the image of $\g_9(n)^{\theta_9}$ under $a\mapsto UaU^\dagger$ consists of all $b \in I \otimes \su(2^{n-1})$ that are fixed by $\theta_9$.
Writing $b=I \otimes c$ with $c\in\su(2^{n-1})$, the condition $b=\theta_9(b)$ is equivalent to $c=-Q c^T Q$, where $Q=YZ\cdots Z \in\mcP_{n-1}$. 
Since $Q^T=-Q$, this determines the Lie algebra $\sp(2^{n-2})$, due to Corollary \ref{corinv4}.
\end{proof}

Next we consider the case $k=7$.

\begin{lemma}\label{lemgkn3}
We have\/ $\aa_7(n) = \g_7(n) \cong \begin{cases} 
\su(2^{n-1}), & n \;\;\mathrm{odd}, \\
\su(2^{n-2})^{\oplus 4}, & n \ge 4 \;\;\mathrm{even}.
\end{cases}$
\end{lemma}
\begin{proof}
Recall that $\Stab(\aa_7(n)) = \{P_I,P_X,P_Y,P_Z\}$. Since $P_X \cdot P_Y = i^n P_Z$, elements that commute with $P_X$ and $P_Y$ will commute with $P_Z$ as well.
Hence, $\su(2^n)^{\Stab(\aa_7(n))} = \su(2^n)^{\{P_X,P_Y\}}$. Recall also that $[P_X,P_Y]=0$ if and only if $n$ is even; in that case, $\su(2^n)^{\{P_X,P_Y\}}$ has a center spanned by $P_X,P_Y,P_Z$ and we need to quotient by it to obtain $\g_7(n)$ (cf.\ \eqref{listgkn3}).

In order to determine the fixed points under $P_X$ and $P_Y$, we transform them as in the proof of Lemma \ref{lemgkn1}. Consider the unitary operator
\begin{equation}\label{eqgkn}
U = \begin{cases}
e^{i\frac{\pi}{4}Z\otimes Y^{\otimes(n-1)}} e^{i\frac{\pi}{4}Y\otimes X^{\otimes(n-1)}}, & n \;\;\mathrm{odd}, \\
e^{i\frac{\pi}{4}X_2} e^{i\frac{\pi}{4} I\otimes X\otimes Z^{\otimes(n-2)}} e^{i\frac{\pi}{4}Y\otimes X^{\otimes(n-1)}}, & n \;\;\mathrm{even}.
\end{cases}
\end{equation}
Using \eqref{equab}, one checks that
\begin{align}
U P_X U^\dagger &= Z_1, & U P_Y U^\dagger &= X_1 & \text{for } &n \;\;\mathrm{odd},  \label{eqgkn1} \\
U P_X U^\dagger &= Z_1, & U P_Y U^\dagger &= (-1)^{(n+2)/2} Z_2 & \text{for } &n \;\;\mathrm{even}. \label{eqgkn2}
\end{align}
Indeed, we have
\begin{align*}
e^{i\frac{\pi}{4} Y\otimes X^{\otimes(n-1)}} P_X e^{-i\frac{\pi}{4}Y\otimes X^{\otimes(n-1)}}
= i (Y\otimes X^{\otimes(n-1)}) \cdot P_X = Z_1.
\end{align*}
Since the other factors of $U$ commute with $Z_1$, we obtain that $U P_X U^\dagger = Z_1$. The calculation of $U P_Y U^\dagger$ is similar.
When $n$ is odd, $Y\otimes X^{\otimes(n-1)}$ commutes with $P_Y$, and we get from \eqref{equab}:
\begin{align*}
U P_Y U^\dagger = e^{i\frac{\pi}{4}Z\otimes Y^{\otimes(n-1)}} P_Y e^{-i\frac{\pi}{4}Z\otimes Y^{\otimes(n-1)}}
= i (Z\otimes Y^{\otimes(n-1)}) \cdot P_Y = X_1.
\end{align*}
When $n$ is even, after applying \eqref{equab} three times, we obtain:
\begin{align*}
U P_Y U^\dagger &= i^3 X_2 \cdot (I\otimes X\otimes Z^{\otimes(n-2)}) \cdot (Y\otimes X^{\otimes(n-1)}) \cdot P_Y \\
&= i^{n+2} X_2 \cdot (I\otimes X\otimes Z^{\otimes(n-2)}) \cdot (I\otimes Z^{\otimes(n-1)}) \\
&= i^{n+1} X_2 \cdot Y_2 = i^{n+2} Z_2.
\end{align*}
This proves \eqref{eqgkn1} and \eqref{eqgkn2}.

It follows from \eqref{eqgkn1} that, for $n$ odd, the map $a\mapsto U a U^\dagger$ gives a Lie algebra isomorphism
\begin{equation*}
\g_7(n) = \su(2^n)^{\{P_X,P_Y\}} \to \su(2^n)^{\{X_1,Z_1\}} = I \otimes \su(2^{n-1}) \cong \su(2^{n-1}).
\end{equation*}
Now suppose that $n$ is even. Then, by \eqref{eqgkn2}, the map $a\mapsto U a U^\dagger$ gives an isomorphism
\begin{equation*}
\su(2^n)^{\{P_X,P_Y\}} \to \su(2^n)^{\{Z_1,Z_2\}} 
= \bigl( \Span_{\mathbb R}\{I,Z\} \otimes \Span_{\mathbb R}\{I,Z\} \otimes \su(2^{n-2}) \bigr) \oplus \Span\{Z_1,Z_2,Z_1Z_2\}.
\end{equation*}
After we quotient by the center $\Span\{Z_1,Z_2,Z_1Z_2\}$, we obtain
\begin{equation*}
\g_7(n) = \su(2^n)^{\{P_X,P_Y\}} / \Span\{P_X,P_Y,P_Z\} \cong \Span_{\mathbb R}\{I,Z\} \otimes \Span_{\mathbb R}\{I,Z\} \otimes \su(2^{n-2}).
\end{equation*}
Again as in the proof of Lemma \ref{lemgkn1}, let $P_\pm := (I\pm Z)/2$, and consider the four projections
\begin{equation*}
P_1:= P_+ \otimes P_+, \qquad P_2:= P_+ \otimes P_-, \qquad P_3:= P_- \otimes P_+, \qquad P_4:= P_- \otimes P_-,
\end{equation*}
which satisfy
\begin{equation*}
P_i \cdot P_i = P_i, \qquad P_i \cdot P_j = 0 \quad (i\ne j), \qquad \sum_{i=1}^4 P_i = I \otimes I.
\end{equation*}
Then the linear map
\begin{equation}\label{eqgkn3}
(a_1,a_2,a_3,a_4) \mapsto \sum_{j=1}^4 P_j \otimes a_j
\end{equation}
is an isomorphism from $\su(2^{n-2})^{\oplus 4}$ to $\Span_{\mathbb R}\{I,Z\} \otimes \Span_{\mathbb R}\{I,Z\} \otimes \su(2^{n-2})$.
\end{proof}

In the remaining two cases $k=3,5$, as before we embed $\aa_3(n)$ and $\aa_5(n)$ as subalgebras of $\aa_7(n)$. 
We continue to use the notation from Sect.\ \ref{seclow} and, as in Lemma \ref{lemgkn3}, we consider separately the cases when $n$ is odd or even.

\begin{lemma}\label{lemgkn4o}
We have\/
$\aa_3(n) \cong \tilde\aa_3(n) = \g_7(n)^{\tilde\theta_3} \cong 
\begin{cases} 
        \so(2^{n-1}), & n\equiv \pm1 \;\mathrm{mod}\; 8, \\
        \sp(2^{n-2}), & n\equiv \pm3 \;\mathrm{mod}\; 8.
\end{cases}$
\end{lemma}
\begin{proof}
We apply the transformation $a\mapsto U a U^\dagger$ from the proof of Lemma \ref{lemgkn3} that gives a Lie algebra isomorphism
$\g_7(n) \to I \otimes \su(2^{n-1}) \cong \su(2^{n-1})$, where $U$ is defined by \eqref{eqgkn} for odd $n$.
Then, by Lemmas \ref{leminv1}, \ref{leminv2},
the fixed points of $\tilde\theta_k$ (see \eqref{tildethetaq35}) are sent to the fixed points of the involution
\begin{equation}\label{eqbuau3}
g \mapsto -(U Q_k U^T) g^T (U Q_k U^T)^\dagger, \qquad k=3,5.
\end{equation}
Recall that $Q_3 = P_{ZIYX}$ is given by \eqref{thetaq3}, and compute
\begin{equation*}
\tilde Q_3 := U Q_3 U^T
= e^{i\frac{\pi}{4}Z\otimes Y^{\otimes(n-1)}} e^{i\frac{\pi}{4}Y\otimes X^{\otimes(n-1)}}
P_{ZIYX} e^{-i\frac{\pi}{4}Y\otimes X^{\otimes(n-1)}} e^{i\frac{\pi}{4}Z\otimes Y^{\otimes(n-1)}}.
\end{equation*}

Note that, when $n\equiv 1$ mod $4$, $P_{ZIYX}$ anticommutes with $Y\otimes X^{\otimes(n-1)}$.
By \eqref{equab}, we have:
\begin{align*}
e^{i\frac{\pi}{4}Y\otimes X^{\otimes(n-1)}} & P_{ZIYX} e^{-i\frac{\pi}{4}Y\otimes X^{\otimes(n-1)}} \\
&= i (Y_1 X_2 X_3 X_4 \cdots X_{n-1} X_n) \cdot (Z_1 Y_3 X_4 Z_5 Y_7 X_8 \cdots X_{n-1} Z_n) \\
&= -X_1 (X_2 Z_3 Y_5) (X_6 Z_7 Y_9) \cdots (X_{n-3} Z_{n-2} Y_n).
\end{align*}
As this anticommutes with $e^{i\frac{\pi}{4}Z\otimes Y^{\otimes(n-1)}}$, we obtain
\begin{equation*}
\tilde Q_3 = -X_1 (X_2 Z_3 Y_5) (X_6 Z_7 Y_9) \cdots (X_{n-3} Z_{n-2} Y_n),
\qquad n\equiv 1 \;\mathrm{mod}\; 4.
\end{equation*}
Hence, restricted to elements $g=I\otimes c$ with $c\in \su(2^{n-1})$, the involution \eqref{eqbuau3} becomes:
\begin{equation*}
c \mapsto -P_{XZIY} c^T P_{XZIY}, \qquad\text{for}\quad n\equiv 1 \;\mathrm{mod}\; 4.
\end{equation*}
For the fixed-point subalgebra, we obtain from Corollary \ref{corinv4}:
\begin{equation*}
(P_{XZIY})^T = \begin{cases} 
P_{XZIY}, & n\equiv 1 \;\mathrm{mod}\; 8, \\
-P_{XZIY}, & n\equiv 5 \;\mathrm{mod}\; 8 \\ 
\end{cases}
\quad\Rightarrow\quad
\tilde\aa_3(n) \cong 
\begin{cases} 
        \so(2^{n-1}), & n\equiv 1 \;\mathrm{mod}\; 8, \\
        \sp(2^{n-2}), & n\equiv 5 \;\mathrm{mod}\; 8.
\end{cases}
\end{equation*}

Alternatively, when $n\equiv 3$ mod $4$, $P_{ZIYX}$ commutes with both $Y\otimes X^{\otimes(n-1)}$ and $Z\otimes Y^{\otimes(n-1)}$. Hence, in this case,
\begin{align*}
\tilde Q_3 &= e^{i\frac{\pi}{4}Z\otimes Y^{\otimes(n-1)}} P_{ZIYX} e^{i\frac{\pi}{4}Z\otimes Y^{\otimes(n-1)}} \\
&= i (Z_1 Y_2 Y_3 Y_4 \cdots Y_{n-1} Y_n) \cdot (Z_1 Y_3 X_4 Z_5 Y_7 X_8 \cdots X_{n-3} Z_{n-2} Y_n) \\
&= i Y_2 (Z_4 X_5 Y_6) (Z_8 X_9 Y_{10}) \cdots (Z_{n-3} X_{n-2} Y_{n-1}).
\end{align*}
Thus, restricted to elements $g=I\otimes c$ with $c\in \su(2^{n-1})$, the involution \eqref{eqbuau3} simplifies to:
\begin{equation*}
c \mapsto -P_{YIZX} c^T P_{YIZX}, \qquad\text{for}\quad n\equiv 3 \;\mathrm{mod}\; 4.
\end{equation*}
Corollary \ref{corinv4} gives for the fixed-point subalgebra:
\begin{equation*}
(P_{YIZX})^T = \begin{cases} 
P_{YIZX}, & n\equiv 7 \;\mathrm{mod}\; 8, \\
-P_{YIZX}, & n\equiv 3 \;\mathrm{mod}\; 8 \\ 
\end{cases}
\quad\Rightarrow\quad
\tilde\aa_3(n) \cong 
\begin{cases} 
        \so(2^{n-1}), & n\equiv 7 \;\mathrm{mod}\; 8, \\
        \sp(2^{n-2}), & n\equiv 3 \;\mathrm{mod}\; 8.
\end{cases}
\end{equation*}
This completes the proof of the lemma.
\end{proof}

\begin{lemma}\label{lemgkn5o}
We have\/
$\aa_5(n) \cong \tilde\aa_5(n) = \g_7(n)^{\tilde\theta_5} \cong 
\begin{cases} 
        \so(2^{n-1}), & n\equiv \pm1 \;\mathrm{mod}\; 6, \\
        \sp(2^{n-2}), & n\equiv 3 \;\mathrm{mod}\; 6.
\end{cases}$
\end{lemma}
\begin{proof}
The proof is analogous to that of Lemma \ref{lemgkn4o}. Recall that $U$ and $Q_5 = P_{IYZ}$ are given by \eqref{eqgkn}, \eqref{thetaq5}, and compute
\begin{equation*}
\tilde Q_5 := U Q_5 U^T
= e^{i\frac{\pi}{4}Z\otimes Y^{\otimes(n-1)}} e^{i\frac{\pi}{4}Y\otimes X^{\otimes(n-1)}}
P_{IYZ} e^{-i\frac{\pi}{4}Y\otimes X^{\otimes(n-1)}} e^{i\frac{\pi}{4}Z\otimes Y^{\otimes(n-1)}}.
\end{equation*}

When $n\equiv 3$ mod $6$, $P_{IYZ}= Y_2 Z_3 Y_5 Z_6 Y_8 Z_9 \cdots Y_{n-1} Z_{n}$ commutes with $Y\otimes X^{\otimes(n-1)}$ and anticommutes with $Z\otimes Y^{\otimes(n-1)}$. 
Hence, by \eqref{equab1}, $\tilde Q_5 = P_{IYZ}$. Restricted to elements $g=I\otimes c$ with $c\in \su(2^{n-1})$, the involution \eqref{eqbuau3} simplifies to:
\begin{equation*}
c \mapsto -P_{YZI} c^T P_{YZI}, \qquad\text{for}\quad n\equiv 3 \;\mathrm{mod}\; 6.
\end{equation*}
Since $(P_{YZI})^T = -P_{YZI}$, the fixed-point subalgebra is isomorphic to $\sp(2^{n-2})$, by Corollary \ref{corinv4}.

For $n\equiv 1$ mod $6$, $P_{IYZ}= Y_2 Z_3 Y_5 Z_6 \cdots Y_{n-2} Z_{n-1}$ commutes with both $Y\otimes X^{\otimes(n-1)}$ and $Z\otimes Y^{\otimes(n-1)}$. 
Hence, in this case,
\begin{align*}
\tilde Q_5 &= i (Z_1 Y_2 Y_3 Y_4 \cdots Y_{n-1} Y_n) \cdot (Y_2 Z_3 Y_5 Z_6 \cdots Y_{n-2} Z_{n-1}) \\
&= i^{(n+2)/3} Z_1 (X_3 Y_4) (X_6 Y_7) \cdots (X_{n-1} Y_n).
\end{align*}
The involution induced by \eqref{eqbuau3} on $c\in \su(2^{n-1})$ is given by
\begin{equation*}
c \mapsto -P_{IXY} c^T P_{IXY}, \qquad\text{for}\quad n\equiv 1 \;\mathrm{mod}\; 6,
\end{equation*}
and the fixed-point subalgebra is isomorphic to $\so(2^{n-1})$, because $(P_{IXY})^T = P_{IXY}$.

Finally, for $n\equiv 5$ mod $6$, using \eqref{equab}, we find
\begin{align*}
e^{i\frac{\pi}{4}Y\otimes X^{\otimes(n-1)}} & P_{IYZ} e^{-i\frac{\pi}{4}Y\otimes X^{\otimes(n-1)}} \\
&= i (Y_1 X_2 X_3 X_4 X_5 \cdots X_n) (Y_2 Z_3 Y_5 \cdots Z_{n-2} Y_{n}) \\
&= -Y_1 Z_2 Y_3 X_4 Z_5 \cdots Y_{n-2} X_{n-1} Z_{n}.
\end{align*}
Then applying \eqref{equab1}, we get
\begin{align*}
\tilde Q_5 &= -i (Z_1 Y_2 Y_3 Y_4 Y_5 \cdots Y_n) \cdot (Y_1 Z_2 Y_3 X_4 Z_5 \cdots Y_{n-2} X_{n-1} Z_{n}) \\
&= -i X_1 X_2 (Z_4 X_5) \cdots (Z_{n-1} X_n).
\end{align*}
This induces the involution on $\su(2^{n-1})$ given by
\begin{equation*}
c \mapsto -P_{XIZ} c^T P_{XIZ}, \qquad\text{for}\quad n\equiv 5 \;\mathrm{mod}\; 6,
\end{equation*}
and the fixed-point subalgebra is isomorphic again to $\so(2^{n-1})$.
\end{proof}

\begin{lemma}\label{lemgkn4e}
We have\/
$\aa_3(n) \cong \tilde\aa_3(n) = \g_7(n)^{\tilde\theta_3} \cong 
\begin{cases} 
        \so(2^{n-2})^{\oplus 4}, & n\equiv 0 \;\mathrm{mod}\; 8, \\
        \su(2^{n-2})^{\oplus2}, & n\equiv \pm2 \;\mathrm{mod}\; 8, \\
        \sp(2^{n-3})^{\oplus4}, & n\equiv 4 \;\mathrm{mod}\; 8.    
\end{cases}$
\end{lemma}
\begin{proof}
As in the proof of Lemma \ref{lemgkn4o}, we need to compute
\begin{equation*}
\tilde Q_3 := U Q_3 U^T
= e^{i\frac{\pi}{4}X_2} e^{i\frac{\pi}{4} I\otimes X\otimes Z^{\otimes(n-2)}} e^{i\frac{\pi}{4}Y\otimes X^{\otimes(n-1)}}
P_{ZIYX} e^{-i\frac{\pi}{4}Y\otimes X^{\otimes(n-1)}} e^{i\frac{\pi}{4} I\otimes X\otimes Z^{\otimes(n-2)}} e^{i\frac{\pi}{4}X_2}.
\end{equation*}
Using \eqref{equab1}, \eqref{equab}, we find for $n\equiv 0$ mod $4$:
\begin{equation*}
\tilde Q_3 = -Z_1(X_3 Y_4 Z_6)(X_7 Y_8 Z_{10}) \cdots X_{n-1} Y_n.
\end{equation*}
Via the isomorphism $\g_7(n) \cong \su(2^{n-2})^{\oplus 4}$ from the proof of Lemma \ref{lemgkn3} (see \eqref{eqgkn3}),
the involution induced from \eqref{eqbuau3} on each copy of $\su(2^{n-2})$ is given by
\begin{equation*}
a_j \mapsto -P_{XYIZ} a_j^T P_{XYIZ}, \qquad\text{for}\quad 1\le j\le 4, \quad n\equiv 0 \;\mathrm{mod}\; 4.
\end{equation*}
For the fixed-point subalgebra, we get from Corollary \ref{corinv4}:
\begin{equation*}
(P_{XYIZ})^T = \begin{cases} 
P_{XYIZ}, & n\equiv 0 \;\mathrm{mod}\; 8, \\
-P_{XYIZ}, & n\equiv 4 \;\mathrm{mod}\; 8 \\ 
\end{cases}
\quad\Rightarrow\quad
\tilde\aa_3(n) \cong 
\begin{cases} 
        \so(2^{n-2})^{\oplus 4}, & n\equiv 0 \;\mathrm{mod}\; 8, \\
        \sp(2^{n-3})^{\oplus4}, & n\equiv 4 \;\mathrm{mod}\; 8.     
\end{cases}
\end{equation*}

Using \eqref{equab1}, \eqref{equab}, we find for $n\equiv 0$ mod $4$:
\begin{equation*}
\tilde Q_3 = X_1 X_2 (Z_4 X_5 Y_6)(Z_8 X_9 Y_{10}) \cdots (Z_{n-2} X_{n-1} Y_n).
\end{equation*}
Consider the unitary operator
\begin{equation*}
V = \begin{cases} 
e^{i\frac{\pi}{4} Z_4 X_5 Y_6 \cdots Z_{n-2} X_{n-1} Y_n}, & n\equiv 2 \;\mathrm{mod}\; 8, \\
e^{i\frac{\pi}{4} Z_2 Z_4 X_5 Y_6 \cdots Z_{n-2} X_{n-1} Y_n}, & n\equiv 6 \;\mathrm{mod}\; 8,
\end{cases}
\end{equation*}
and perform the transformation $a\mapsto V a V^\dagger$ on $\g_7(n)$. Since $V$ commutes with $Z_1$ and $Z_2$, this transformation preserves the decomposition
$\g_7(n) \cong \su(2^{n-2})^{\oplus 4}$ given by \eqref{eqgkn3}. For $n\equiv 2$ mod $8$, we have $V^T=V$ and $V$ commutes with $\tilde Q_3$. Thus, $\tilde Q_3$
gets transformed to
\begin{equation*}
V \tilde Q_3 V^T = V^2 \tilde Q_3 = i (Z_4 X_5 Y_6 \cdots Z_{n-2} X_{n-1} Y_n) \cdot \tilde Q_3 = i X_1 X_2.
\end{equation*}
The involution induced on $\g_7(n)$ is given by
\begin{equation*}
a \mapsto - X_1 X_2 a^T X_1 X_2.
\end{equation*}
Writing $a$ as in \eqref{eqgkn3}, we note that
\begin{equation*}
X_1 X_2 \cdot P_1 \cdot X_1 X_2 = P_4, \qquad X_1 X_2 \cdot P_2 \cdot X_1 X_2 = P_3.
\end{equation*}
Hence, $a=(a_1,a_2,a_3,a_4)$ is a fixed point of the above involution if and only if $a_4=-a_1^T$, $a_3=-a_2^T$. Sending such $a$ to $(a_1,a_2)$ gives an isomorphism
from the fixed-point subalgebra to $\su(2^{n-2})^{\oplus2}$.

When $n\equiv 6$ mod $8$, we have $V^T=V^{-1}$ and $V$ anticommutes with $\tilde Q_3$. Thus, $\tilde Q_3$
gets transformed to
\begin{equation*}
V \tilde Q_3 V^T = V \tilde Q_3 V^{-1} = V^2 \tilde Q_3 = i (Z_2 Z_4 X_5 Y_6 \cdots Z_{n-2} X_{n-1} Y_n) \cdot \tilde Q_3 = -X_1 Y_2.
\end{equation*}
The rest of the proof is similar to the case $n\equiv 2$ mod $8$ above.
\end{proof}

\begin{lemma}\label{lemgkn5e}
We have\/ $\aa_5(n) \cong \tilde\aa_5(n) = \g_7(n)^{\tilde\theta_5} \cong 
\begin{cases} 
        \so(2^{n-2})^{\oplus 4}, & n\equiv 0 \;\mathrm{mod}\; 6, \\
        \su(2^{n-2})^{\oplus2}, & n\equiv \pm2 \;\mathrm{mod}\; 6.
\end{cases}$
\end{lemma}
\begin{proof}
The proof is very similar to Lemma \ref{lemgkn4e}, so we only indicate the differences. We compute
\begin{equation*}
\tilde Q_5 := U Q_5 U^T
= e^{i\frac{\pi}{4}X_2} e^{i\frac{\pi}{4} I\otimes X\otimes Z^{\otimes(n-2)}} e^{i\frac{\pi}{4}Y\otimes X^{\otimes(n-1)}}
P_{IYZ} e^{-i\frac{\pi}{4}Y\otimes X^{\otimes(n-1)}} e^{i\frac{\pi}{4} I\otimes X\otimes Z^{\otimes(n-2)}} e^{i\frac{\pi}{4}X_2},
\end{equation*}
and find that
\begin{equation*}
\tilde Q_5 = \begin{cases} 
i^{-n/3} Z_2 (Z_4 X_5) (Z_7 X_8) \cdots X_{n-1}, & n\equiv 0 \;\mathrm{mod}\; 6, \\
-Y_1 Z_2 (Y_3 X_4 Z_5) (Y_6 X_7 Z_8) \cdots Z_n, & n\equiv 2 \;\mathrm{mod}\; 6, \\
P_{IYZ} = (Y_2 Z_3) (Y_5 Z_6) (Y_8 Z_9) \cdots Z_{n-1}, & n\equiv 4 \;\mathrm{mod}\; 6.
\end{cases}
\end{equation*}
For $n\equiv 0$ mod $6$, the involution induced from \eqref{eqbuau3} on each copy of $\su(2^{n-2})$ from the decomposition \eqref{eqgkn3} is given by
\begin{equation*}
a_j \mapsto -P_{IZX} a_j^T P_{IZX}, \qquad\text{for}\quad 1\le j\le 4, \quad n\equiv 0 \;\mathrm{mod}\; 6.
\end{equation*}
For $n\equiv \pm2$ mod $6$, we use the transformation $a\mapsto V a V^\dagger$, where
\begin{equation*}
V = \begin{cases} 
e^{i\frac{\pi}{4} Z_2 Y_3 X_4 Z_5 Y_6 X_7 Z_8 \cdots Z_n}, & n\equiv 2 \;\mathrm{mod}\; 6, \\
e^{i\frac{\pi}{4} Z_3 Y_5 Z_6 Y_8 Z_9 \cdots Z_{n-1}}, & n\equiv 4 \;\mathrm{mod}\; 6,
\end{cases}
\end{equation*}
which allows us to replace $\tilde Q_5$ with $V \tilde Q_5 V^T$. This gives the involutions
$a \mapsto - Y_1 a^T Y_1$ and $a \mapsto - Y_2 a^T Y_2$ for $n\equiv 2$ and $n\equiv 4$ mod $6$, respectively.
\end{proof}

\subsection{Periodic boundary conditions}\label{secper}

Recall that the subalgebras $\aa_k^\circ(n), \bb_l^\circ(n) \subseteq\su(2^n)$ ($0\le k\le 22$, $0\le l\le 4$) are defined by \eqref{ext5}.
In this subsection, we prove Theorem \ref{the:classification-p}, which we reproduce here for convenience:
\allowdisplaybreaks
\begin{align*}
\aa_0^\circ(n) &\cong \uu(1)^{\oplus n}, \\
\aa_1^\circ(n) &\cong \so(n)^{\oplus 2}, \\
\aa_2^\circ(n) &\cong \so(n)^{\oplus 4}, \\
\aa_3^\circ(n) &= \begin{cases} 
\aa_{13}(n), & n \;\;\mathrm{odd}, \\
\aa_3(n), & n\equiv 0 \mod 4, \\
\aa_6(n), & n\equiv 2 \mod 4,
\end{cases} \\
\aa_4^\circ(n) &\cong \begin{cases}
\so(2n), \quad\;\; n \;\;\mathrm{odd}, \\
\so(n)^{\oplus 4}, \quad n \;\mathrm{even},
\end{cases} \\
\aa_{5}^\circ(n) &= \begin{cases}
\aa_{16}(n), & n\equiv \pm1 \mod 3, \\ 
\aa_{5}(n), & n\equiv 0 \mod 3, 
\end{cases} \\
\aa_6^\circ(n) &= \begin{cases}
\aa_{13}(n), & n \;\;\mathrm{odd}, \\
\aa_6(n), & n \;\;\mathrm{even},
\end{cases} \\
\aa_k^\circ(n) &= \aa_k(n), \quad k=7,13,16,20, \\
\aa_8^\circ(n) &\cong \so(2n)^{\oplus 2}, \\
\aa_{9}^\circ(n) &= \bb_2^\circ(n) \cong \so(2^n), \quad n \ge 4, \\
\aa_{10}^\circ(n) &= \begin{cases}
\su(2^n), & n\equiv \pm1 \mod 3, \\ 
\aa_{10}(n), & n\equiv 0 \mod 3,
\end{cases} \\
\aa_{11}^\circ(n) &= \so(2^n), \quad n\ge 4, \\
\aa_k^\circ(n) &= \bb_4^\circ(n) = \su(2^n), \;\; k=12,15,17,18,19,21,22, \\
\aa_{14}^\circ(n) &\cong \so(2n)^{\oplus 2}, \\
\bb_0^\circ(n) &= \bb_0(n) \cong \uu(1)^{\oplus n}, \\
\bb_1^\circ(n) &\cong \uu(1)^{\oplus 2n}, \\
\bb_3^\circ(n) &= \bb_3(n) \cong \su(2)^{\oplus n}.
\end{align*}

We start the \textbf{proof} by observing that
due to \eqref{ext1}, \eqref{ext2} and from $\dim \aa_{12}^\circ(3) = \dim \aa_{17}^\circ(3) = 63$, we have:
\begin{equation*}
\aa_k^\circ(n) = \su(2^n), \qquad k=12,17,18,19,21,22, \;\; n\ge 3.
\end{equation*}
Moreover,
\begin{equation*}
\aa_{15}^\circ(n) = \su(2^n), \qquad n\ge 3,
\end{equation*}
because $\aa_{15}(n)$ contains ($i$ times) all Pauli strings that start with $X$ or $I$, except $I^{\otimes n}$.
Then, applying the cyclic shift $\tau_n$ defined in \eqref{taun}, we can generate all Pauli strings $\ne I^{\otimes n}$.

We also note that 
\begin{equation*}
\aa_k^\circ(n) = \aa_k(n), \qquad k=7,13,16,20, \;\; n\ge 3; \qquad
\aa_{11}^\circ(n) = \aa_{11}(n), \qquad n\ge 4,
\end{equation*}
due to \eqref{ext3}, \eqref{ext4} and Lemmas \ref{lem4}, \ref{lem6}, \ref{lem7}, because in this case $\tau_n \aa_k(n) \subseteq \aa_k(n)$.

In Sect.\ \ref{secfrus}, using frustration graphs, we determined the Lie algebras $\aa_k^\circ(n)$ for $k=1,2,4,8,14$ (see Lemmas \ref{lemfrus2}, \ref{lemfrus3}, \ref{lemfrus4}).
It is also obvious that
\begin{align*}
\bb_0^\circ(n) &= \bb_0(n),  \qquad\qquad \bb_2^\circ(n) = \aa_9^\circ(n), \\
\bb_3^\circ(n) &= \bb_3(n),  \qquad\qquad \bb_4^\circ(n) = \aa_{15}^\circ(n), \\
\bb_1^\circ(n) &= \Span\{X_i, X_1 X_n, X_j X_{j+1}\}_{1\le i\le n, \; 1\le j\le n-1} \cong \uu(1)^{\oplus 2n}, \\
\aa_0^\circ(n) &= \Span\{X_1 X_n, X_j X_{j+1}\}_{1\le j\le n-1} \cong \uu(1)^{\oplus n}.
\end{align*}

We discuss the remaining cases $\aa_k^\circ(n)$ ($k=3,5,6,9,10$) in a sequence of lemmas.

\begin{lemma}\label{lem8}
We have\/
$\aa_{10}^\circ(n) = \begin{cases}
\su(2^n), & n\equiv \pm1 \mod 3, \\ 
\aa_{10}(n), & n\equiv 0 \mod 3.
\end{cases}$
\end{lemma}
\begin{proof}
Recall from Theorem \ref{thmtheta} that $\aa_{10}(n) = \g_{10}(n)$ where $\g_{10}(n)$ is given by \eqref{listgkn2}.
When $n\equiv 0$ mod $3$, we have:
\begin{equation*}
\tau_n P_{XYZ} = P_{YZX}, \qquad \tau_n P_{YZX} = P_{ZXY}, \qquad \tau_n P_{ZXY} = P_{XYZ},
\end{equation*}
which imply that $\tau_n \aa_{10}(n) \subseteq \aa_{10}(n)$, and hence $\aa_{10}^\circ(n) = \aa_{10}(n)$.

On the other hand, for $n\equiv 1$ mod $3$, we have:
\begin{align*}
\tau_n^{-1} P_{XYZ} &= X \otimes P_{XYZ} = XXYZXYZ \cdots XYZ, \\ 
\tau_n^{-1} P_{YZX} &= Y \otimes P_{YZX} = YYZXYZX \cdots YZX, \\ 
\tau_n^{-1} P_{ZXY} &= Z \otimes P_{ZXY} = ZZXYZXY \cdots ZXY.
\end{align*}
In particular, their centralizer contains the elements
\begin{equation*}
X_1X_2,Y_1Y_2,Z_1Z_2 \in\tau_n^{-1} \aa_{10}(n) \subset \aa_{10}^\circ(n).
\end{equation*}
From these elements and 
\begin{equation*}
X_1Y_2, Y_1Z_2, Z_1X_2 \in \aa_{10}(n) \subset \aa_{10}^\circ(n),
\end{equation*}
we can generate all $2$-qubit gates: $\su(4) \otimes I^{\otimes (n-2)} \subset \aa_{10}^\circ(n)$. 
Therefore, $\aa_{10}^\circ(n) = \su(2^n)$. 

The case $n\equiv -1$ mod $3$ is similar.
\end{proof}

\begin{lemma}\label{lem9}
We have\/
$\aa_{5}^\circ(n) = \begin{cases}
\aa_{16}(n), & n\equiv \pm1 \mod 3, \\ 
\aa_{5}(n), & n\equiv 0 \mod 3.
\end{cases}$
\end{lemma}
\begin{proof}
Recall the automorphism $\gamma_n$ of $\su(2^n)$ defined by \eqref{rho1}, \eqref{rho2}.
Then, by Lemma \ref{lem4},
$\gamma_n \aa_5(n) = \g_7(n)^{\tilde\theta_5}$,
where $\tilde\theta_5$ is given by \eqref{tildethetaq35}, \eqref{thetaq5}, and $\g_7(n)$ is given by \eqref{listgkn3}.
From this, we get
\begin{equation*}
\gamma_n \tau_n \aa_5(n) = (\gamma_n \tau_n \gamma_n^{-1}) \gamma_n \aa_5(n) = (\gamma_n \tau_n \gamma_n^{-1}) \g_7(n)^{\tilde\theta_5}.
\end{equation*}

When $n\equiv 0$ mod $3$, we have
\begin{equation*}
(\gamma_n \tau_n \gamma_n^{-1})(P_X) = P_Y, \qquad
(\gamma_n \tau_n \gamma_n^{-1})(P_Y) = P_Z, \qquad
(\gamma_n \tau_n \gamma_n^{-1})(P_Z) = P_X, 
\end{equation*}
which imply that 
$(\gamma_n \tau_n \gamma_n^{-1}) \g_7(n) \subseteq \g_7(n)$. 
Next, we compute (cf.\ \eqref{thetaq5}):
\begin{align*}
(\gamma_n \tau_n \gamma_n^{-1}) Q_5 
&= (\gamma_n \tau_n \gamma_n^{-1}) (P_{IYZ}) \\ 
&= Z_1 X_2 Z_4 X_5 \cdots Z_{n-2} X_{n-1} \\
&= i^{n/3} P_Z \cdot Q_5 = i^{-n/3} Q_5 \cdot P_Z.
\end{align*}
From this, we deduce that $\tilde\theta_5$ commutes with $\gamma_n \tau_n \gamma_n^{-1}$. Indeed, as it commutes with the trace, we find
for $g \in \g_7(n)$:
\begin{equation*}
(\gamma_n \tau_n \gamma_n^{-1}) \tilde\theta_5(g) = -(\gamma_n \tau_n \gamma_n^{-1} Q_5) h^T (\gamma_n \tau_n \gamma_n^{-1} Q_5)
= - Q_5 \cdot P_Z h^T P_Z \cdot Q_5 = - Q h^T Q = \tilde\theta_5(h),
\end{equation*}
where we set $h:=(\gamma_n \tau_n \gamma_n^{-1})g$ and use that $h,h^T \in \g_7(n)$.
Therefore, $\gamma_n\tau_n \aa_{5}(n) \subseteq \gamma_n \aa_{5}(n)$, and hence $\aa_{5}^\circ(n) = \aa_{5}(n)$
for $n\equiv 0$ mod $3$.

Suppose now that $n\equiv 1$ mod $3$. Observe that $\aa_5^\circ(n) \subseteq \so(2^n) = \aa_{16}(n)$ for all $n\ge 2$, 
because all generators of $\aa_5^\circ(n)$ have an odd number of $Y$'s.
On the other hand, we have 
\begin{equation*}
X_2 X_n \in \g_7(n)^{\tilde\theta_5} = \gamma_n \aa_5(n) 
\;\Rightarrow\; \gamma_n^{-1}(X_2 X_n) = Z_2 Y_n \in \aa_5(n)
\;\Rightarrow\; \tau_n^{-1}(Z_2 Y_n) = Y_1 Z_3 \in \aa_5^\circ(n).
\end{equation*}
Since $Y_1 X_3 \in \aa_5(3)$ (see Sect.\ \ref{secsu8}),
we get that $Y_1 X_3 \in \aa_5^\circ(n)$.
Hence, $[Y_1 Z_3,Y_1 X_3] = 2iY_3 \in \aa_5^\circ(n)$, and cyclic shifts give $Y_1,Y_2\in \aa_5^\circ(n)$.
Together with $\aa_5=\Lie{XY, YZ}$, the elements $YI,IY$ can generate $\aa_{16} = \Lie{XY, YX, YZ, ZY}$.
Therefore, $\aa_5^\circ(n) \supseteq \aa_{16}(n)$, which proves that $\aa_5^\circ(n) = \aa_{16}(n)$.

Similarly, in the case $n\equiv -1$ mod $3$, we have:
\begin{align*}
X_1 X_n \in \g_7(n)^{\tilde\theta_5} = \gamma_n \aa_5(n) 
\;&\Rightarrow\; \gamma_n^{-1}(X_1 X_n) = Y_1 Z_n \in \aa_5(n)
\;\Rightarrow\; \tau_n^{-1}(Y_1 Z_n) = Z_1 Y_2 \in \aa_5^\circ(n), \\
Z_1 Z_n \in \g_7(n)^{\tilde\theta_5} = \gamma_n \aa_5(n) 
\;&\Rightarrow\; \gamma_n^{-1}(Z_1 Z_n) = X_1 Y_n \in \aa_5(n)
\;\Rightarrow\; \tau_n^{-1}(X_1 Y_n) = Y_1 X_2 \in \aa_5^\circ(n).
\end{align*}
Hence, $\aa_5^\circ(n)$ contains $\aa_{16}(n)$, so it must be equal to it.
\end{proof}

\begin{lemma}\label{lem9a}
We have\/
$\aa_{9}^\circ(n) \cong \so(2^n)$ for $n \ge 4$.
\end{lemma}
\begin{proof}
First, recall from Lemma \ref{lem5} that $\aa_9(n) = \g_9(n)^{\theta_9}$, where $\g_9(n)$ is given by \eqref{listgkn4},
$\theta_9(g) = -Q_9 g^T Q_9$, and $Q_9=IYZ \cdots Z$ is given by \eqref{thetaq9}.
For example, $g=Y_3 X_4 \in\aa_9(n)$, as $g^T=-g$ and $g$ commutes with $X_1$, $Y_1X_2$, $Z_1X_2$ and $Q_9$
(or one can check directly that $IIYX\in\aa_9(4)$). Similarly, we check that $Z_3 X_4 \in\aa_9(n)$.

Now let us relabel $X \rightleftharpoons Y$, so that $\aa_9 = \Lie{YX, YZ}$.
Then $\aa_9^\circ(n) \subseteq \so(2^n) = \aa_{16}(n)$, because all generators of $\aa_9^\circ(n)$ contain an odd number of $Y$'s. 
From above after relabeling, we have $X_3 Y_4, Z_3 Y_4 \in\aa_9(n)$, which after a cyclic shift gives $X_1Y_2,Z_1Y_2\in \aa_9^\circ(n)$.
Since $\aa_{16} = \Lie{XY, YX, YZ, ZY}$, we obtain that $\aa_9^\circ(n) \supseteq \aa_{16}(n)$.
\end{proof}

\begin{lemma}\label{lem10}
We have:
\begin{align*}
\aa_6^\circ(n) &= \aa_6(n), \quad n \;\;\mathrm{even}, &
\aa_3^\circ(n) &= \aa_6^\circ(n) = \aa_{13}(n), \quad n \;\;\mathrm{odd}, \\
\aa_3^\circ(n) &= \aa_6(n), \quad n\equiv 2 \mod 4, &
\aa_3^\circ(n) &= \aa_3(n), \quad n\equiv 0 \mod 4.
\end{align*}
\end{lemma}
\begin{proof}
First of all, note that $\aa_3^\circ(n) \subseteq \aa_6^\circ(n)$ for all $n$, because $\aa_3\subset\aa_6$.
By Lemma \ref{lem4a}, we have for even $n$:
\begin{equation*}
\aa_6(n) = \g_6(n) = \su(2^n)^{\{P_X,P_{YZ},P_{ZY}\}} / \Span\{P_X,P_{YZ},P_{ZY}\}.
\end{equation*}
In this case,
\begin{equation*}
\tau_n(P_{X}) = P_{X}, \qquad \tau_n(P_{YZ}) = P_{ZY}, \qquad \tau_n(P_{ZY}) = P_{YZ},
\end{equation*}
which implies that $\tau_n \aa_6(n) \subseteq \aa_6(n)$, and hence $\aa_6^\circ(n) = \aa_6(n)$.

Recall the automorphism $\varphi_n$ of $\su(2^n)$ that up to a sign swaps $Y$ and $Z$ on all even qubits; see \eqref{varphi1}, \eqref{varphi2}.
By Lemma \ref{lem4}, we have $\varphi_n \aa_3(n) = \g_7(n)^{\tilde\theta_3}$,
where $\tilde\theta_3$ is given by \eqref{tildethetaq35}, \eqref{thetaq3}, and $\g_7(n)$ is given by \eqref{listgkn3}.
Hence,
\begin{equation*}
\varphi_n \tau_n \aa_3(n) = (\varphi_n \tau_n \varphi_n^{-1}) \varphi_n \aa_3(n) = (\varphi_n \tau_n \varphi_n^{-1}) \g_7(n)^{\tilde\theta_3}.
\end{equation*}
When $n$ is even, we have:
\begin{equation*}
(\varphi_n \tau_n \varphi_n^{-1})(P_X) = P_X, \qquad
(\varphi_n \tau_n \varphi_n^{-1})(P_Y) = P_Z, \qquad
(\varphi_n \tau_n \varphi_n^{-1})(P_Z) = P_Y, 
\end{equation*}
which implies that
$(\varphi_n \tau_n \varphi_n^{-1}) \g_7(n) \subseteq \g_7(n)$.
For $n\equiv 0$ mod $4$, we find
\begin{align*}
(\varphi_n \tau_n \varphi_n^{-1}) Q_3 
= (\varphi_n \tau_n \varphi_n^{-1})(P_{ZIYX})
= P_{IZXY} = P_Z \cdot Q_3 = Q_3 \cdot P_Z.
\end{align*}
Then, as in the proof of Lemma \ref{lem9}, we conclude that in this case $\aa_3^\circ(n) = \aa_3(n)$.

Next, in the case $n\equiv 2$ mod $4$, one checks that
\begin{equation*}
Y_1 Y_n \in \g_7(n)^{\tilde\theta_3} = \varphi_n \aa_3(n) 
\;\Rightarrow\; \varphi_n^{-1}(Y_1 Y_n) = Y_1 Z_n \in \aa_3(n)
\;\Rightarrow\; \tau_n^{-1}(Y_1 Z_n) = Z_1 Y_2 \in \aa_3^\circ(n).
\end{equation*}
Hence, $\aa_3^\circ(n)$ contains all generators of $\aa_6(n)$, proving that $\aa_3^\circ(n) = \aa_6(n)$.

Finally, consider the case when $n$ is odd. Recall that, by Lemma \ref{lem6},
\begin{equation*}
\aa_3^\circ(n) \subseteq \aa_6^\circ(n) \subseteq \aa_{13}(n) = \su(2^n)^{P_X} / \Span\{P_X\}.
\end{equation*}
In order to prove that these are equalities, it is enough to show that $\aa_3^\circ(n)$ contains the generators of $\aa_{13}(n)$.
When $n\equiv 1$ mod $4$, we have
\begin{equation*}
Z_1 Z_{n-1} \in \g_7(n)^{\tilde\theta_3} = \varphi_n \aa_3(n) 
\;\Rightarrow\; \varphi_n^{-1}(Z_1 Z_{n-1}) = -Z_1 Y_{n-1} \in \aa_3(n)
\;\Rightarrow\; \tau_n^{-2}(Z_1 Y_{n-1}) = Y_1 Z_3 \in \aa_3^\circ(n).
\end{equation*}
Since $ZIZ \in \aa_3(3)$, we get that $Z_1 Z_3 \in \aa_3^\circ(n)$ and hence $X_1 = -\frac{i}{2} [Y_1 Z_3, Z_1 Z_3] \in \aa_3^\circ(n)$.
Similarly, when $n\equiv 3$ mod $4$, we have
\begin{equation*}
Z_1 Z_n \in \g_7(n)^{\tilde\theta_3} = \varphi_n \aa_3(n) 
\;\Rightarrow\; \varphi_n^{-1}(Z_1 Z_n) = Z_1 Z_n \in \aa_3(n)
\;\Rightarrow\; \tau_n^{-1}(Z_1 Z_n) = Z_1 Z_2 \in \aa_3^\circ(n).
\end{equation*}
Then from $Y_1 Z_2 \in \aa_3^\circ(n)$, we get again that $X_1\in \aa_3^\circ(n)$.
Therefore, all $X_i\in \aa_3^\circ(n)$, and we can generate $\aa_{13}(n)$ from them and the generators
$X_i X_{i+1}$, $Y_i Z_{i+1}$ of $\aa_3(n)$.
\end{proof}

The above lemmas complete the proof of Theorem \ref{the:classification-p}.

\subsection{Permutation-invariant subalgebras}\label{app:sym}

In this subsection, we classify all permutation-invariant subalgebras of $\su(2^n)$ that are generated by single Paulis and products of two Paulis, 
thus proving Theorem \ref{the:classification-s}.
Recall that, starting from a subalgebra $\aa\subseteq\su(4)$, we generate the subalgebra $\aa^\pi(n)\subseteq\su(2^n)$, given by \eqref{ext7}.
Moreover, in Sect.\ \ref{secext}, we explained that $\aa$ can be assumed itself invariant under the flip of the two qubits;
so we only need to consider $\aa_k^\pi(n)$ for $k=0,2,4,6,7,14,16,20$ and $\bb_l^\pi(n)$ for $l=0,1,3$.
The complete list of such Lie algebras is presented in Theorem \ref{the:classification-s} and reproduced here as follows:
\begin{align*}
\aa_k^\pi(n) &= \aa_k(n), \qquad k=7,16,20,22, \\
\aa_0^\pi(n) &\cong \uu(1)^{\oplus n(n-1)/2}, \\
\aa_2^\pi(n) &= \so(2^n)^{P_Z} \cong \so(2^{n-1})^{\oplus2}, \\
\aa_4^\pi(n) &= \aa_7(n), \\
\aa_{14}^\pi(n) &\cong \aa_6^\pi(n) = \aa_{20}(n), \\
\bb_l^\pi(n) &= \bb_l(n), \qquad l=0,3, \\
\bb_1^\pi(n) &\cong \uu(1)^{\oplus n(n+1)/2}.
\end{align*}

To start the \textbf{proof} of the theorem, we first
observe that the following subalgebras of $\su(2^n)$ are permutation invariant, due to their explicit descriptions (cf.\ Theorem \ref{thmtheta}):
\begin{align*}
\aa_7(n) &= 
\begin{cases} 
\su(2^n)^{\{P_X,P_Y,P_Z\}} / \Span\{P_X,P_Y,P_Z\}, & n \;\;\mathrm{even}, \\
\su(2^n)^{\{P_X,P_Y,P_Z\}}, & n \;\;\mathrm{odd},
\end{cases} \\
\aa_{16}(n) &= \so(2^n), \\
\aa_{20}(n) &= \su(2^n)^{P_X} / \Span\{P_X\}, \\
\aa_{22}(n) &= \su(2^n), \\
\bb_0(n) &= \Span\{X_i\}_{1\le i\le n}, \\
\bb_3(n) &= \Span\{X_i, Y_i, Z_i\}_{1\le i\le n}.
\end{align*}
It is also easy to see that
\begin{align*}
\aa_0^\pi(n) &= \Span\{X_i X_j\}_{1\le i<j\le n}, \\
\bb_1^\pi(n) &= \Span\{X_k, X_i X_j\}_{1\le i<j\le n, \; 1\le k\le n}.
\end{align*}
Thus, we are left to determine $\aa_k^\pi(n)$ for $k=2,4,6,14$.
These cases are treated in the next three lemmas.

\begin{lemma}\label{lem11}
We have\/
$\aa_2^\pi(n) = \so(2^n)^{P_Z}$ for all $n\ge2$.
\end{lemma}
\begin{proof}
Note that all generators $X_i Y_j$ ($i\ne j$) of $\aa_2^\pi(n)$ commute with $P_Z$ and are skew-symmetric, i.e., satisfy $a^T=-a$.
Hence, $\aa_2^\pi(n) \subseteq \so(2^n)^{P_Z}$. For the opposite inclusion, we use the same strategy as in the proof of Lemma \ref{lem6}.
Pick an arbitrary Pauli string $a \in \so(2^n)^{P_Z}$ not containing any $I$'s; then we want to find a Pauli string $b\in \aa_2^\pi(n)$ such that $[a,b]\ne0$ and $[a,b]$ has an $I$ in some position. Note that $a$ has an odd number of $X$'s and an odd number of $Y$'s. In particular, after a permutation, $a$ must start with $XYZ$, $XXY$, or $XYY$. Then we let
$b=X_1 Y_3$, $X_1 Z_2 Y_3$, or $X_1 Z_2 Y_3$, respectively. Here $b\in \aa_2^\pi(n)$ because $XZY \in \aa_2(3)$; cf.\ Sect.\ \ref{secsu8}.
\end{proof}

\begin{lemma}\label{lem12}
We have\/
$\aa_{14}^\pi(n) \cong \aa_6^\pi(n) = \aa_{20}(n)$
for $n\ge3$.
\end{lemma}
\begin{proof}
Let us relabel $X \rightleftharpoons Z$ in $\aa_{14}$. Then $\aa_{14}\subset\aa_{20}$, which implies $\aa_{14}^\pi(n) \subseteq \aa_{20}^\pi(n) = \aa_{20}(n)$ for all $n$.
Similarly, from $\aa_6\subset\aa_{20}$, we get $\aa_6^\pi(n) \subseteq \aa_{20}(n)$. To finish the proof of the lemma, it is enough to show that 
$\aa_6^\pi(3) = \aa_{14}^\pi(3) = \aa_{20}(3)$, because $\aa_{20}(n)$ is generated from $\aa_{20}(3)$ using a process similar to \eqref{ext0}.
The claim now follows from $\aa_k^\circ(3) \subseteq \aa_k^\pi(3)$ and
\begin{equation*}
\dim \aa_6^\circ(3) = \dim \aa_{14}^\circ(3) = \dim \aa_{20}(3) = 30;
\end{equation*}
see Sect.\ \ref{secsu8}.
\end{proof}

\begin{lemma}\label{lem13}
We have\/ $\aa_4^\pi(n) = \aa_7(n)$ for $n\ge3$.
\end{lemma}
\begin{proof}
Since $\aa_4\subset\aa_7$, we have $\aa_4^\pi(n) \subseteq \aa_7^\pi(n) = \aa_7(n)$ for all $n\ge3$. To prove the opposite inclusion,
it is enough to show that $\aa_4^\pi(3) = \aa_7(3)$, because $\aa_7(n)$ is generated from $\aa_7(3)$ using a process similar to \eqref{ext0}.
From $IXX,YZX \in \aa_4(3)$, we get $ZYX \in \aa_4^\pi(3)$ and $[IXX,ZYX] = 2i ZZI \in \aa_4^\pi(3)$. 
Then, by permutation invariance, also $IZZ\in \aa_4^\pi(3)$. Hence, $\aa_4^\pi(3)$ contains all generators of $\aa_7(3)$,
so it must be equal to it.
\end{proof}

The only thing left to finish the proof of Theorem \ref{the:classification-s} is to show that $\so(2^n)^{P_Z} \cong \so(2^{n-1})^{\oplus2}$.
This follows from the isomorphism $\su(2^n)^{P_Z}/\Span\{P_Z\} \cong \su(2^{n-1})^{\oplus2}$ (see Lemma \ref{lemgkn1}),
which is compatible with taking matrix transpose.

\end{document}